\documentclass[12pt]{article}

\usepackage[margin=1in]{geometry}
\usepackage{amsmath,amsthm,amssymb,amsfonts}
\usepackage{mathtools}
\usepackage{hyperref}
\usepackage{url}
\usepackage[numbers,sort&compress]{natbib}
\usepackage{enumitem}
\usepackage{silence}
\usepackage[expansion=false]{microtype}
\usepackage{graphicx}

\usepackage{tikz}
\usetikzlibrary{automata, positioning, arrows.meta, decorations.pathmorphing, calc, decorations.markings, shadows}

\tikzset{
	state/.style={
		circle, 
		draw=black, 
		thick, 
		minimum size=1.3cm, 
		inner sep=0pt, 
		font=\small\bfseries,
		fill=white,
		drop shadow={opacity=0.1}
	},
	final/.style={
		double, 
		draw=black, 
		thick, 
		minimum size=1.3cm,
		fill=white
	},
	transition/.style={
		->, 
		>={Stealth[length=2.5mm]}, 
		thick, 
		draw=black!85
	},
	edge label/.style={
		midway,
		above,
		font=\footnotesize,
		color=black!90,
		inner sep=2pt
	}
}

\theoremstyle{plain}
\newtheorem{theorem}{Theorem}[section]
\newtheorem{proposition}[theorem]{Proposition}
\newtheorem{lemma}[theorem]{Lemma}
\newtheorem{corollary}[theorem]{Corollary}

\theoremstyle{definition}
\newtheorem{definition}[theorem]{Definition}
\newtheorem{example}[theorem]{Example}

\theoremstyle{remark}
\newtheorem{remark}[theorem]{Remark}

\newcommand{\R}{\mathbb{R}}
\newcommand{\C}{\mathbb{C}}

\newcommand{\doi}[1]{\textsc{doi:}~\texttt{#1}}

\DeclareMathOperator{\rank}{rank}

\DeclareMathOperator{\im}{im}
\DeclareMathOperator{\Pf}{Pf}
\DeclareMathOperator{\spn}{span}
\newcommand{\dd}{\mathrm{d}}
\newcommand{\pd}[1]{\partial_{#1}}
\newcommand{\Om}{\Omega}
\newcommand{\Mm}{\mathcal{M}}
\newcommand{\Ical}{\mathcal{I}}
\newcommand{\Rring}{\mathcal{R}}
\newcommand{\gen}[1]{\langle #1\rangle}
\newcommand{\FJ}{\mathrm{FJ}}

\title{From Null Modes to Gauge Generators:\\
	A Bordered Faddeev--Jackiw Theory of Constraint Chains}

\author{%
	E.~Chan-L\'opez\thanks{%
		E-mail: \texttt{eduardo.clopez13@gmail.com}.\;
		\textsc{orcid}:~0009-0003-7712-7907.}}

\date{%
	Divisi\'on Acad\'emica de Ciencias B\'asicas,
	Universidad Ju\'arez Aut\'onoma de Tabasco, 86690 Cunduac\'an,
	Tabasco, Mexico}

\begin{document}
	\maketitle
	
	\begin{abstract}
		In the bordered formulation of the Faddeev--Jackiw algorithm, the residual
		kernel of the extended symplectic matrix is routinely read as the space of
		gauge generators. We show that this reading is not correct, and we determine
		exactly what the kernel does encode. Writing
		$f^{(m)}=\bigl(\begin{smallmatrix}f^{(0)}&B\\-B^{\top}&0\end{smallmatrix}\bigr)$
		with $\Gamma=N^{\top}B$ the kernel pairing and $M=B^{\top}f^{+}B$ the
		Faddeev--Jackiw bracket matrix of the generated constraints relative to a choice
		of second-class splitting, we prove that
		$\dim\ker f^{(m)}=\dim\ker\Gamma^{\top}+\dim\ker\Gamma-\rank(\Pi M\Pi)$, and we
		identify the two summands: the first counts null directions of $f^{(0)}$
		tangent to the constraint surface, the second counts independent first-class
		combinations. The distinction matters because $\Gamma$ is not the Dirac
		constraint matrix, as is sometimes asserted: $\Gamma$ is symmetric whenever the
		constraints come from the consistency condition and the kernel frame is
		constant, whereas a bracket matrix is antisymmetric. The correct statement is
		that $\Gamma$ is the primary--secondary block and $M$ the secondary--secondary
		block of the reduced Dirac matrix $C$, so that
		$\Pf(C)=(-1)^{d(d+1)/2}\det\Gamma$ independently of $M$; this explains the
		determinantal factorization of the bordered matrix without the erroneous
		identification. Only the compression $\Pi M\Pi$ to $\ker\Gamma$, and not $M$
		itself, is independent of the splitting. We then show that a null mode with nonvanishing multiplier block
		is not a gauge generator but the first link of a two-step chain
		$\delta\xi=\rho\,u_{0}+\dot\rho\,u_{1}$, and that the chain closes if and only if
		the class $[u_{0}\cdot\nabla V]$ vanishes in the effective ring modulo the
		constraint ideal. A four-variable counterexample exhibits a system in which the
		algorithm halts on weakly vanishing contractions, exposes two null modes, and
		nevertheless possesses no gauge freedom whatsoever; the chain criterion rejects
		it correctly. We close with the hypotheses under which each statement holds and
		with the mechanical realizations that exhibit the criterion at work.
	\end{abstract}
	
	
	
	
	\section{Introduction}
	\label{sec:intro}
	
	The Faddeev--Jackiw formalism \cite{FJ1988} replaces the Dirac--Bergmann
	classification of constraints by the analysis of a single antisymmetric matrix,
	the presymplectic form of a first-order Lagrangian. In the iterative version of
	Barcelos-Neto and Wotzasek \cite{BNW1992,BNW1993} the degeneracy of that matrix
	generates constraints, and the constraints are appended to the system by
	bordering: one adjoins a block of constraint gradients and asks whether the
	enlarged matrix is invertible. In a companion paper \cite{CLMCP2026} we showed
	that this step is an instance of matrix bordering in the linear-algebraic sense,
	and that the regularity of the bordered matrix is governed by a reduced pairing
	\begin{equation}
		\Gamma=N^{\top}B,
		\qquad
		\Gamma_{a\alpha}=N_{a}^{\,i}\,\pd{i}\Om_{\alpha},
		\label{eq:intro-Gamma}
	\end{equation}
	between the constraint gradients $B_{i\alpha}=\pd{i}\Om_{\alpha}$ and a basis
	$N$ of the kernel of the presymplectic form, through the factorization
	$\det f^{(m)}=(\prod_{\ell}\mu_{\ell}^{2})\det(\Gamma)^{2}$. The matrix $\Gamma$
	and that factorization are the results of \cite{CLMCP2026} on which everything
	below rests, and we retain its notation throughout.
	
	That result settles the \emph{regular} case: when $\Gamma$ is square and
	nonsingular the iteration terminates in a genuine symplectic form. The present
	paper concerns the complementary situation, which is the physically interesting
	one. When $\Gamma$ is singular the bordered matrix retains a kernel, and this
	residual kernel is customarily described as the space of gauge generators. Our
	purpose is to show that this description is false as stated, to determine
	exactly what the residual kernel encodes, and to supply the additional condition
	that separates genuine gauge symmetries from the rest.
	
	Three findings organize the paper.
	
	First, the dimension of the residual kernel admits a canonical decomposition
	into two contributions with different meanings.
	Theorem~\ref{thm:kernel} parametrizes $\ker f^{(m)}$ completely, and
	Corollary~\ref{cor:dim} gives
	\begin{equation}
		\dim\ker f^{(m)}
		=\dim\ker\Gamma^{\top}+\dim\ker\Gamma-\rank(\Pi M\Pi),
		\label{eq:intro-dim}
	\end{equation}
	where $\Pi$ is the orthogonal projector onto $\ker\Gamma$ and
	$M=B^{\top}f^{+}B$ is built from the Moore--Penrose-type inverse of $f^{(0)}$ on
	its image. The first summand counts null directions of $f^{(0)}$ tangent to the
	constraint surface; the second counts independent first-class combinations of
	the generated constraints. Only the second produces gauge candidates. The
	correction term $\rank(\Pi M\Pi)$ has no counterpart in the informal slogan
	``number of generators $=$ number of first-class constraints''.
	
	Second, the matrices $\Gamma$ and $M$ have been conflated in the literature that
	motivated this work, our own companion paper included. When the constraints are generated by
	the consistency condition and the kernel frame is locally constant, $\Gamma$ is
	the contraction of the Hessian of the potential with the kernel frame, hence
	\emph{symmetric}; a matrix of Poisson brackets is \emph{antisymmetric}. The two
	can therefore never coincide unless both vanish. Theorem~\ref{thm:dirac}
	establishes the correct correspondence, in the cotangent extension where the
	primary constraints live: $M$ is the bracket matrix of the generated constraints
	in the Dirac bracket of the second-class primary pairs, $\Gamma$ is the
	primary--secondary block of the reduced constraint matrix $C$, and
	\begin{equation}
		C=\begin{pmatrix}0&-\Gamma\\ \Gamma^{\top}&M\end{pmatrix},
		\qquad
		\Pf(C)=(-1)^{d(d+1)/2}\det\Gamma,
		\qquad
		\det C=\det(\Gamma)^{2},
		\label{eq:intro-C}
	\end{equation}
	independently of $M$. The determinantal factorization of \cite{CLMCP2026} is
	thereby not weakened but explained: the bordered determinant is the determinant
	of the reduced Dirac matrix times the positive factor
	$\prod_{\ell}\mu_{\ell}^{2}$. What bordering does \emph{not} compute is the
	block $M$, and that is exactly the block the first-class analysis needs.
	
	Third, and this is the thesis of the title, a null mode with nonvanishing
	multiplier block is not a gauge generator. It is the first link of a chain
	\begin{equation}
		\delta\xi=\rho(t)\,u_{0}(\xi)+\dot\rho(t)\,u_{1}(\xi),
		\qquad u_{1}\in\ker f^{(0)},
		\label{eq:intro-chain}
	\end{equation}
	and the chain closes into an exact symmetry of the action if and only if a
	second, independent equation holds: the class $[u_{0}\cdot\nabla V]$ must vanish
	in the effective ring modulo the constraint ideal
	(Theorem~\ref{thm:chain-decisive}). Membership of the contraction in the
	constraint ideal, which is the criterion on which the iterative algorithm halts,
	is necessary but not sufficient. Section~\ref{sec:weak} exhibits a four-variable
	system in which the algorithm halts with two weakly vanishing contractions and
	reports two null modes, while the Euler--Lagrange equations have a unique
	solution: there is no gauge freedom at all, both constraints are second class,
	and the chain criterion rejects the two candidates correctly.
	
	The logical spine of the paper is therefore
	\[
	\text{null modes}\longrightarrow\text{constraints}
	\longrightarrow\text{first class}\longrightarrow\text{chains}
	\longrightarrow\text{gauge generators},
	\]
	with an explicit obstruction at each arrow. Sections~\ref{sec:conventions}
	to~\ref{sec:strong} fix the framework and prove the master variational identity
	on which everything rests. Sections~\ref{sec:kernel} to~\ref{sec:firstclass}
	develop the exact algebra of the bordered kernel. Sections~\ref{sec:chains}
	and~\ref{sec:weak} contain the main results. Sections~\ref{sec:involutivity}
	to~\ref{sec:charges} treat the geometric qualifications: involutivity, ambient
	versus restricted kernels, degeneracy strata, basis dependence, reducibility,
	real inconsistency, canonical charges and degree-of-freedom counting.
	Section~\ref{sec:examples} works out the algebraic examples in full and
	Section~\ref{sec:mechanical} the mechanical ones, and Section~\ref{sec:scope}
	lists the hypotheses.
	
	\begin{remark}[Relation to the companion paper]
		\label{rem:companion}
		Reference \cite{CLMCP2026} answers the question \emph{when does bordering
			restore nondegeneracy}, and introduces the reduced pairing
		\eqref{eq:intro-Gamma} that answers it. The present paper answers \emph{what
			happens when it does not}. The two are independent as texts: nothing below
		presupposes the reader has \cite{CLMCP2026} at hand, and the one result imported
		from it, the determinantal factorization, is restated and reproved here as
		Proposition~\ref{prop:factorization} in the form required. One statement of
		\cite{CLMCP2026} is corrected rather than extended, namely the identification of
		$\Gamma$ with the constraint bracket matrix; the correction is the subject of
		Section~\ref{sec:dirac} and it leaves the factorization intact.
	\end{remark}
	
	\section{Bordered Faddeev--Jackiw structure and conventions}
	\label{sec:conventions}
	
	\subsection{The first-order Lagrangian and its presymplectic form}
	
	Throughout, $\Mm$ is an open subset of $\R^{n}$ with coordinates
	$\xi=(\xi^{1},\dots,\xi^{n})$, and
	\begin{equation}
		L=a_{A}(\xi)\,\dot\xi^{A}-V(\xi),
		\qquad A=1,\dots,n,
		\label{eq:lag}
	\end{equation}
	is a first-order Lagrangian: the kinetic term is linear in the velocities and
	$V$ is the symplectic potential. The presymplectic form and its matrix are
	\begin{equation}
		f=\dd a,\qquad
		f_{AB}=\pd{A}a_{B}-\pd{B}a_{A},
		\label{eq:f}
	\end{equation}
	so $f$ is closed by construction. We write $f^{(0)}$ for \eqref{eq:f} before any
	bordering.
	
	\begin{lemma}[Equations of motion]
		\label{lem:eom}
		The Euler--Lagrange equations of \eqref{eq:lag} are
		\begin{equation}
			f_{AB}(\xi)\,\dot\xi^{B}=\pd{A}V(\xi).
			\label{eq:eom}
		\end{equation}
	\end{lemma}
	
	\begin{proof}
		$\partial L/\partial\dot\xi^{A}=a_{A}$, so
		$\tfrac{\dd}{\dd t}(a_{A})=(\pd{B}a_{A})\dot\xi^{B}$, while
		$\partial L/\partial\xi^{A}=(\pd{A}a_{B})\dot\xi^{B}-\pd{A}V$. Subtracting,
		$(\pd{A}a_{B}-\pd{B}a_{A})\dot\xi^{B}=\pd{A}V$.
	\end{proof}
	
	\begin{corollary}[Consistency condition]
		\label{cor:consistency}
		If $v\in\ker f(\xi)$ then $\Phi_{v}(\xi):=v^{A}\pd{A}V(\xi)=0$ is a necessary
		condition for \eqref{eq:eom} to admit a solution at $\xi$.
	\end{corollary}
	
	The zero locus of the functions $\Phi_{v}$, as $v$ runs over a frame of
	$\ker f$, produces the first generation of constraints. In the
	Barcelos-Neto--Wotzasek iteration these are appended by bordering.
	
	\subsection{Bordering}
	
	\begin{definition}[Bordering round]
		\label{def:bordering}
		Let $\Om_{1},\dots,\Om_{k}$ be the constraints accumulated so far and
		$B_{i\alpha}=\pd{i}\Om_{\alpha}$ the matrix of their gradients. The bordered
		matrix is
		\begin{equation}
			f^{(m)}=\begin{pmatrix}f^{(0)}&B\\-B^{\top}&0\end{pmatrix}
			\in\R^{(n+k)\times(n+k)} .
			\label{eq:fm}
		\end{equation}
	\end{definition}
	
	\begin{remark}[Sign convention]
		\label{rem:sign}
		Building the extended one-form as
		$a^{(m)}=(a_{i},-\Om_{\alpha})$ over the variables $(\xi^{i},\lambda^{\alpha})$
		yields $f^{(m)}_{i\alpha}=-B_{i\alpha}$, that is, the opposite convention
		$\bigl(\begin{smallmatrix}f^{(0)}&-B\\ B^{\top}&0\end{smallmatrix}\bigr)$. The
		two are related by $\lambda\mapsto-\lambda$, hence by $w\mapsto-w$ on the
		multiplier block of the kernel. All kernels, ranks and determinants are
		unaffected. We fix \eqref{eq:fm} once and for all and translate where necessary.
	\end{remark}
	
	\begin{lemma}[The potential does not depend on the multipliers]
		\label{lem:Vlambda}
		In the extended Lagrangian the symplectic potential is a function of $\xi$
		alone. Consequently, for $v=(u,w)\in\R^{n}\times\R^{k}$ the contraction is
		$\Phi_{v}=u\cdot\nabla V$, and it depends on the multiplier block only through
		the coupled equations of Definition~\ref{def:bordering}.
	\end{lemma}
	
	\begin{remark}[What bordering imposes]
		\label{rem:borderimposes}
		Varying the extended Lagrangian with respect to $\lambda^{\alpha}$ gives
		$\dot\Om_{\alpha}=0$, not $\Om_{\alpha}=0$. The bordered system is therefore
		equivalent to the original one only for initial data on the constraint surface.
		This is hypothesis \textbf{H11} of Section~\ref{sec:scope}; it is invisible in
		the algebra but essential in every dynamical interpretation, and it is the
		formal reason why the bordered system can carry more freedom than the original
		one (Section~\ref{sec:weak}).
	\end{remark}
	
	\subsection{Blocks, the canonical form, and the two reduced matrices}
	
	\begin{lemma}[Real antisymmetric canonical form]
		\label{lem:canonical}
		Let $f^{(0)}\in\R^{n\times n}$ be antisymmetric of rank $r=n-d$. There exists
		$Q\in O(n)$ with
		\begin{equation}
			Q^{\top}f^{(0)}Q=\begin{pmatrix}J&0\\0&0_{d}\end{pmatrix},
			\qquad
			J=\bigoplus_{\ell=1}^{r/2}\mu_{\ell}
			\begin{pmatrix}0&1\\-1&0\end{pmatrix},\quad\mu_{\ell}>0 .
			\label{eq:canonical}
		\end{equation}
		The last $d$ columns of $Q$ form an orthonormal basis $N_{0}$ of
		$\ker f^{(0)}$ and the first $r$ an orthonormal basis $Q_{R}$ of
		$\im f^{(0)}=(\ker f^{(0)})^{\perp}$. In particular $r$ is even and
		$\det J=\prod_{\ell}\mu_{\ell}^{2}>0$.
	\end{lemma}
	
	Let $N$ be any basis of $\ker f^{(0)}$ (not necessarily orthonormal; see
	Section~\ref{sec:strata}) and set
	\begin{equation}
		\Gamma:=N^{\top}B\in\R^{d\times k},
		\qquad
		B_{u}:=Q_{R}^{\top}B\in\R^{r\times k},
		\qquad
		M:=B_{u}^{\top}J^{-1}B_{u}\in\R^{k\times k}.
		\label{eq:GammaM}
	\end{equation}
	Since the inverse of an invertible antisymmetric matrix is antisymmetric,
	$M^{\top}=-M$. Writing $f^{+}=Q_{R}J^{-1}Q_{R}^{\top}$ for the pseudo-inverse of
	$f^{(0)}$ supported on $\im f^{(0)}$, we record the identity that will carry
	the interpretation of $M$:
	\begin{equation}
		M=B^{\top}f^{+}B .
		\label{eq:Mpseudo}
	\end{equation}
	
	\subsection{The effective ring}
	
	Ideal membership is not the same notion in $C^{\infty}(\Mm)$ and in a polynomial
	ring. Throughout we work in
	\begin{equation}
		\Rring=\R[\xi^{1},\dots,\xi^{n},\text{parameters}]
		\quad\text{or its localization at the relevant denominators,}
		\label{eq:ring}
	\end{equation}
	and every ideal-theoretic statement below is made in $\Rring$. We write
	$\Ical=\gen{\Om_{1},\dots,\Om_{k}}\subseteq\Rring$ and
	$\Sigma=\{\xi\in\Mm:\Om_{1}(\xi)=\dots=\Om_{k}(\xi)=0\}$ for the real zero
	locus. This is hypothesis \textbf{H12}.
	
	\section{The master variational identity and the strong gauge criterion}
	\label{sec:strong}
	
	Everything in this paper reduces to one identity, proved without using the
	equations of motion and without discarding any term.
	
	\begin{lemma}[First variation]
		\label{lem:firstvar}
		For an arbitrary variation $\delta\xi^{A}(t)$ of \eqref{eq:lag},
		\begin{equation}
			\delta S=\int_{t_{0}}^{t_{1}}E_{A}\,\delta\xi^{A}\,\dd t
			+\bigl[a_{A}\,\delta\xi^{A}\bigr]_{t_{0}}^{t_{1}},
			\qquad
			E_{A}:=f_{AB}\dot\xi^{B}-\pd{A}V .
			\label{eq:firstvar}
		\end{equation}
	\end{lemma}
	
	\begin{lemma}[Master identity]
		\label{lem:master}
		Let $u$ be a vector field on $\Mm$ and $\eta\in C^{\infty}([t_{0},t_{1}])$. For
		$\delta\xi^{A}=\eta(t)\,u^{A}(\xi(t))$,
		\begin{equation}
			\delta S=-\int_{t_{0}}^{t_{1}}\eta\,A_{u}\,\dd t
			+\bigl[\eta\,a_{A}u^{A}\bigr]_{t_{0}}^{t_{1}},
			\qquad
			A_{u}(\xi,\dot\xi):=\dot\xi^{A}f_{AB}u^{B}+u^{A}\pd{A}V .
			\label{eq:master}
		\end{equation}
	\end{lemma}
	
	\begin{proof}
		Varying term by term,
		$\delta(a_{A}\dot\xi^{A})=(\pd{B}a_{A})\delta\xi^{B}\dot\xi^{A}
		+a_{A}\delta\dot\xi^{A}$ and $\delta V=(\pd{A}V)\delta\xi^{A}$, with
		$\delta\dot\xi^{A}=\dot\eta\,u^{A}+\eta(\pd{B}u^{A})\dot\xi^{B}$. The integrand
		of $\delta S$ is thus
		\[
		\underbrace{\eta\,u^{B}(\pd{B}a_{A})\dot\xi^{A}}_{T_{1}}
		+\underbrace{a_{A}\bigl[\dot\eta\,u^{A}
			+\eta(\pd{B}u^{A})\dot\xi^{B}\bigr]}_{T_{2}}
		-\underbrace{\eta\,u^{A}\pd{A}V}_{T_{3}} .
		\]
		Since
		$\tfrac{\dd}{\dd t}[\eta a_{A}u^{A}]
		=\dot\eta a_{A}u^{A}+\eta(\pd{B}a_{A})\dot\xi^{B}u^{A}
		+\eta a_{A}(\pd{B}u^{A})\dot\xi^{B}$, we have
		$T_{2}=\tfrac{\dd}{\dd t}[\eta a_{A}u^{A}]-\eta(\pd{B}a_{A})\dot\xi^{B}u^{A}$.
		Relabelling $A\leftrightarrow B$ in the second term and adding $T_{1}$,
		\[
		T_{1}+\bigl(T_{2}-\tfrac{\dd}{\dd t}[\eta a_{A}u^{A}]\bigr)
		=\eta\,\dot\xi^{A}u^{B}\bigl(\pd{B}a_{A}-\pd{A}a_{B}\bigr)
		=-\eta\,\dot\xi^{A}f_{AB}u^{B},
		\]
		which together with $-T_{3}$ gives \eqref{eq:master}.
	\end{proof}
	
	\begin{remark}
		Comparing \eqref{eq:firstvar} and \eqref{eq:master}, $A_{u}=-u^{A}E_{A}$: the
		master identity is the contraction of the Euler--Lagrange expression with $u$.
		The entire theory consists in asking when that contraction vanishes
		identically, vanishes modulo the constraint ideal, or is a total derivative.
	\end{remark}
	
	\begin{theorem}[Strong gauge criterion]
		\label{thm:strong}
		Let $v$ be a smooth vector field on an open set $U\subseteq\Mm$. The
		transformation $\delta\xi=\varepsilon(t)v(\xi)$ is a symmetry of the action, in
		the sense that $\delta S$ reduces to a boundary term, for every compactly
		supported $\varepsilon$ and along every curve in $U$, if and only if
		\begin{equation}
			f(\xi)v(\xi)=0
			\quad\text{and}\quad
			v(\xi)\cdot\nabla V(\xi)=0
			\qquad\text{for every }\xi\in U .
			\label{eq:strong}
		\end{equation}
		Moreover, it suffices to test curves whose velocity at each point takes the
		$n+1$ values $\{0,e_{1},\dots,e_{n}\}$.
	\end{theorem}
	
	\begin{proof}
		$(\Leftarrow)$ Both conditions give $A_{v}\equiv0$ and \eqref{eq:master} leaves
		only the boundary term. $(\Rightarrow)$ Fix $\xi_{0}\in U$ and a curve through
		$\xi_{0}$; taking $\varepsilon$ supported near the corresponding time kills the
		boundary term, and the fundamental lemma of the calculus of variations forces
		$A_{v}(\xi_{0},X)=0$ for the velocity $X$ realized there. The constant curve
		gives $X=0$, hence $v\cdot\nabla V=0$ at $\xi_{0}$; the curves with $X=e_{i}$
		then give $(fv)_{i}=0$ for each $i$. As $\xi_{0}$ was arbitrary,
		\eqref{eq:strong} holds on $U$.
	\end{proof}
	
	\begin{definition}[Strong gauge direction]
		\label{def:strongdir}
		A vector field $v$ satisfying \eqref{eq:strong} on an open set is called a
		\emph{strong gauge direction}. Theorem~\ref{thm:strong} is the only place in
		this paper where the word \emph{generator} may be used without qualification.
	\end{definition}
	
	\begin{remark}[Rigid symmetries]
		If invariance is demanded only for constant $\varepsilon$, the condition weakens
		to $A_{v}=\dd K/\dd t$ for some $K$, with Noether charge $a_{A}v^{A}-K$.
		Theorem~\ref{thm:strong} characterizes symmetries local in time, which is the
		operative definition of gauge freedom.
	\end{remark}
	
	\begin{proposition}[Solvability and dynamical indeterminacy]
		\label{prop:solv}
		At a fixed point $\xi$, regard \eqref{eq:eom} as a linear system
		$f(\xi)X=\nabla V(\xi)$ in $X=\dot\xi$. Then it is solvable if and only if
		$\Phi_{v}(\xi)=0$ for every $v\in\ker f(\xi)$, and in that case the solution set
		is the affine space $X_{0}+\ker f(\xi)$.
	\end{proposition}
	
	\begin{proof}
		$\im f=(\ker f^{\top})^{\perp}=(\ker f)^{\perp}$ because $f^{\top}=-f$; so
		$\nabla V\in\im f$ if and only if $\nabla V\perp\ker f$. The second statement is
		the affine structure of the solution set of a consistent linear system.
	\end{proof}
	
	\begin{definition}[Strong and weak null modes]
		\label{def:strongweak}
		A null mode $v\in\ker f^{(m)}$ is \emph{strong} if $\Phi_{v}\equiv0$ on an open
		set, and \emph{weak} if $\Phi_{v}\in\Ical$ but $\Phi_{v}\not\equiv0$.
	\end{definition}
	
	\begin{proposition}[Invariance defect of a weak mode]
		\label{prop:defect}
		If $v\in\ker f^{(m)}$ is weak with $\Phi_{v}=\sum_{\gamma}c^{\gamma}\Om_{\gamma}$
		then, for $\delta\xi=\varepsilon v$,
		\begin{equation}
			\delta S=-\int\varepsilon\sum_{\gamma}c^{\gamma}\Om_{\gamma}\,\dd t
			+\text{boundary}.
			\label{eq:defect}
		\end{equation}
		Hence $\delta S$ vanishes on curves lying in $\Sigma$ but not off the constraint
		surface, and by Theorem~\ref{thm:strong} the transformation is not a symmetry of
		the action.
	\end{proposition}
	
	\begin{proposition}[Independence of the bordering convention]
		\label{prop:convention}
		Let $\widetilde V=V-\sum_{\gamma}g^{\gamma}\Om_{\gamma}$ be any potential
		agreeing with $V$ on $\Sigma$. For every $v=(u,w)\in\ker f^{(m)}$,
		$\Phi_{v}-\widetilde\Phi_{v}\in\Ical$. Therefore the weak verdict is
		convention-independent, whereas the strong verdict need not be.
	\end{proposition}
	
	\begin{proof}
		$\nabla V-\nabla\widetilde V
		=\sum_{\gamma}(g^{\gamma}\nabla\Om_{\gamma}+\Om_{\gamma}\nabla g^{\gamma})$.
		Contract with $u$ and use $B^{\top}u=0$, which holds for kernel elements by
		Theorem~\ref{thm:kernel}(b); only
		$\sum_{\gamma}\Om_{\gamma}(u\cdot\nabla g^{\gamma})\in\Ical$ survives.
	\end{proof}
	
	\section{Exact kernel structure of the bordered matrix}
	\label{sec:kernel}
	
	\begin{theorem}[Structure of the bordered kernel]
		\label{thm:kernel}
		With the notation of \eqref{eq:canonical}--\eqref{eq:GammaM} and $N=N_{0}$
		orthonormal, write $u=Q\binom{x}{y}$ with $x\in\R^{r}$ and $y=N^{\top}u\in\R^{d}$.
		Then $(u,w)\in\ker f^{(m)}$ if and only if
		\begin{enumerate}[label=(\alph*),leftmargin=2.2em,itemsep=0.25em]
			\item $\Gamma w=0$;
			\item $B^{\top}u=0$, equivalently $\Gamma^{\top}y=Mw$;
			\item $x=-J^{-1}B_{u}w$.
		\end{enumerate}
	\end{theorem}
	
	\begin{proof}
		The condition $f^{(m)}(u,w)^{\top}=0$ is the pair
		$f^{(0)}u+Bw=0$, $B^{\top}u=0$. Multiplying the first by $Q^{\top}$ and
		inserting $QQ^{\top}=I$,
		\[
		\begin{pmatrix}J&0\\0&0\end{pmatrix}\begin{pmatrix}x\\y\end{pmatrix}
		+\begin{pmatrix}B_{u}\\ \Gamma\end{pmatrix}w
		=\begin{pmatrix}Jx+B_{u}w\\ \Gamma w\end{pmatrix}=0,
		\]
		whose lower block is (a) and whose upper block gives (c) because $J$ is
		invertible. For the second equation,
		$B^{\top}u=B^{\top}QQ^{\top}u=B_{u}^{\top}x+\Gamma^{\top}y$, and substituting
		(c), $B^{\top}u=-Mw+\Gamma^{\top}y$. Every step is an identity, so the
		equivalence holds in both directions.
	\end{proof}
	
	\begin{corollary}[Dimension of the kernel]
		\label{cor:dim}
		Let $\Pi$ be the orthogonal projector onto $\ker\Gamma\subseteq\R^{k}$. Then
		\begin{equation}
			\dim\ker f^{(m)}
			=\dim\ker\Gamma^{\top}+\dim\ker\Gamma-\rank(\Pi M\Pi),
			\label{eq:dim}
		\end{equation}
		and $\rank(\Pi M\Pi)$ is even.
	\end{corollary}
	
	\begin{proof}
		By Theorem~\ref{thm:kernel} a kernel element is determined by the pair $(w,y)$,
		since $x$ is fixed by $w$; and $(w,y)\mapsto(u,w)$ is a linear bijection onto
		$\ker f^{(m)}$. The equation $\Gamma^{\top}y=Mw$ is solvable if and only if
		$Mw\in\im\Gamma^{\top}=(\ker\Gamma)^{\perp}$, that is $\Pi Mw=0$. For
		$w\in\ker\Gamma$ one has $\Pi w=w$, so the admissible $w$ form
		$W=\ker(\Pi M\Pi|_{\ker\Gamma})$, of dimension
		$\dim\ker\Gamma-\rank(\Pi M\Pi)$; note $\rank(\Pi M\Pi)$ as a $k\times k$ matrix
		equals the rank of the induced endomorphism of $\ker\Gamma$. For each admissible
		$w$ the fibre of $y$'s is a translate of $\ker\Gamma^{\top}$, and the fibration
		is linear, so dimensions add. Finally $\Pi M\Pi$ is antisymmetric and a real
		antisymmetric matrix has even rank.
	\end{proof}
	
	\begin{corollary}[Regularity]
		\label{cor:regularity}
		Suppose $k=d$. If $\det\Gamma\neq0$ then $\ker f^{(m)}=0$. Conversely, if
		$\det\Gamma=0$ then $\dim\ker f^{(m)}\geq\dim\ker\Gamma\geq1$, because
		$\rank(\Pi M\Pi)\leq\dim\ker\Gamma$. Hence
		$\det f^{(m)}\neq0\iff\det\Gamma\neq0$.
	\end{corollary}
	
	\begin{remark}[$k=d$ is not necessary for regularity]
		\label{rem:knotd}
		Corollary~\ref{cor:regularity} is stated for $k=d$, which is the situation of
		the first bordering round, when one constraint is generated per null direction.
		It should not be read as saying that $k=d$ is necessary for $f^{(m)}$ to be
		regular. Take $d=1$, $k=3$, $\Gamma=(1,0,0)$ and $M$ with $M_{23}=-M_{32}=1$ and
		all other entries zero: then $\ker\Gamma^{\top}=0$, $\dim\ker\Gamma=2$ and
		$\rank(\Pi M\Pi)=2$, so \eqref{eq:dim} gives $\dim\ker f^{(m)}=0$ and the
		bordered matrix is invertible with $k\neq d$. In later rounds, where $k$ may
		exceed $d$, the criterion is \eqref{eq:dim} itself and not the determinant of a
		square $\Gamma$.
	\end{remark}
	
	The two summands of \eqref{eq:dim} are not interchangeable. The next proposition
	identifies the first; Section~\ref{sec:firstclass} identifies the second.
	
	\begin{proposition}[The multiplier-free part of the kernel]
		\label{prop:mfree}
		The set of kernel elements with vanishing multiplier block is
		\begin{equation}
			\{(u,0)\in\ker f^{(m)}\}
			=\bigl(\ker f^{(0)}\cap\ker B^{\top}\bigr)\times\{0\}
			\;\cong\;\ker\Gamma^{\top},
			\label{eq:mfree}
		\end{equation}
		the isomorphism being $u=Ny\mapsto y$. These are exactly the null directions of
		$f^{(0)}$ tangent to $\Sigma$.
	\end{proposition}
	
	\begin{proof}
		For $w=0$, Theorem~\ref{thm:kernel} reads $x=0$ and $\Gamma^{\top}y=0$; so
		$u=Ny$ with $y\in\ker\Gamma^{\top}$, and $B^{\top}u=\Gamma^{\top}y=0$.
		Conversely any such $u$ gives a kernel element. Tangency to $\Sigma$ is
		$B^{\top}u=0$ under regularity of $\Sigma$.
	\end{proof}
	
	\begin{remark}[Automatic tangency]
		\label{rem:tangency}
		Item (b) of Theorem~\ref{thm:kernel} says that \emph{every} element of
		$\ker f^{(m)}$ is tangent to $\Sigma$ in its physical block: bordering enforces
		tangency algebraically, a property that in the Dirac formalism is established a
		posteriori.
	\end{remark}
	
	\section{What \texorpdfstring{$\Gamma$}{Gamma} and \texorpdfstring{$M$}{M} really are}
	\label{sec:dirac}
	
	The bordered matrix is built from two reduced objects, $\Gamma$ and $M$. They are
	routinely treated as one, and, as we now show, incorrectly. This section
	identifies each of them inside the Dirac--Bergmann picture, and does so without
	leaving the frame in which $\Gamma$, $B_{u}$, $M$ and the Pfaffian factors
	$\mu_{\ell}$ were defined in Section~\ref{sec:conventions}. The identification
	requires an excursion into the cotangent bundle $T^{*}\Mm$, and we are explicit
	about that: the constraints $\pi_{\alpha}$ that appear below do not live on
	$\Mm$.
	
	\subsection{A symmetry obstruction}
	\label{sec:obstruction}
	
	Suppose the constraints are generated by the consistency condition,
	$\Om_{\beta}=N_{\beta}^{j}\pd{j}V$, and the kernel frame is locally constant
	(hypothesis \textbf{H3}). Then
	\begin{equation}
		\Gamma_{\alpha\beta}=N_{\alpha}^{i}\pd{i}\Om_{\beta}
		=\sum_{i,j}N_{\alpha}^{i}\bigl(\pd{i}\pd{j}V\bigr)N_{\beta}^{j},
		\label{eq:hessian}
	\end{equation}
	Identity \eqref{eq:hessian} is correct, and it is the kernel--Hessian identity
	established in \cite{CLMCP2026}. What cannot be sustained is the further
	identification made there of the right-hand side with the bracket matrix of the
	generated constraints, and hence of $\Gamma$ with the Dirac constraint matrix.
	The obstruction is one of symmetry: the right-hand side of \eqref{eq:hessian} is
	\emph{symmetric} in $(\alpha,\beta)$ because the Hessian is, whereas a matrix of
	Poisson brackets $\{\Om_{\alpha},\Om_{\beta}\}$ is \emph{antisymmetric}. Two
	matrices that are simultaneously symmetric and antisymmetric vanish, so the
	identification forces $\Gamma=0$. It is already visible in the smallest possible
	case.
	
	\begin{example}[Minimal obstruction, $d=k=1$]
		\label{ex:d1}
		Let $L=q_{2}\dot q_{1}-\tfrac{m}{2}q_{3}^{2}$ with $m\neq0$. Then
		$\ker f^{(0)}=\spn\{e_{3}\}$, the generated constraint is $\Om=\pd{3}V=mq_{3}$,
		and $\Gamma=(m)$, a nonzero $1\times1$ matrix. The bracket matrix of the single
		constraint is $(\{\Om,\Om\})=(0)$, because every antisymmetric $1\times1$ matrix
		vanishes. Thus $\det\Gamma=m\neq0$ while the constraint bracket matrix is
		singular: a biconditional of the form
		$\det f^{(m)}\neq0\iff\det\{\Om_{\alpha},\Om_{\beta}\}\neq0$ cannot hold in this
		reading whenever $d$ is odd.
		
		The remedy is already contained in the example. Adjoining to $\Om$ the primary
		constraint $\pi$ attached to the degenerate direction $e_{3}$ produces the
		$2\times2$ matrix
		\[
		C=\begin{pmatrix}0&-m\\ m&0\end{pmatrix},
		\qquad
		\det C=m^{2}=\det(\Gamma)^{2}=\det f^{(1)} ,
		\]
		and the biconditional is restored. What fails is not the criterion but the
		identification of the object it refers to: $\Gamma$ is a block of $C$, not $C$.
	\end{example}
	
	The rest of the section makes this precise in general. Two ingredients are
	needed, and it is worth separating them at the outset: a canonical extension in
	which the primary constraints exist, and a piece of linear algebra about
	Pfaffians of block antisymmetric matrices. Only the first carries physical
	content; the second is used to control determinants and is stated so that it
	drags no hypothesis with it.
	
	\subsection{The canonical extension and the reduced constraint set}
	\label{sec:extension}
	
	The Lagrangian \eqref{eq:lag} is defined on $\Mm$, on which there are no momenta.
	The comparison with Dirac--Bergmann theory therefore takes place on the
	cotangent bundle, and we set it up explicitly rather than by analogy.
	
	\begin{lemma}[Canonical extension]
		\label{lem:extension}
		Let $T^{*}\Mm$ carry coordinates $(\xi^{A},p_{A})$ and the canonical bracket
		$\{\xi^{A},p_{B}\}_{\mathrm{can}}=\delta^{A}_{B}$. The Legendre transform of
		\eqref{eq:lag} is maximally degenerate and produces the $n$ primary constraints
		\begin{equation}
			\chi_{A}:=p_{A}-a_{A}(\xi)\approx0,
			\qquad
			\{\chi_{A},\chi_{B}\}_{\mathrm{can}}=f^{(0)}_{AB},
			\label{eq:chi}
		\end{equation}
		while the canonical Hamiltonian is
		\begin{equation}
			H=p_{A}\dot\xi^{A}-L=\chi_{A}\dot\xi^{A}+V\;\approx\;V .
			\label{eq:H}
		\end{equation}
	\end{lemma}
	
	\begin{proof}
		Since $L$ is linear in the velocities, $p_{A}=\partial L/\partial\dot\xi^{A}
		=a_{A}(\xi)$, so all $n$ momenta are constrained. For the bracket,
		\[
		\{\chi_{A},\chi_{B}\}_{\mathrm{can}}
		=-\{p_{A},a_{B}\}_{\mathrm{can}}-\{a_{A},p_{B}\}_{\mathrm{can}}
		=\pd{A}a_{B}-\pd{B}a_{A}=f^{(0)}_{AB},
		\]
		the remaining two terms vanishing because $\{p_{A},p_{B}\}=\{a_{A},a_{B}\}=0$.
		Equation \eqref{eq:H} is the definition of $H$ with $L$ substituted.
	\end{proof}
	
	Lemma~\ref{lem:extension} says that the presymplectic matrix of the first-order
	Lagrangian \emph{is} the constraint matrix of its own primary constraints. Its
	rank $r$ therefore splits the $\chi_{A}$ into a second-class part and a
	remainder, and the splitting is the one already fixed in
	Section~\ref{sec:conventions}.
	
	\begin{definition}[Second-class part and null-direction constraints]
		\label{def:pi}
		With $Q=(Q_{R}\;N)$ as in Lemma~\ref{lem:canonical}, set
		\begin{equation}
			\chi^{i}_{R}:=Q_{R}^{Ai}\chi_{A}\quad(i=1,\dots,r),
			\qquad
			\pi_{\alpha}:=N_{\alpha}^{A}\chi_{A}\quad(\alpha=1,\dots,d).
			\label{eq:pi}
		\end{equation}
		The matrix $\{\chi_{R}^{i},\chi_{R}^{j}\}_{\mathrm{can}}=J^{ij}$ is invertible,
		so the $\chi_{R}^{i}$ are the second-class part of the primary set; the
		$\pi_{\alpha}$ are the combinations along the degenerate directions.
	\end{definition}
	
	Three points deserve emphasis, because the literature and our own earlier
	formulation blur them. First, $\pi_{\alpha}$ is not a coordinate momentum of
	$\Mm$ and is not obtained by declaring $z^{\alpha}$ to be a coordinate: it is a
	particular linear combination of the primary constraints, living on
	$T^{*}\Mm$. Second, the full primary set is $\{\chi_{A}\}$, all $n$ of them; the
	$d$ functions $\pi_{\alpha}$ are not ``the'' primary constraints but the
	first-class-candidate part of them. Third, and this is what keeps the section
	internally consistent, definition \eqref{eq:pi} uses \emph{the same} frame
	$N$ that defines $\Gamma=N^{\top}B$ and, when $N=N_{0}$ is the orthonormal frame
	of Lemma~\ref{lem:canonical}, the same one that defines $B_{u}$, $M$ and the
	factors $\mu_{\ell}$.
	
	Eliminating the second-class part gives the bracket in which everything below is
	computed, and at the same time supplies the justification, absent from
	Definition~\ref{def:fjbracket}, that the Faddeev--Jackiw bracket is a bracket at
	all.
	
	\begin{definition}[Faddeev--Jackiw bracket]
		\label{def:fjbracket}
		For $F,G\in\Rring$ set
		$\{F,G\}_{\FJ}:=(\nabla F)^{\top}f^{+}(\nabla G)$, with
		$f^{+}=Q_{R}J^{-1}Q_{R}^{\top}$. When $f^{(0)}$ is invertible this reduces to
		$\{F,G\}=\pd{A}F\,(f^{-1})^{AB}\pd{B}G$, which by Lemma~\ref{lem:eom} satisfies
		$\dot F=\{F,V\}$.
	\end{definition}
	
	\begin{lemma}[The Faddeev--Jackiw bracket is a reduced Dirac bracket]
		\label{lem:fjdirac}
		Let $\{\cdot,\cdot\}_{D}$ denote the Dirac bracket built from the second-class
		set $\{\chi_{R}^{i}\}$ of Definition~\ref{def:pi}. For functions $F,G$ of $\xi$
		alone,
		\begin{equation}
			\{F,G\}_{D}\approx(\nabla F)^{\top}f^{+}(\nabla G)=\{F,G\}_{\FJ} ,
			\label{eq:fjdirac}
		\end{equation}
		with equality modulo the primary constraints $\chi_{A}$, and with strong
		equality whenever $Q_{R}$ is locally constant. In particular
		$\{\cdot,\cdot\}_{\FJ}$ satisfies the Jacobi identity on such functions.
	\end{lemma}
	
	\begin{proof}
		For $F,G$ functions of $\xi$ only one has $\{F,G\}_{\mathrm{can}}=0$, while
		$\{F,\chi_{R}^{i}\}_{\mathrm{can}}=Q_{R}^{Ai}\pd{A}F$ and
		$\{\chi_{R}^{j},G\}_{\mathrm{can}}=-Q_{R}^{Bj}\pd{B}G$ exactly, because
		$Q_{R}^{Ai}$ is a function of $\xi$ and therefore commutes with $F$ and $G$.
		Hence
		\[
		\{F,G\}_{D}
		=-\{F,\chi_{R}^{i}\}\,(\mathcal{C}^{-1})_{ij}\,\{\chi_{R}^{j},G\}
		=(\nabla F)^{\top}Q_{R}\,\mathcal{C}^{-1}Q_{R}^{\top}(\nabla G),
		\qquad
		\mathcal{C}_{ij}:=\{\chi_{R}^{i},\chi_{R}^{j}\}_{\mathrm{can}} .
		\]
		Expanding the brackets of the $\chi$'s as in the proof of Theorem~\ref{thm:dirac} gives $\mathcal{C}_{ij}=J_{ij}+O(\chi)$, with
		$J=Q_{R}^{\top}f^{(0)}Q_{R}$ and with the remainder vanishing identically when
		$Q_{R}$ is locally constant; substituting yields \eqref{eq:fjdirac}. Jacobi is
		the standard property of the Dirac bracket.
	\end{proof}
	
	\begin{remark}[Scope of Lemma~\ref{lem:fjdirac}]
		The identification is asserted only for functions on $\Mm$, which is all that is
		needed: the constraints $\Om_{\alpha}$ and the potential $V$ are of that type.
		Nothing is claimed about a bracket structure on $T^{*}\Mm$ beyond the standard
		Dirac construction, and $f^{+}$ is not a Poisson bivector on $\Mm$ in its own
		right; it inherits Jacobi from the reduction.
	\end{remark}
	
	\subsection{The three blocks}
	\label{sec:threeblocks}
	
	\begin{theorem}[Block structure of the reduced constraint matrix]
		\label{thm:dirac}
		Let $N$ be any smooth frame of $\ker f^{(0)}$, let $\pi_{\alpha}$ be as in
		\eqref{eq:pi}, and let $\Om_{1},\dots,\Om_{k}$ be the accumulated constraints,
		functions of $\xi$ alone. Then, in the Dirac bracket of
		Lemma~\ref{lem:fjdirac},
		\begin{equation}
			\{\pi_{\alpha},\pi_{\beta}\}_{D}\approx0,
			\qquad
			\{\pi_{\alpha},\Om_{\beta}\}_{D}\approx-\Gamma_{\alpha\beta},
			\qquad
			\{\Om_{\alpha},\Om_{\beta}\}_{D}\approx M_{\alpha\beta},
			\label{eq:blocks}
		\end{equation}
		where $\approx$ denotes equality modulo the primary constraints $\chi_{A}$, that
		is, on the primary constraint surface, which is the setting in which a
		constraint algebra is always evaluated. If the frame data are locally constant,
		\begin{equation}
			\pd{A}N_{\beta}^{B}=0
			\quad\text{and}\quad
			\pd{A}Q_{R}^{Bi}=0
			\qquad\text{on the neighbourhood considered,}
			\label{eq:constframe}
		\end{equation}
		then all the correction terms vanish and the three relations of
		\eqref{eq:blocks} are strong equalities,
		\begin{equation}
			\{\pi_{\alpha},\pi_{\beta}\}_{D}=0,
			\qquad
			\{\pi_{\alpha},\Om_{\beta}\}_{D}=-\Gamma_{\alpha\beta},
			\qquad
			\{\Om_{\alpha},\Om_{\beta}\}_{D}=M_{\alpha\beta}.
			\label{eq:blocksstrong}
		\end{equation}
		Consequently the matrix of brackets of the reduced constraint set
		$(\pi_{1},\dots,\pi_{d},\Om_{1},\dots,\Om_{k})$, evaluated on the primary
		surface, is
		\begin{equation}
			C=\begin{pmatrix}0&-\Gamma\\ \Gamma^{\top}&M\end{pmatrix}
			\in\R^{(d+k)\times(d+k)},
			\qquad C^{\top}=-C .
			\label{eq:Cmatrix}
		\end{equation}
		No hypothesis on the frame, and in particular not \textbf{H3}, is used for
		\eqref{eq:blocks} or for \eqref{eq:Cmatrix}; condition \eqref{eq:constframe} is
		needed only to upgrade $\approx$ to $=$.
	\end{theorem}
	
	\begin{proof}
		We first record the canonical brackets, which are exact. Since $N_{\alpha}^{A}$
		and $\Om_{\beta}$ are functions of $\xi$ alone, their mutual bracket vanishes,
		so
		\begin{equation}
			\{\pi_{\alpha},\Om_{\beta}\}_{\mathrm{can}}
			=N_{\alpha}^{A}\{\chi_{A},\Om_{\beta}\}_{\mathrm{can}}
			=N_{\alpha}^{A}\{p_{A},\Om_{\beta}\}_{\mathrm{can}}
			=-N_{\alpha}^{A}\pd{A}\Om_{\beta}
			=-\Gamma_{\alpha\beta},
			\label{eq:canpiOm}
		\end{equation}
		an identity valid on all of $T^{*}\Mm$ and for any frame, by the definition
		$\Gamma=N^{\top}B$ with $B_{A\beta}=\pd{A}\Om_{\beta}$. Likewise
		\[
		\{\pi_{\alpha},\pi_{\beta}\}_{\mathrm{can}}
		=N_{\alpha}^{A}N_{\beta}^{B}f^{(0)}_{AB}
		+\bigl(Q\text{-independent terms}\propto\chi\bigr)
		\approx0,
		\]
		the first summand vanishing because $N_{\beta}\in\ker f^{(0)}$, and the
		remaining ones arising from $\{N_{\alpha}^{A},\chi_{B}\}=\pd{B}N_{\alpha}^{A}$
		and vanishing identically under \eqref{eq:constframe}.
		
		It remains to control the Dirac corrections. Expanding
		$\pi_{\alpha}=N_{\alpha}^{A}\chi_{A}$ and $\chi_{R}^{i}=Q_{R}^{Ai}\chi_{A}$,
		\begin{equation}
			\{\pi_{\alpha},\chi_{R}^{i}\}_{\mathrm{can}}
			=\underbrace{N_{\alpha}^{A}f^{(0)}_{AB}Q_{R}^{Bi}}_{=\,0}
			-N_{\alpha}^{A}\bigl(\pd{A}Q_{R}^{Bi}\bigr)\chi_{B}
			+Q_{R}^{Bi}\bigl(\pd{B}N_{\alpha}^{A}\bigr)\chi_{A}
			\;\approx\;0 ,
			\label{eq:pichiR}
		\end{equation}
		which vanishes identically under \eqref{eq:constframe}. Hence the Dirac
		correction to any bracket involving $\pi_{\alpha}$ is proportional to the
		$\chi_{A}$, and \eqref{eq:canpiOm} gives
		$\{\pi_{\alpha},\Om_{\beta}\}_{D}\approx-\Gamma_{\alpha\beta}$, with strong
		equality under \eqref{eq:constframe}; the same applies to the first block.
		
		The third block is Lemma~\ref{lem:fjdirac} applied to $F=\Om_{\alpha}$,
		$G=\Om_{\beta}$, together with $M=B^{\top}f^{+}B$ from \eqref{eq:Mpseudo}; by
		that lemma the relation is weak for a non-constant $Q_{R}$ and strong under
		\eqref{eq:constframe}. Antisymmetry of $C$ follows from $M^{\top}=-M$ and the
		placement of $-\Gamma$ and $\Gamma^{\top}$, and is exact because $\Gamma$ and
		$M$ are the algebraic objects of \eqref{eq:GammaM}, not the brackets they
		represent.
	\end{proof}
	
	\begin{remark}[What $C$ is, and what it is not]
		\label{rem:whatCis}
		$C$ is the bracket matrix, on the primary constraint surface, of the
		\emph{reduced} constraint set obtained after the second-class primary pairs
		$\chi_{R}^{i}$ have been eliminated. As a matrix of functions, $C$ is defined by
		\eqref{eq:Cmatrix} in terms of the algebraic objects $\Gamma$ and $M$; the
		brackets it represents reproduce those entries weakly in general and strongly
		under \eqref{eq:constframe}. It is not
		the Dirac matrix of the full unreduced primary set $\{\chi_{A}\}$, which is
		$f^{(0)}$ itself by \eqref{eq:chi}, nor of $\{\chi_{A}\}\cup\{\Om_{\alpha}\}$.
		The reduction is exactly the same splitting $\R^{n}=\im f^{(0)}\oplus\ker f^{(0)}$
		that bordering performs, which is why the two constructions see the same
		objects.
	\end{remark}
	
	\begin{remark}[The canonical bracket of the $\Om$'s vanishes]
		\label{rem:canvanishes}
		It is worth stating explicitly what the third identity of \eqref{eq:blocks} does
		\emph{not} say. The $\Om_{\alpha}$ are functions of $\xi$ alone, so
		$\{\Om_{\alpha},\Om_{\beta}\}_{\mathrm{can}}=0$ identically on $T^{*}\Mm$. The
		content of $\{\Om_{\alpha},\Om_{\beta}\}_{D}=M_{\alpha\beta}$ is entirely in the
		Dirac reduction; writing \eqref{eq:blocks} with the canonical bracket would
		force $M=0$ and destroy the statement. All three entries of \eqref{eq:blocks}
		are computed in one and the same bracket.
	\end{remark}
	
	\subsection{\texorpdfstring{$M$}{M} depends on the splitting; \texorpdfstring{$\Pi M\Pi$}{Pi M Pi} does not}
	\label{sec:invariance}
	
	The second-class part of the primary set is not canonically determined: any
	$r$-dimensional subspace on which the restriction of $f^{(0)}$ is nondegenerate
	will serve, and the choice $Q_{R}$ made in Lemma~\ref{lem:canonical} uses the
	Euclidean structure of the coordinates, which is extraneous data. It is
	therefore necessary to ask which of the objects built from it are intrinsic.
	
	\begin{proposition}[Invariance]
		\label{prop:invariance}
		Let $Q_{R}'=Q_{R}A+ND$ with $A\in GL(r,\R)$ and $D\in\R^{d\times r}$ be another
		admissible complement, and let $M'=B^{\top}f'^{+}B$ be the matrix built from it,
		where $f'^{+}=Q_{R}'\bigl(Q_{R}'^{\top}f^{(0)}Q_{R}'\bigr)^{-1}Q_{R}'^{\top}$.
		Writing $E=NDA^{-1}$,
		\begin{equation}
			M'=M+B^{\top}\bigl(EJ^{-1}Q_{R}^{\top}+Q_{R}J^{-1}E^{\top}
			+EJ^{-1}E^{\top}\bigr)B,
			\label{eq:Mprime}
		\end{equation}
		so $M'\neq M$ in general. However $w'^{\top}M'w=w'^{\top}Mw$ for all
		$w,w'\in\ker\Gamma$; equivalently $\Pi M\Pi$, and therefore $\rank(\Pi M\Pi)$,
		is independent of the choice of complement.
	\end{proposition}
	
	\begin{proof}
		From $f^{(0)}N=0$ one gets $f^{(0)}Q_{R}'=f^{(0)}Q_{R}A$ and
		$Q_{R}'^{\top}f^{(0)}Q_{R}'=A^{\top}JA$, whence
		$f'^{+}=(Q_{R}+E)J^{-1}(Q_{R}+E)^{\top}$ and \eqref{eq:Mprime} follows by
		expansion. Every correction term carries a factor $E^{\top}B$ or its transpose,
		and for $w\in\ker\Gamma$,
		$E^{\top}Bw=A^{-\top}D^{\top}N^{\top}Bw=A^{-\top}D^{\top}\Gamma w=0$.
	\end{proof}
	
	\begin{remark}[Which objects are intrinsic]
		\label{rem:intrinsic}
		The subspace $\ker\Gamma=\{w\in\R^{k}:Bw\in\im f^{(0)}\}$ is intrinsic, as is
		$\ker f^{(0)}\cap\ker B^{\top}$; consequently both summands of the dimension
		formula \eqref{eq:dim} are intrinsic, and so is the first-class count
		\eqref{eq:fccount} of Section~\ref{sec:firstclass}. What is not intrinsic is
		$M$ itself. Accordingly we shall always say that $M$ is the constraint bracket
		matrix \emph{relative to a choice of second-class splitting}, and every
		statement in which $M$ is used will involve it either through $\Pi M\Pi$ or
		through $\det C$, which by Corollary~\ref{cor:pfC} does not depend on $M$ at
		all.
	\end{remark}
	
	\subsection{A Pfaffian lemma}
	\label{sec:pfaffian}
	
	The determinantal content of $C$ is pure linear algebra and is best isolated
	from everything above.
	
	\begin{lemma}[Block Pfaffian]
		\label{lem:pf}
		Let $A\in\R^{d\times d}$ and let $D\in\R^{d\times d}$ be antisymmetric. Set
		$K=\bigl(\begin{smallmatrix}0&A\\-A^{\top}&D\end{smallmatrix}\bigr)$. Then
		\begin{equation}
			\det K=\det(A)^{2},
			\qquad
			\Pf(K)=(-1)^{d(d-1)/2}\det A ,
			\label{eq:pflemma}
		\end{equation}
		both independent of $D$.
	\end{lemma}
	
	\begin{proof}
		Assume first $A$ invertible. The block column operation
		$\mathrm{col}_{2}\mapsto\mathrm{col}_{2}+\mathrm{col}_{1}X$ with
		$X=(A^{\top})^{-1}D$ leaves the determinant unchanged and turns $K$ into
		$\bigl(\begin{smallmatrix}0&A\\-A^{\top}&0\end{smallmatrix}\bigr)$. For that
		matrix,
		\[
		\det\begin{pmatrix}0&A\\-A^{\top}&0\end{pmatrix}
		=(-1)^{d^{2}}\det(A)\det(-A^{\top})
		=(-1)^{d^{2}+d}\det(A)^{2}
		=\det(A)^{2},
		\]
		since $d^{2}+d$ is even. The singular case follows by continuity. For the
		Pfaffian, $\Pf(K)^{2}=\det K=\det(A)^{2}$ is an identity of polynomials in the
		entries of $A$ and $D$; as the polynomial ring is an integral domain,
		$\Pf(K)=\epsilon\det A$ with $\epsilon\in\{\pm1\}$ constant, and evaluating at
		$D=0$, $A=I_{d}$ gives $\epsilon=(-1)^{d(d-1)/2}$.
	\end{proof}
	
	\begin{corollary}[Pfaffian and determinant of $C$]
		\label{cor:pfC}
		If $k=d$ then
		\begin{equation}
			\Pf(C)=(-1)^{d(d+1)/2}\det\Gamma,
			\qquad
			\det C=\det(\Gamma)^{2},
			\label{eq:Pf}
		\end{equation}
		both independent of $M$.
	\end{corollary}
	
	\begin{proof}
		Apply Lemma~\ref{lem:pf} with $A=-\Gamma$ and $D=M$:
		$\Pf(C)=(-1)^{d(d-1)/2}\det(-\Gamma)=(-1)^{d(d-1)/2+d}\det\Gamma$, and
		$d(d-1)/2+d=d(d+1)/2$.
	\end{proof}
	
	\begin{remark}[$\Gamma$ need not be symmetric]
		\label{rem:nosym}
		Corollary~\ref{cor:pfC} uses nothing about $\Gamma$ beyond its being a square
		matrix. The symmetry of $\Gamma$ recorded in \eqref{eq:hessian} is a consequence
		of \textbf{H3} together with the specific form $\Om_{\beta}=N_{\beta}^{j}\pd{j}V$
		of the generated constraints, not a structural feature; $C$ is antisymmetric,
		and \eqref{eq:Pf} holds, whether or not $\Gamma$ is symmetric. See
		Section~\ref{sec:strata} for the correction term that measures the failure of
		symmetry in a non-constant frame.
	\end{remark}
	
	\subsection{The determinantal factorization}
	\label{sec:factorization}
	
	\begin{proposition}[Factorization]
		\label{prop:factorization}
		Let $N=N_{0}$ be the orthonormal frame of Lemma~\ref{lem:canonical} and $k=d$.
		Then
		\begin{equation}
			\det f^{(m)}
			=\det J\,\det C
			=\Bigl(\prod_{\ell=1}^{r/2}\mu_{\ell}^{2}\Bigr)\det C
			=\Bigl(\prod_{\ell=1}^{r/2}\mu_{\ell}^{2}\Bigr)\det(\Gamma)^{2}.
			\label{eq:factorization}
		\end{equation}
	\end{proposition}
	
	\begin{proof}
		Put $T=\operatorname{diag}(Q,I_{k})$, so $\det T=\pm1$ and
		$\det f^{(m)}=\det\bigl(T^{\top}f^{(m)}T\bigr)$. Using
		$Q^{\top}f^{(0)}Q=\operatorname{diag}(J,0_{d})$ and
		$Q^{\top}B=\binom{B_{u}}{\Gamma}$,
		\begin{equation}
			T^{\top}f^{(m)}T
			=\begin{pmatrix}J&0&B_{u}\\ 0&0&\Gamma\\ -B_{u}^{\top}&-\Gamma^{\top}&0\end{pmatrix},
			\label{eq:congruent}
		\end{equation}
		with blocks of sizes $r$, $d$ and $k$. Since $J$ is invertible, the determinant
		equals $\det J$ times the determinant of the Schur complement of $J$, namely
		\[
		R:=\begin{pmatrix}0&\Gamma\\ -\Gamma^{\top}&0\end{pmatrix}
		-\begin{pmatrix}0\\ -B_{u}^{\top}\end{pmatrix}J^{-1}
		\begin{pmatrix}0&B_{u}\end{pmatrix}
		=\begin{pmatrix}0&\Gamma\\ -\Gamma^{\top}&M\end{pmatrix},
		\qquad M=B_{u}^{\top}J^{-1}B_{u}.
		\]
		The residual block $R$ is not literally $C$: the two differ by the sign of the
		off-diagonal blocks. They are congruent,
		\begin{equation}
			C=\begin{pmatrix}-I_{d}&0\\ 0&I_{k}\end{pmatrix}^{\!\top}
			R
			\begin{pmatrix}-I_{d}&0\\ 0&I_{k}\end{pmatrix},
			\label{eq:RtoC}
		\end{equation}
		and the transforming matrix has determinant $(-1)^{d}$, so $\det R=\det C$.
		Finally
		$\det J=\prod_{\ell}\det\bigl(\mu_{\ell}
		\left(\begin{smallmatrix}0&1\\-1&0\end{smallmatrix}\right)\bigr)
		=\prod_{\ell}\mu_{\ell}^{2}$, and the last equality of \eqref{eq:factorization}
		is Corollary~\ref{cor:pfC}.
	\end{proof}
	
	\begin{remark}[Frame consistency]
		\label{rem:frameconsistency}
		The chain \eqref{eq:factorization} is an exact identity of numbers, not an
		identity up to normalization, because a single frame is used throughout: the
		$\mu_{\ell}$ come from $Q$, the matrix $\Gamma=N_{0}^{\top}B$ from the last $d$
		columns of the same $Q$, and $\pi_{\alpha}=N_{0\alpha}^{A}\chi_{A}$ from those
		same columns. Any other frame $N=N_{0}S$ rescales $\Gamma$ by $S^{\top}$ and
		$C$ correspondingly, and the two sides of \eqref{eq:factorization} then differ
		by $\det(S)^{2}$; the invariant content is
		Proposition~\ref{prop:basis} of Section~\ref{sec:strata}.
	\end{remark}
	
	\begin{remark}[What bordering computes]
		\label{rem:whatbordering}
		Two statements can now be made precisely.
		
		\emph{(i)} The bordered determinant is the determinant of the reduced Dirac
		matrix multiplied by the strictly positive factor $\prod_{\ell}\mu_{\ell}^{2}$;
		equivalently, by Corollary~\ref{cor:pfC}, $\det\Gamma$ is, up to sign, the
		Pfaffian of the reduced Dirac matrix. This is why an object of size $d\times k$
		built from gradients can decide a question about a $(d+k)\times(d+k)$ bracket
		matrix.
		
		\emph{(ii)} Bordering does not compute the full reduced Dirac matrix. It
		computes the primary--secondary block $\Gamma$; the secondary--secondary block
		$M$, which is what the first-class analysis requires, is not determined by
		$\Gamma$ alone. Section~\ref{sec:firstclass} shows that the missing block is
		exactly what separates $\ker\Gamma$ from the space of first-class combinations,
		and Section~\ref{sec:weak} shows that the separation can be total.
	\end{remark}
	
	\begin{remark}[Why the conflation went unnoticed]
		\label{rem:k1}
		For $k=1$ every antisymmetric $1\times1$ matrix vanishes, so $M=0$ automatically
		and $C$ is determined by $\Gamma$ alone; the distinction of this section is then
		invisible, as Example~\ref{ex:d1} shows. The phenomenon requires $k\geq2$, where
		$M$ can be invertible while $\Gamma=0$; Example~\ref{ex:counter} is the minimal
		system of that kind.
	\end{remark}
	
	\begin{remark}[Adapted coordinates]
		\label{rem:darboux}
		Since $f^{(0)}$ is closed and, under \textbf{H1}, of locally constant rank, the
		Darboux theorem for presymplectic forms provides local coordinates
		$(q^{a},p_{a},z^{\alpha})$ with $f^{(0)}=\sum_{a}\dd p_{a}\wedge\dd q^{a}$ and
		$\ker f^{(0)}=\spn\{\partial_{z^{1}},\dots,\partial_{z^{d}}\}$. In such
		coordinates the null frame is constant, so the frame-variation term of
		Section~\ref{sec:strata} vanishes and $\Gamma$ reduces to the Hessian block
		\eqref{eq:hessian}. This is a convenient normalization and nothing more: it is
		not used in the proofs of Theorem~\ref{thm:dirac},
		Corollary~\ref{cor:pfC} or Proposition~\ref{prop:factorization}, and it does not
		dispose of \textbf{H3}, which remains the hypothesis under which
		\eqref{eq:hessian} and the symmetry of $\Gamma$ hold in a given frame. Note also
		that the Darboux normalization sets $\mu_{\ell}=1$, so it is not compatible with
		the orthonormal normalization used in \eqref{eq:factorization}; the two must not
		be mixed within one computation.
	\end{remark}
	
	\section{First-class combinations}
	\label{sec:firstclass}
	
	\begin{theorem}[Characterization of first-class combinations]
		\label{thm:firstclass}
		Let $w\in\R^{k}$ have constant entries and put
		$\Om_{w}=\sum_{\beta}w^{\beta}\Om_{\beta}$. No hypothesis on the kernel frame is
		required: Theorem~\ref{thm:dirac} holds for an arbitrary smooth frame. Then
		\begin{enumerate}[label=(\alph*),leftmargin=2.2em,itemsep=0.25em]
			\item $\Gamma w=0\iff\{\pi_{\alpha},\Om_{w}\}=0$ for all $\alpha$, that is,
			$\Om_{w}$ commutes with all primary constraints;
			\item $\Om_{w}$ is first class within the set
			$(\pi_{1},\dots,\pi_{d},\Om_{1},\dots,\Om_{k})$ if and only if
			\begin{equation}
				\Gamma w=0
				\quad\text{and}\quad
				\Pi M\Pi\,w=0 .
				\label{eq:firstclass}
			\end{equation}
			\item The number of independent first-class combinations is
			\begin{equation}
				\dim\{w:\Gamma w=0,\ \Pi M\Pi w=0\}
				=\dim\ker\Gamma-\rank(\Pi M\Pi).
				\label{eq:fccount}
			\end{equation}
		\end{enumerate}
	\end{theorem}
	
	\begin{proof}
		(a) is Theorem~\ref{thm:dirac} together with bilinearity of the bracket in
		constant coefficients. (b) Given $w\in\ker\Gamma$, the requirement
		$\{\Om_{w'},\Om_{w}\}=0$ for every surviving combination $w'$ reads
		$w'^{\top}Mw=0$ for all $w'\in\ker\Gamma$, that is $\Pi Mw=0$; since $\Pi w=w$
		this is $\Pi M\Pi w=0$. (c) is the rank--nullity theorem applied to the
		endomorphism $\Pi M\Pi$ of $\ker\Gamma$.
	\end{proof}
	
	\begin{corollary}[Reading of the dimension formula]
		\label{cor:reading}
		Combining Corollary~\ref{cor:dim}, Proposition~\ref{prop:mfree} and
		Theorem~\ref{thm:firstclass},
		\begin{equation}
			\dim\ker f^{(m)}
			=\underbrace{\#\{\text{null directions tangent to }\Sigma\}}
			_{\dim\ker\Gamma^{\top}}
			+\underbrace{\#\{\text{first-class combinations}\}}
			_{\dim\ker\Gamma-\rank(\Pi M\Pi)} .
			\label{eq:reading}
		\end{equation}
		Only the second summand produces kernel elements with nonvanishing multiplier
		block, and only those can be links of a nontrivial gauge chain
		(Section~\ref{sec:chains}).
	\end{corollary}
	
	\begin{remark}[What \eqref{eq:reading} is and is not]
		\label{rem:dimonly}
		Identity \eqref{eq:reading} is a decomposition of the dimension of
		$\ker f^{(m)}$ into two intrinsically defined contributions, not a
		decomposition of $\ker f^{(m)}$ into a canonical direct sum. The
		multiplier-free directions do form a canonical subspace, by
		Proposition~\ref{prop:mfree}, but no canonical complement to it inside
		$\ker f^{(m)}$ is constructed here, and none is needed for any of the
		statements that follow.
	\end{remark}
	
	\begin{remark}[The slogan and its correction]
		The informal statement ``the number of gauge generators equals the number of
		first-class constraints'' fails on both counts: the kernel is larger than the
		first-class count by $\dim\ker\Gamma^{\top}$, and the first-class count itself is
		smaller than $\dim\ker\Gamma$ by $\rank(\Pi M\Pi)$. Equality
		$\dim\ker f^{(m)}=2\dim\ker\Gamma$ holds only when $\Gamma$ is square and
		$\Pi M\Pi=0$.
	\end{remark}
	
	\section{Gauge-generator chains}
	\label{sec:chains}
	
	\subsection{The general chain condition}
	
	\begin{theorem}[Chain condition]
		\label{thm:chain}
		Let $u_{0},\dots,u_{S}$ be vector fields and
		$\delta\xi^{A}=\sum_{s=0}^{S}\rho^{(s)}(t)u_{s}^{A}(\xi)$. Then $\delta S$
		reduces to a boundary term for every compactly supported $\rho$ and along every
		curve if and only if
		\begin{equation}
			\sum_{s=0}^{S}(-1)^{s}\frac{\dd^{s}}{\dd t^{s}}A_{u_{s}}\equiv0
			\label{eq:chaincond}
		\end{equation}
		as an identity on the jet variables $(\xi,\dot\xi,\ddot\xi,\dots)$.
	\end{theorem}
	
	\begin{proof}
		By linearity of \eqref{eq:master} in $u$ and $\eta$,
		$\delta S=-\int\sum_{s}\rho^{(s)}A_{u_{s}}+\text{boundary}$. Integrating the
		$s$-th term by parts $s$ times gives
		$\delta S=-\int\rho\sum_{s}(-1)^{s}\dd^{s}A_{u_{s}}/\dd t^{s}+\text{boundary}$,
		and the fundamental lemma applies along each curve; since curves realize
		arbitrary jets, \eqref{eq:chaincond} is an identity.
	\end{proof}
	
	\subsection{Two-step chains}
	
	\begin{theorem}[Two-step chain]
		\label{thm:chain2}
		Let $u_{1}\in\ker f^{(0)}$ be a null mode with contraction
		$\Phi_{u_{1}}=u_{1}\cdot\nabla V$. Then
		$\delta\xi=\rho\,u_{0}+\dot\rho\,u_{1}$ is a symmetry of the action
		\eqref{eq:lag} for every $\rho$ if and only if
		\begin{equation}
			f^{(0)}u_{0}=\nabla\Phi_{u_{1}}
			\qquad\text{and}\qquad
			u_{0}\cdot\nabla V=0 .
			\label{eq:chain2}
		\end{equation}
		Moreover:
		\begin{enumerate}[label=(\alph*),leftmargin=2.2em,itemsep=0.25em]
			\item the first equation is solvable in $u_{0}$ if and only if
			$N^{\top}\nabla\Phi_{u_{1}}=0$; if $\Phi_{u_{1}}=\Om_{b}$ this is the
			vanishing of the $b$-th column of $\Gamma$, that is $e_{b}\in\ker\Gamma$;
			\item the solution is unique modulo $\ker f^{(0)}$, and under
			$u_{0}\mapsto u_{0}+\sum_{a}c^{a}N_{a}$ with $c^{a}\in\Rring$ the function
			$u_{0}\cdot\nabla V$ changes by $\sum_{a}c^{a}\Om_{a}\in\Ical$.
		\end{enumerate}
	\end{theorem}
	
	\begin{proof}
		Apply Theorem~\ref{thm:chain} with $S=1$: the condition is
		$A_{u_{0}}-\tfrac{\dd}{\dd t}A_{u_{1}}=0$. Since $f^{(0)}u_{1}=0$ identically,
		$A_{u_{1}}=\Phi_{u_{1}}(\xi)$ depends on $\xi$ alone, so
		$\tfrac{\dd}{\dd t}A_{u_{1}}=\dot\xi^{i}\pd{i}\Phi_{u_{1}}$; and
		$A_{u_{0}}=\dot\xi^{i}(f^{(0)}u_{0})_{i}+u_{0}\cdot\nabla V$. Separating the
		part linear in $\dot\xi$ from the rest gives \eqref{eq:chain2}. (a) For an
		antisymmetric matrix $\im f^{(0)}=(\ker f^{(0)})^{\perp}$, so solvability is
		$N_{a}\cdot\nabla\Phi=0$ for all $a$; with $\Phi=\Om_{b}$ this is
		$\Gamma_{ab}=0$ for all $a$. (b) The indeterminacy of a consistent linear system
		is its kernel, and $N_{a}\cdot\nabla V=\Om_{a}$ by construction of the
		constraints.
	\end{proof}
	
	The following is the decisive statement of the paper.
	
	\begin{theorem}[Chain solvability criterion]
		\label{thm:chain-decisive}
		Let $u_{1}=N_{b}$ with $\Phi_{u_{1}}=\Om_{b}$ and assume $e_{b}\in\ker\Gamma$,
		so that the first equation of \eqref{eq:chain2} has a solution $u_{0}$. Then the
		class
		\begin{equation}
			\bigl[\,u_{0}\cdot\nabla V\,\bigr]\in\Rring/\Ical
			\label{eq:class}
		\end{equation}
		is independent of the choice of $u_{0}$, and the two-step chain closes into an
		exact symmetry of the action if and only if that class vanishes. Explicitly,
		\[
		\exists\,u_{0}\ \text{with}\ f^{(0)}u_{0}=\nabla\Om_{b}
		\ \text{and}\ u_{0}\cdot\nabla V=0
		\iff
		e_{b}\in\ker\Gamma
		\ \text{and}\
		\bigl[\,u_{0}\cdot\nabla V\,\bigr]=0\ \text{in}\ \Rring/\Ical .
		\]
	\end{theorem}
	
	\begin{proof}
		Well-definedness of the class is Theorem~\ref{thm:chain2}(b). If the class
		vanishes, write $u_{0}\cdot\nabla V=\sum_{a}c^{a}\Om_{a}$ with
		$c^{a}\in\Rring$ and replace $u_{0}$ by $u_{0}-\sum_{a}c^{a}N_{a}$; this remains
		a solution of the first equation because $f^{(0)}N_{a}=0$ pointwise, and now the
		second equation holds exactly. Conversely, if some $u_{0}$ satisfies both
		equations then its class is $0$, and by (b) so is that of any other solution.
	\end{proof}
	
	\subsection{Chains and the bordered kernel}
	
	\begin{theorem}[Exact relation to the kernel]
		\label{thm:chainkernel}
		Let $(u_{0},w)\in\ker f^{(m)}$ with $w=-e_{b}$ in the convention \eqref{eq:fm}.
		Then $f^{(0)}u_{0}=-Bw=\nabla\Om_{b}$, so $u_{0}$ solves the first chain
		equation with $u_{1}=N_{b}$; in addition $B^{\top}u_{0}=0$. The two-step chain
		$\delta\xi=\rho\,u_{0}+\dot\rho\,N_{b}$ is an exact symmetry of the original
		action if and only if
		\begin{equation}
			\Phi_{u_{0}}=u_{0}\cdot\nabla V\equiv0,
			\label{eq:chainkernel}
		\end{equation}
		that is, if and only if $(u_{0},w)$ is a \emph{strong} null mode. Conversely, a
		solution $u_{0}$ of the chain equations need not satisfy $B^{\top}u_{0}=0$, so
		the inclusion
		\[
		\{\text{kernel elements with }w=-e_{b}\}
		\subsetneq
		\{\text{solutions of the first chain equation}\}
		\]
		is strict in general.
	\end{theorem}
	
	\begin{proof}
		The physical block of $f^{(m)}(u_{0},w)=0$ is $f^{(0)}u_{0}+Bw=0$, and the
		multiplier block is $B^{\top}u_{0}=0$. Theorem~\ref{thm:chain2} then requires
		exactly the second equation, which by Lemma~\ref{lem:Vlambda} is the vanishing
		of the contraction computed by the algorithm. Strictness follows because
		$B^{\top}u_{0}=0$ is $k$ extra linear conditions not implied by
		$f^{(0)}u_{0}=\nabla\Om_{b}$.
	\end{proof}
	
	\begin{remark}[Terminology]
		\label{rem:terminology}
		We call $u_{0}$ a \emph{chain element} and reserve \emph{gauge generator} for
		the full transformation $\delta\xi=\rho u_{0}+\dot\rho u_{1}$ once
		Theorem~\ref{thm:chain-decisive} has been verified. A \emph{null mode} is an
		element of $\ker f^{(m)}$; a \emph{strong gauge direction} is a null mode with
		identically vanishing contraction, which by Theorem~\ref{thm:strong} generates a
		symmetry on its own. These four notions are distinct and are never used
		interchangeably below.
	\end{remark}
	
	\begin{remark}[Longer chains]
		\label{rem:longer}
		Theorem~\ref{thm:chain} covers chains of arbitrary length. The systematic
		analysis of $S\geq2$ requires an inductive solvability tower whose obstruction
		at each level is again a class in $\Rring/\Ical$; we have verified the pattern
		for $S=2$ in examples but do not have a general theorem, and we flag this as the
		principal open point of the present work
		(Section~\ref{sec:scope}, item~\ref{open:chains}).
	\end{remark}
	
	\section{Why weak modes are not gauge}
	\label{sec:weak}
	
	The halting criterion of the iterative algorithm is that every contraction lies
	in the constraint ideal. This section shows that the criterion is necessary but
	not sufficient for gauge freedom, by means of a counterexample that we regard as
	a structural result rather than an illustration.
	
	\begin{example}[A four-variable system with no gauge freedom]
		\label{ex:counter}
		Let $\xi=(q_{1},q_{2},q_{3},q_{4})$ and
		\begin{equation}
			L=q_{2}\dot q_{1}-q_{1}q_{3}-q_{2}q_{4},
			\qquad\text{so}\qquad
			a=(q_{2},0,0,0),\quad V=q_{1}q_{3}+q_{2}q_{4}.
			\label{eq:counterL}
		\end{equation}
		Then
		\[
		f^{(0)}=\begin{pmatrix}0&-1&0&0\\1&0&0&0\\0&0&0&0\\0&0&0&0\end{pmatrix},
		\qquad
		\ker f^{(0)}=\spn\{e_{3},e_{4}\},
		\qquad
		\Om_{1}=\pd{3}V=q_{1},\quad\Om_{2}=\pd{4}V=q_{2}.
		\]
		Hence $B=(e_{1}\ e_{2})$ and
		\[
		\Gamma=N^{\top}B=0_{2\times2},
		\qquad
		M=B_{u}^{\top}J^{-1}B_{u}=J^{-1}\ \text{invertible},
		\qquad
		\Pi=I_{2},\quad\rank(\Pi M\Pi)=2 .
		\]
		Formula \eqref{eq:dim} gives $\dim\ker f^{(1)}=2+2-2=2$, with basis
		$\{(e_{3},0),(e_{4},0)\}$; by Proposition~\ref{prop:mfree} both are
		multiplier-free, and by \eqref{eq:fccount} there is \emph{no} first-class
		combination. The contractions are
		\[
		\Phi_{e_{3}}=\pd{3}V=q_{1}=\Om_{1},
		\qquad
		\Phi_{e_{4}}=\pd{4}V=q_{2}=\Om_{2},
		\]
		both in $\Ical$ and neither identically zero: the algorithm halts on weakly
		vanishing contractions and reports two null modes.
		
		The Euler--Lagrange equations of \eqref{eq:counterL} are, however,
		\[
		\dot q_{2}=-q_{3},\qquad
		\dot q_{1}=q_{4},\qquad
		q_{1}=0,\qquad
		q_{2}=0,
		\]
		so $q_{1}\equiv q_{2}\equiv0$ and, differentiating, $q_{3}=q_{4}=0$. The
		solution is unique: the system has zero degrees of freedom and admits no gauge
		transformation whatsoever. Consistently with Theorem~\ref{thm:firstclass}, the
		two constraints are second class, $\{\Om_{1},\Om_{2}\}_{\FJ}=M_{12}\neq0$.
	\end{example}
	
	\begin{theorem}[The chain criterion rejects the false candidates]
		\label{thm:counterchain}
		In Example~\ref{ex:counter}, take $u_{1}=e_{3}$, so $\Phi_{u_{1}}=\Om_{1}=q_{1}$.
		The first chain equation $f^{(0)}u_{0}=\nabla\Om_{1}=e_{1}$ has general solution
		$u_{0}=(0,-1,c_{3},c_{4})$ with $c_{3},c_{4}\in\Rring$ arbitrary. The decisive
		class of Theorem~\ref{thm:chain-decisive} is
		\begin{equation}
			\bigl[\,u_{0}\cdot\nabla V\,\bigr]
			=\bigl[-\pd{2}V+c_{3}\pd{3}V+c_{4}\pd{4}V\bigr]
			=\bigl[-q_{4}+c_{3}q_{1}+c_{4}q_{2}\bigr]
			=-q_{4}\neq0
			\quad\text{in }\Rring/\gen{q_{1},q_{2}},
			\label{eq:counterclass}
		\end{equation}
		so no two-step chain closes. The same computation with $u_{1}=e_{4}$ gives
		$[\,u_{0}\cdot\nabla V\,]=q_{3}\neq0$.
	\end{theorem}
	
	\begin{remark}[Structural diagnosis]
		\label{rem:diagnosis}
		The mechanism is transparent: in \eqref{eq:counterL} the potential depends on the
		degenerate coordinates only through the constraints themselves,
		$V=q_{3}\Om_{1}+q_{4}\Om_{2}$. Every contraction is then automatically in
		$\Ical$, so the ideal test cannot distinguish this situation from genuine gauge
		freedom. What the bordering does is to insert $\dot\lambda^{\alpha}$ into the
		equations of motion, absorbing precisely the terms that determined $q_{3}$ and
		$q_{4}$ in the original system; the resulting freedom belongs to the extended
		system, not to the original one. This is
		Remark~\ref{rem:borderimposes} in action.
	\end{remark}
	
	\begin{proposition}[General form of the mechanism]
		\label{prop:mechanism}
		Suppose that in adapted coordinates (Remark~\ref{rem:darboux}) the potential has
		the form $V=V_{0}(q,p)+\sum_{\alpha}z^{\alpha}\Om_{\alpha}(q,p)$ with
		$\Om_{\alpha}=\pd{z^{\alpha}}V$ independent of $z$. Then every contraction
		$\Phi_{N_{\alpha}}=\Om_{\alpha}$ lies in $\Ical$ and the algorithm halts, while
		$\Gamma=0$ and $M_{\alpha\beta}=\{\Om_{\alpha},\Om_{\beta}\}_{\FJ}$. There is
		gauge freedom if and only if $\rank(M)<k$, and the number of first-class
		combinations is $k-\rank M$.
	\end{proposition}
	
	\begin{proof}
		$\Gamma_{\alpha\beta}=\pd{z^{\alpha}}\Om_{\beta}=0$ since the $\Om_{\beta}$ do
		not depend on $z$; $\Pi=I$ so $\Pi M\Pi=M$, and \eqref{eq:fccount} gives
		$k-\rank M$.
	\end{proof}
	
	\begin{corollary}[Necessary but not sufficient]
		\label{cor:necsuf}
		Membership of all contractions in $\Ical$, which is the criterion on which the
		iteration halts, is a necessary condition for the existence of gauge freedom,
		but not a sufficient one. The sufficient test is
		Theorem~\ref{thm:chain-decisive}.
	\end{corollary}
	
	\section{Involutivity and the reduced gauge distribution}
	\label{sec:involutivity}
	
	\begin{proposition}[Closedness]
		\label{prop:closed}
		$\dd f^{(m)}=0$ for every $m$, since $f^{(m)}=\dd a^{(m)}$.
	\end{proposition}
	
	\begin{theorem}[Involutivity of the kernel]
		\label{thm:frobenius}
		Let $f$ be a closed $2$-form of locally constant rank on an open set $U$. Then
		$K=\ker f$ is an involutive distribution and, by Frobenius, integrable.
	\end{theorem}
	
	\begin{proof}
		For sections $u,v$ of $K$ and arbitrary $w$, the intrinsic formula
		\[
		\dd f(u,v,w)=u\,f(v,w)-v\,f(u,w)+w\,f(u,v)
		-f([u,v],w)+f([u,w],v)-f([v,w],u)
		\]
		has all terms zero except $-f([u,v],w)$, because $f(u,\cdot)=f(v,\cdot)=0$. With
		$\dd f=0$ we get $f([u,v],w)=0$ for all $w$, i.e. $[u,v]\in K$. Constant rank
		makes $K$ a subbundle with smooth local frames.
	\end{proof}
	
	\begin{proposition}[Strong gauge distribution]
		\label{prop:stronginv}
		$\mathcal{G}_{\mathrm{strong}}=\ker f\cap\ker\dd V$ is involutive. If in addition
		its rank is locally constant, it is integrable.
	\end{proposition}
	
	\begin{proof}
		$[u,v]\in\ker f$ by Theorem~\ref{thm:frobenius}, and
		$[u,v]V=u(vV)-v(uV)=0$. The rank of an intersection of two distributions can
		jump even when each has constant rank, whence the extra hypothesis.
	\end{proof}
	
	\begin{lemma}[Kernel of the pullback]
		\label{lem:pullback}
		Let $\iota:\Sigma\hookrightarrow\Mm^{(m)}$ be the inclusion of a regular
		constraint surface, $T_{\xi}\Sigma=\{v:v\cdot\nabla\Om_{\alpha}=0\ \forall\alpha\}$.
		Then
		\begin{equation}
			\ker\bigl(\iota^{*}f^{(m)}\bigr)_{\xi}
			=\{v\in T_{\xi}\Sigma:f^{(m)}v\in\spn\{\nabla\Om_{1},\dots,\nabla\Om_{k}\}\}
			\supseteq\ker f^{(m)},
			\label{eq:pullback}
		\end{equation}
		the inclusion being strict in general.
	\end{lemma}
	
	\begin{proof}
		$v\in\ker(\iota^{*}f)$ iff $v\in T\Sigma$ and $f(v,z)=0$ for all $z\in T\Sigma$,
		i.e. the covector $fv$ annihilates $T\Sigma$, i.e. lies in $(T\Sigma)^{0}$,
		generated by the $\nabla\Om_{\alpha}$ under regularity. The inclusion uses
		Remark~\ref{rem:tangency}: every kernel element is already tangent.
	\end{proof}
	
	\begin{remark}[Ambient versus restricted, and the Gotay--Nester surface]
		\label{rem:gotay}
		The number of null modes reported by the algorithm is a lower bound for
		$\dim\ker(\iota^{*}f^{(m)})$. It is tempting to define the physical space as the
		leaf space $\Mm^{*}=\Sigma/\ker(\iota^{*}f^{(m)})$ and to count
		$\dim\Mm^{*}=\dim\Sigma-\dim\ker(\iota^{*}f^{(m)})$. This is legitimate only
		when $\Sigma$ is the final constraint manifold of the Gotay--Nester--Hinds
		algorithm \cite{GNH1978}, that is, when the dynamical vector field is tangent to
		$\Sigma$ everywhere, and when the kernel has locally constant rank. In
		Example~\ref{ex:counter} the quotient identifies points that are not dynamically
		equivalent, and the correct answer (zero degrees of freedom) is obtained only by
		accident. We therefore state all counting formulas conditionally
		(Section~\ref{sec:charges}).
	\end{remark}
	
	\section{Generic kernels, degeneracy strata and basis dependence}
	\label{sec:strata}
	
	Symbolic kernels are computed over a field of rational functions, which
	introduces three logically independent qualifications.
	
	\subsection{Genericity and strata}
	
	\begin{theorem}[Rank stratification]
		\label{thm:strat}
		Let $f(\xi)$ be antisymmetric with entries in $\Rring$ and let
		$r_{\mathrm{gen}}$ be its rank over the field of fractions. Then:
		\begin{enumerate}[label=(\alph*),leftmargin=2.2em,itemsep=0.25em]
			\item $\Delta=\{\xi:\rank f(\xi)<r_{\mathrm{gen}}\}$ is the common zero locus
			of the numerators of the maximal minors, hence Zariski closed;
			\item on the dense open set $\Mm\setminus(\Delta\cup P)$, with $P$ the pole
			locus of the rational frame, $\ker f$ is a subbundle of constant rank
			$d=n-r_{\mathrm{gen}}$;
			\item on $\Delta$ the kernel can be larger, and gauge directions invisible to
			the generic frame may appear.
		\end{enumerate}
	\end{theorem}
	
	\begin{proof}
		(a) is the minor criterion for rank; (b) follows by Cramer's rule from a fixed
		nonvanishing minor; (c) is lower semicontinuity of the rank.
	\end{proof}
	
	The sets $\Delta$ and $P$ are of different nature and must not be merged: $P$ is
	an artifact of the chosen frame and can be removed by clearing denominators,
	while $\Delta$ is intrinsic.
	
	\begin{example}[An invisible rank jump]
		\label{ex:jump}
		Let $n=2$, $a=(0,\tfrac12 q_{1}^{2})$, i.e.
		$L=\tfrac12 q_{1}^{2}\dot q_{2}-V$. Then
		$f^{(0)}=\bigl(\begin{smallmatrix}0&q_{1}\\-q_{1}&0\end{smallmatrix}\bigr)$ has
		generic rank $2$ and empty kernel over the field of fractions, with bracket
		$\{q_{1},q_{2}\}=1/q_{1}$. On $\Delta=\{q_{1}=0\}$, however, $f^{(0)}=0$ and
		every vector is null. The verdict ``regular'' must be read as ``regular off the
		degeneracy locus''; hypothesis \textbf{H1} is exactly the exclusion of this
		phenomenon.
	\end{example}
	
	\subsection{Non-orthonormal frames}
	
	\begin{proposition}[Change of kernel basis]
		\label{prop:basis}
		Let $N_{0}$ be orthonormal and $N=N_{0}S$ with $S\in GL(d,\R)$. Then
		$\Gamma=S^{\top}\Gamma_{0}$, $\det(N^{\top}N)=\det(S)^{2}$, and for $k=d$
		\begin{equation}
			\det f^{(m)}
			=\Bigl(\prod_{\ell}\mu_{\ell}^{2}\Bigr)\det(\Gamma_{0})^{2}
			=\Bigl(\prod_{\ell}\mu_{\ell}^{2}\Bigr)
			\frac{\det(\Gamma)^{2}}{\det(N^{\top}N)} .
			\label{eq:basiscorr}
		\end{equation}
		The criterion $\det\Gamma\neq0$ is basis-independent; the value $\det\Gamma$ is
		not.
	\end{proposition}
	
	\begin{proof}
		$\Gamma=(N_{0}S)^{\top}B=S^{\top}\Gamma_{0}$ and
		$N^{\top}N=S^{\top}N_{0}^{\top}N_{0}S=S^{\top}S$. Substitute
		$\det\Gamma_{0}=\det\Gamma/\det S$ into Proposition~\ref{prop:factorization}.
	\end{proof}
	
	\subsection{Non-constant frames}
	
	\begin{proposition}[Kernel--Hessian identity with correction]
		\label{prop:Tterm}
		Without assuming that $N$ is constant,
		\begin{equation}
			\Gamma_{\alpha\beta}
			=\sum_{i,j}N_{\alpha}^{i}\bigl(\pd{i}\pd{j}V\bigr)N_{\beta}^{j}
			+T_{\alpha\beta},
			\qquad
			T_{\alpha\beta}:=\sum_{i,j}N_{\alpha}^{i}\bigl(\pd{i}N_{\beta}^{j}\bigr)\pd{j}V .
			\label{eq:Tterm}
		\end{equation}
		Moreover $T\equiv0$ if $N$ is locally constant, in particular in the adapted
		coordinates of Remark~\ref{rem:darboux}, and $\Gamma$ is symmetric if and only
		if the antisymmetric part of $T$ vanishes.
	\end{proposition}
	
	\begin{proof}
		Product rule in $\Gamma_{\alpha\beta}=N_{\alpha}^{i}\pd{i}
		(N_{\beta}^{j}\pd{j}V)$; the Hessian term is symmetric.
	\end{proof}
	
	\begin{remark}
		On $\Sigma$ one has $N^{\top}\nabla V=(\Om_{1},\dots,\Om_{d})^{\top}=0$, so the
		component of $\nabla V$ along the kernel drops out of $T$, but the component
		along $\im f^{(0)}$ does not: $T|_{\Sigma}$ need not vanish. Identity
		\eqref{eq:Tterm} is valid in an arbitrary frame, and $T$ measures the failure of
		that frame to be constant; by Remark~\ref{rem:darboux} it can always be made to
		vanish by passing to adapted coordinates. This does not dispose of \textbf{H3}:
		the hypothesis is what licenses the Hessian representation \eqref{eq:hessian}
		and the symmetry of $\Gamma$ \emph{in the frame one is actually using}, and it
		is needed wherever those two properties are invoked. It is not needed for the
		block structure of $C$, which by Theorem~\ref{thm:dirac} holds in any frame, nor
		for Corollary~\ref{cor:pfC}, which by Remark~\ref{rem:nosym} does not require
		$\Gamma$ to be symmetric.
	\end{remark}
	
	\section{Reducibility, radicality and real inconsistency}
	\label{sec:reducibility}
	
	\subsection{Reducible borderings}
	
	\begin{proposition}[Spurious strong modes]
		\label{prop:spurious}
		If $\rank B<k$ there is $w\neq0$ with $Bw=0$, and $(0,w)\in\ker f^{(m)}$ with
		$\Phi_{(0,w)}\equiv0$. Dependent constraint gradients therefore manufacture
		strong null modes with no dynamical content.
	\end{proposition}
	
	\begin{proof}
		$f^{(m)}(0,w)^{\top}=(Bw,0)^{\top}=0$, and the contraction uses only the
		physical block, which vanishes.
	\end{proof}
	
	\begin{proposition}[Ideal novelty prevents reducibility]
		\label{prop:novelty}
		Let $\Om_{1},\Om_{2}$ have proportional gradients on a neighbourhood $U$ with
		$\nabla\Om_{1}\neq0$, and suppose $U\cap\{\Om_{1}=0\}\neq\emptyset$. Then
		$\Om_{2}\in\gen{\Om_{1}}$ on $U$, so a normal-form novelty test discards
		$\Om_{2}$ and the bordering stays irreducible.
	\end{proposition}
	
	\begin{proof}
		Complete $\Om_{1}$ to a coordinate system $(\Om_{1},z^{2},\dots,z^{n})$. The
		proportionality gives $\partial\Om_{2}/\partial z^{s}=0$ for $s\geq2$, so
		$\Om_{2}=F(\Om_{1})$. At a point of $U$ where both $\Om_{1}$ and $\Om_{2}$
		vanish one gets $F(0)=0$, and Hadamard's lemma yields
		$F(\tau)=\tau G(\tau)$.
	\end{proof}
	
	\begin{remark}
		Radicality of $\Ical$ plays no role in Proposition~\ref{prop:novelty}; the
		hypotheses actually used are the nonvanishing gradient and the presence of a
		zero of $\Om_{1}$ in the neighbourhood considered. We state it this way to avoid
		carrying an unused assumption.
	\end{remark}
	
	\subsection{Ideal versus vanishing on the surface}
	
	\begin{proposition}[Conservative character of the ideal test]
		\label{prop:radical}
		$\Ical\subseteq\Ical(\Sigma):=\{h:h|_{\Sigma}=0\}$, with strict inclusion in
		general: for $\Om=q^{2}$ one has $q\in\Ical(\Sigma)\setminus\gen{q^{2}}$.
		Consequently $\Phi_{v}\in\Ical$ is sufficient but not necessary for
		$\Phi_{v}|_{\Sigma}=0$; the algorithmic weak verdict is conservative.
	\end{proposition}
	
	\begin{theorem}[Corrected trichotomy at halting]
		\label{thm:trichotomy}
		Consider a given stage of the iteration and let $v\in\ker f^{(m)}$. Exactly
		one of the following holds for the contraction $\Phi_{v}$, and the
		iteration may be declared terminated only when every $v$ in a frame of
		$\ker f^{(m)}$ falls under case~(1) or case~(2).
		\begin{enumerate}[label=(\arabic*),leftmargin=2.2em,itemsep=0.25em]
			\item $\Phi_{v}\equiv0$ on an open set: strong null mode.
			\item $\Phi_{v}\not\equiv0$ and $\Phi_{v}|_{\Sigma}=0$, i.e.
			$\Phi_{v}\in\Ical(\Sigma)$: weak null mode.
			\item $\Phi_{v}|_{\Sigma}\not\equiv0$: the iteration is not yet
			terminated. Either $\Phi_{v}$ has no zero on $\Sigma$, in which case no
			admissible point survives and the system is inconsistent, or
			$\Sigma\cap\{\Phi_{v}=0\}$ is a nonempty proper subset of $\Sigma$, and one
			more bordering round follows with that reduced surface.
		\end{enumerate}
		Testing $\Phi_{v}\in\Ical$ instead of $\Phi_{v}\in\Ical(\Sigma)$ is the
		algebraically effective substitute for case~(2), and by
		Proposition~\ref{prop:radical} it is conservative: it may classify a case~(2) as
		a case~(3), never the converse.
	\end{theorem}
	
	\subsection{Real inconsistency}
	
	\begin{proposition}[Emptiness over $\R$ is not $1\in\Ical$]
		\label{prop:real}
		Over $\C$, $\Sigma=\emptyset\iff1\in\Ical$ by the weak Nullstellensatz. Over
		$\R$ the forward implication fails. For $L=q_{2}\dot q_{1}-(q_{1}^{2}+1)q_{3}$
		the generated constraint is $\Om=q_{1}^{2}+1$, which has no real zero, so the
		system is inconsistent; yet $1\notin\gen{q_{1}^{2}+1}$ in $\R[q]$.
	\end{proposition}
	
	\begin{remark}[Effective consequence]
		\label{rem:cad}
		The complete test over $\R$ is the emptiness of the real variety
		$\Sigma\subseteq\R^{n}$, which is decidable by cylindrical algebraic
		decomposition. The Gr\"obner condition $1\in\Ical$ remains a sufficient
		certificate of inconsistency and is far cheaper, so it is naturally used first,
		the doubly exponential real test being needed only when it is inconclusive. The exact algebraic characterization is the Real
		Nullstellensatz: $\Sigma=\emptyset$ if and only if $-1$ is a sum of squares
		modulo $\Ical$ \cite{BCR1998}.
	\end{remark}
	
	\subsection{Constant versus functional coefficients}
	
	\begin{proposition}[Module structure]
		\label{prop:syzygy}
		The set $\mathcal{G}=\{v\in\ker f^{(m)}:\Phi_{v}\in\Ical\}$ is a module over
		$\Rring$, whereas the constant-coefficient analysis of a fixed frame reaches
		only the $\R$-span $\mathcal{G}_{\R}$ of those combinations of that frame whose
		contractions lie in $\Ical$. The inclusion
		$\mathcal{G}_{\R}\subseteq\mathcal{G}$ can be strict: if a frame yields
		$\Phi_{v_{1}}=q_{1}$, $\Phi_{v_{2}}=q_{2}$ and $\Ical=\gen{q_{1}q_{2}}$, no
		nonzero constant combination lies in $\Ical$, while
		$\Phi_{q_{2}v_{1}}=q_{1}q_{2}\in\Ical$.
	\end{proposition}
	
	\begin{proof}
		$hv\in\ker f^{(m)}$ and $\Phi_{hv}=h\Phi_{v}$ for $h\in\Rring$, so
		$\mathcal{G}$ is a module. Determining the functional coefficients requires the
		syzygy module of $(\Phi_{v_{1}},\dots,\Phi_{v_{s}})$ modulo $\Ical$.
	\end{proof}
	
	\section{Canonical charges and degree-of-freedom counting}
	\label{sec:charges}
	
	\begin{proposition}[The presymplectic form generates no charge]
		\label{prop:nocharge}
		If $v\in\ker f^{(m)}$ then $\iota_{v}f^{(m)}=0$, so an equation
		$\iota_{v}f^{(m)}=\dd G$ forces $\dd G=0$ and $G$ locally constant.
	\end{proposition}
	
	Charges therefore require an auxiliary canonical structure, declared by the
	user. The charge constructed below is accordingly not an intrinsic output of
	the Faddeev--Jackiw formalism: it is a canonical representative obtained after
	such a structure has been declared, it depends on which pairs are declared, as
	Remark~\ref{rem:pairs} shows, and it does not modify the presymplectic
	statement of Proposition~\ref{prop:nocharge}. The data required are a set of
	Darboux pairs $q^{i}\leftrightarrow p_{i}$ with
	$\omega_{\mathrm{can}}=\sum_{i}\dd q^{i}\wedge\dd p_{i}$ and
	\begin{equation}
		\alpha:=\iota_{v}\omega_{\mathrm{can}}
		=\sum_{i}\bigl(v^{q^{i}}\dd p_{i}-v^{p_{i}}\dd q^{i}\bigr).
		\label{eq:alpha}
	\end{equation}
	
	\begin{lemma}[Poincar\'e homotopy]
		\label{lem:poincare}
		If $U\subseteq\R^{n}$ is star-shaped about the origin and $\alpha$ is a closed
		$C^{1}$ one-form on $U$, then
		$G(\xi)=\int_{0}^{1}\alpha_{A}(t\xi)\xi^{A}\,\dd t$ satisfies $\dd G=\alpha$.
	\end{lemma}
	
	\begin{proof}
		Differentiating under the integral,
		$\pd{B}G=\int_{0}^{1}[\alpha_{B}(t\xi)+t\xi^{A}(\pd{B}\alpha_{A})(t\xi)]\dd t$;
		closedness turns $\pd{B}\alpha_{A}$ into $\pd{A}\alpha_{B}$, and the integrand
		becomes $\tfrac{\dd}{\dd t}[t\alpha_{B}(t\xi)]$.
	\end{proof}
	
	\begin{proposition}[Constant directions]
		If $v$ is constant then $\alpha$ is constant and
		$G=\sum_{i}(v^{q^{i}}p_{i}-v^{p_{i}}q^{i})$.
	\end{proposition}
	
	\begin{remark}[Three obstructions]
		\label{rem:obstructions}
		The construction of $G$ is subject to: exactness (for point-dependent
		directions $\iota_{v}\omega_{\mathrm{can}}$ may fail to be closed); star-shaped
		domain (angular variables violate it, and the formula then yields only a local
		primitive); and topology (if $H^{1}\neq0$ a closed form may not be globally
		exact, and the charge exists only as a multivalued function). None of these
		invalidates the construction; all must be declared. The charge is an auxiliary
		object attached to a user-declared canonical structure, not an intrinsic output
		of the Faddeev--Jackiw formalism.
	\end{remark}
	
	\subsection{Trivial transformations and counting}
	
	\begin{definition}[Trivial gauge transformation]
		A transformation $\delta\xi^{A}=\mu^{AB}E_{B}$ with $\mu^{AB}=-\mu^{BA}$ is a
		symmetry of any action, since by Lemma~\ref{lem:firstvar}
		$\delta S=\int E_{A}\mu^{AB}E_{B}=0$. Such transformations vanish on shell and
		carry no physical content.
	\end{definition}
	
	\begin{proposition}[Conditional counting]
		\label{prop:counting}
		Assume \textbf{H1}, \textbf{H2}, \textbf{H4}, \textbf{H10} and that $\Sigma$ is
		the final constraint manifold in the sense of Remark~\ref{rem:gotay}. Then the
		reduced space $\Mm^{*}=\Sigma/\ker(\iota^{*}f^{(m)})$, where defined, is a
		symplectic manifold and
		\begin{equation}
			\dim\Mm^{*}=\dim\Sigma-\dim\ker\bigl(\iota^{*}f^{(m)}\bigr),
			\label{eq:count}
		\end{equation}
		necessarily even. Without those hypotheses \eqref{eq:count} is not a valid count
		of physical degrees of freedom, as Example~\ref{ex:counter} shows.
	\end{proposition}
	
	\section{Worked examples}
	\label{sec:examples}
	
	We use a single family to exhibit the four regimes, then the counterexample of
	Section~\ref{sec:weak}, then a rank-jump example, and finally a system in which
	the decisive locus lies in the space of parameters rather than in the space of
	states. Each is presented as a unit test with the same data: $f$, $\ker f$,
	$\Om$, $\Gamma$, $M$, $\ker f^{(m)}$, classification.
	
	\subsection{A family with four regimes: the elementary parametric stratification}
	\label{sec:family}
	
	Let $\xi=(q_{1},q_{2},q_{3})$, $a=(q_{2},0,0)$ and
	\begin{equation}
		V=W(q_{1},q_{2})+c\,q_{3}+\tfrac{m}{2}q_{3}^{2},
		\label{eq:family}
	\end{equation}
	so that
	$f^{(0)}=\bigl(\begin{smallmatrix}0&-1&0\\1&0&0\\0&0&0\end{smallmatrix}\bigr)$,
	$\ker f^{(0)}=\spn\{e_{3}\}$, $\Om=\pd{3}V=c+mq_{3}$, and
	$\Gamma=(m)$, $M=(0)$ since $k=1$.
	
	\begin{enumerate}[label=\textbf{(\Roman*)},leftmargin=2.6em,itemsep=0.35em]
		\item $m\neq0$: $\det\Gamma\neq0$, $\ker f^{(1)}=0$ by
		Corollary~\ref{cor:regularity}, $\det f^{(1)}=m^{2}$. Second class;
		the multiplier is determined.
		\item $m=0$, $c$ a nonzero constant: $\Om=c$ has no zero,
		$\Sigma=\emptyset$; inconsistent.
		\item $m=0$, $c\equiv0$, $W$ arbitrary: no constraint is generated at all;
		$e_{3}$ satisfies $f^{(0)}e_{3}=0$ and $e_{3}\cdot\nabla V=0$ identically, so
		by Theorem~\ref{thm:strong} $\delta q_{3}=\varepsilon(t)$ is an exact
		symmetry. Strong gauge direction.
		\item $W=0$, $c=q_{2}$, $m=0$, i.e. $L=q_{2}\dot q_{1}-q_{2}q_{3}$:
		developed below.
	\end{enumerate}
	
	\begin{remark}[Reading the family as a stratification of $\det\Gamma$]
		\label{rem:familystrat}
		Regimes I to III are the elementary parametric stratification of the
		criterion of Corollary~\ref{cor:regularity}. The critical locus is
		$\det\Gamma=0$, that is $m=0$, and the family shows that crossing it has no
		single universal consequence: what happens on the locus is decided by the
		remaining structure, here by the second parameter $c$. For $m\neq0$ the
		system is regular and the constraint is second class; for $m=0$ it is
		inconsistent if $c$ is a nonzero constant and carries a strong gauge
		direction if $c\equiv0$. The locus is therefore not a single phenomenon but
		a boundary between regimes whose identity depends on further data. This is
		the simplest instance of the point made throughout
		Section~\ref{sec:mechanical}, that parameters are part of the classification
		and not coefficients to be cleared once the equations have been formed.
	\end{remark}
	
	\subsection{Regime IV: a genuine two-step chain}
	\label{sec:regimeIV}
	
	Let $L=q_{2}\dot q_{1}-q_{2}q_{3}$. Then $\Om=\pd{3}V=q_{2}$,
	$B=\nabla\Om=e_{2}$,
	\[
	\Gamma=N^{\top}B=e_{3}^{\top}e_{2}=0,
	\qquad
	M=(0)\ \text{(since }k=1),
	\qquad
	\dim\ker f^{(1)}=1+1-0=2 .
	\]
	Explicitly, with the ordering $(q_{1},q_{2},q_{3}\,|\,\lambda)$,
	\[
	f^{(1)}=\begin{pmatrix}0&-1&0&0\\1&0&0&1\\0&0&0&0\\0&-1&0&0\end{pmatrix},
	\qquad
	\ker f^{(1)}=\spn\{(1,0,0\,|\,-1),\,(0,0,1\,|\,0)\}.
	\]
	By \eqref{eq:fccount} there is exactly one first-class combination, namely $\Om$
	itself. The second basis vector is multiplier-free
	(Proposition~\ref{prop:mfree}) with contraction
	$\Phi_{e_{3}}=\pd{3}V=q_{2}=\Om\in\Ical$: a weak null mode. The first has
	multiplier block $-e_{1}$ and physical block $u_{0}=(1,0,0)$, with
	\[
	f^{(0)}u_{0}=e_{2}=\nabla\Om,
	\qquad
	\Phi_{u_{0}}=u_{0}\cdot\nabla V=\pd{1}V=0 .
	\]
	Both chain equations hold exactly, so by Theorem~\ref{thm:chainkernel} the
	transformation
	\begin{equation}
		\delta q_{1}=\rho(t),
		\qquad
		\delta q_{2}=0,
		\qquad
		\delta q_{3}=\dot\rho(t),
		\label{eq:chainIV}
	\end{equation}
	is an exact gauge symmetry: indeed
	$\delta L=\delta(q_{2}\dot q_{1}-q_{2}q_{3})
	=q_{2}\dot\rho-q_{2}\dot\rho=0$ identically. The canonical charge with the pair
	$q_{1}\leftrightarrow q_{2}$ is $G=q_{2}=\Om$, the first-class constraint, as
	expected.
	
	Two qualifications are visible here. First, the cross-check
	$r_{1}=\rank\{v^{(a)}\}$ versus $r_{2}=s-\rank\{\dd\Phi_{v^{(a)}}\}$ gives
	$r_{1}=2$ but $r_{2}=1$, because $\dd\Phi_{e_{3}}=\dd q_{2}\neq0$; the mismatch
	is legitimate and must not be read as dependence of the generators. Second, the
	ambient kernel has dimension $2$ while $\ker(\iota^{*}f^{(1)})$ has dimension
	$3$ on $\Sigma=\{q_{2}=0\}$: the multiplier direction $\partial_{\lambda}$ lies
	in the restricted kernel but not in the ambient one, which is the familiar
	statement that multipliers conjugate to first-class constraints are arbitrary
	functions of time.
	
	\subsection{The counterexample as a unit test}
	
	Example~\ref{ex:counter} and Theorem~\ref{thm:counterchain} give
	\[
	\Gamma=0_{2\times2},\quad
	M=J^{-1}\ \text{invertible},\quad
	\dim\ker f^{(1)}=2,
	\]
	\[
	\text{first-class count}=0,\qquad
	\text{ChainSolvable}=\text{false},
	\]
	against $\text{DependentConstraints}=\text{true}$. This is the pair of outputs
	that distinguishes the present theory from the naive reading, and it is the
	minimal system on which the two disagree.
	
	\subsection{Rank jump}
	
	Example~\ref{ex:jump} is the test case for \textbf{H1}: the generic verdict is
	``regular'', the degeneracy locus is the zero set of $\gen{q_{1}}$, and on that
	stratum every direction is null.
	
	\subsection{A parameter-controlled gauge transition}
	\label{sec:paramgauge}
	
	The degeneracy loci met so far live in the space of states: the stratum
	$\Delta$ of Theorem~\ref{thm:strat}, the locus $q_{1}=0$ of
	Example~\ref{ex:jump}. The example of this subsection is of a different kind.
	Its critical locus depends on the parameters of the model alone, and crossing
	it changes not only the algebra but the general solution of the equations of
	motion. The system is already in first-order form, so no Legendre transform
	intervenes and every object of the theory can be read off directly.
	
	Let $\xi=(q_{1},q_{2},p_{1})\in\R^{3}$ and
	\begin{equation}
		L=p_{1}\dot q_{1}+\alpha\,q_{1}\dot q_{2}
		-\tfrac12 p_{1}^{2}-\tfrac{\beta}{2}\,(q_{1}-q_{2})^{2},
		\qquad \alpha,\beta\in\R,
		\label{eq:pgL}
	\end{equation}
	so that, in the notation of \eqref{eq:lag},
	\begin{equation}
		a=\bigl(p_{1},\;\alpha q_{1},\;0\bigr),
		\qquad
		V=\tfrac12 p_{1}^{2}+\tfrac{\beta}{2}(q_{1}-q_{2})^{2}.
		\label{eq:pgaV}
	\end{equation}
	
	\paragraph{Presymplectic matrix.} From $f^{(0)}_{AB}=\pd{A}a_{B}-\pd{B}a_{A}$,
	\begin{equation}
		f^{(0)}=\begin{pmatrix}0&\alpha&-1\\ -\alpha&0&0\\ 1&0&0\end{pmatrix},
		\qquad
		\rank f^{(0)}=2,
		\qquad
		\ker f^{(0)}=\spn\{N\},\quad N=(0,1,\alpha)^{\top},
		\label{eq:pgf0}
	\end{equation}
	for every $(\alpha,\beta)$. The rank is constant on the whole parameter plane,
	because the entry $f^{(0)}_{q_{1}p_{1}}=-1$ never vanishes. Whatever happens
	below is therefore \emph{not} a rank jump of $f^{(0)}$, and hypothesis
	\textbf{H1} is satisfied uniformly.
	
	\paragraph{Constraint and reduced pairing.} Contracting the null direction with
	$\nabla V$,
	\begin{equation}
		\Om=N\!\cdot\!\nabla V=\alpha p_{1}-\beta(q_{1}-q_{2}),
		\qquad
		B=\nabla\Om=(-\beta,\;\beta,\;\alpha)^{\top},
		\label{eq:pgOm}
	\end{equation}
	and therefore
	\begin{equation}
		\Gamma=N^{\top}B=\alpha^{2}+\beta .
		\label{eq:pgGamma}
	\end{equation}
	Here $\Gamma$ is a function of the parameters alone. Since $k=d=1$, the matrix
	$M$ is $1\times1$ antisymmetric and vanishes identically, so by
	Corollary~\ref{cor:dim}
	\begin{equation}
		\dim\ker f^{(1)}
		=\begin{cases}
			0, & \alpha^{2}+\beta\neq0,\\[0.2em]
			2, & \alpha^{2}+\beta=0 .
		\end{cases}
		\label{eq:pgdim}
	\end{equation}
	The locus $\Gamma=0$ is a smooth curve in the $(\alpha,\beta)$ plane, of
	codimension one, and it is the parabola $\beta=-\alpha^{2}$.
	
	\begin{remark}[Parameter locus, not a frame artifact]
		\label{rem:pgframe}
		A symbolic computation of $\ker f^{(0)}$ over the field of rational
		functions typically returns the normalized frame
		$N'=(0,\alpha^{-1},1)^{\top}$, for which
		$\Gamma'=(\alpha^{2}+\beta)/\alpha^{2}$ and
		$\det f^{(1)}=(\alpha^{2}+\beta)^{2}/\alpha^{2}$. The discrepancy between
		$\Gamma$ and $\Gamma'$ is exactly the factor $\det(N^{\top}N)$ of
		Proposition~\ref{prop:basis}, and the vanishing locus is the same in both
		frames. The pole $\alpha=0$ is the locus $P$ of Theorem~\ref{thm:strat}, an
		artifact of the rational frame with no intrinsic content; the frame
		\eqref{eq:pgf0} is regular there. The two must be kept apart:
		$\alpha^{2}+\beta=0$ is a stratification of the parameter space, whereas
		$\alpha=0$ is a defect of a representation, and neither is a stratification
		of the state space.
	\end{remark}
	
	\paragraph{The dynamical content.} The algebra locates the locus. To see what
	it means one integrates the equations of motion, which by Lemma~\ref{lem:eom}
	are $f^{(0)}\dot\xi=\nabla V$, that is
	\begin{equation}
		\alpha\dot q_{2}-\dot p_{1}=\beta(q_{1}-q_{2}),
		\qquad
		\alpha\dot q_{1}=\beta(q_{1}-q_{2}),
		\qquad
		\dot q_{1}=p_{1} .
		\label{eq:pgeom}
	\end{equation}
	
	\begin{proposition}[Transition of functional indeterminacy]
		\label{prop:pgtransition}
		Assume $\alpha\neq0$ and set $u=q_{1}-q_{2}$. Then \eqref{eq:pgeom} implies
		\begin{equation}
			\alpha p_{1}=\beta u
			\qquad\text{and}\qquad
			\bigl(\alpha^{2}+\beta\bigr)\,\dot u=0 .
			\label{eq:pgkey}
		\end{equation}
		Consequently the general solution of \eqref{eq:pgeom} contains
		\begin{equation}
			\dim\bigl(\text{functional indeterminacy}\bigr)
			=\begin{cases}
				0, & \alpha^{2}+\beta\neq0,\\[0.2em]
				1, & \alpha^{2}+\beta=0,
			\end{cases}
			\label{eq:pgindet}
		\end{equation}
		arbitrary functions of time. Explicitly, for $\alpha^{2}+\beta\neq0$ one has
		$u$ constant, $p_{1}=\beta u/\alpha$ constant and $q_{1}$ affine in $t$, so
		the evolution is fixed by the admissible initial data; for
		$\alpha^{2}+\beta=0$ the function $q_{2}(t)$ may be prescribed arbitrarily
		and $q_{1}$, $p_{1}$ follow.
	\end{proposition}
	
	\begin{proof}
		The second and third equations of \eqref{eq:pgeom} give
		$\alpha p_{1}=\alpha\dot q_{1}=\beta u$, which is the first relation in
		\eqref{eq:pgkey} and is $\Om=0$ rewritten. Substituting
		$\dot q_{2}=\dot q_{1}-\dot u$ and $\alpha\dot q_{1}=\beta u$ into the first
		equation,
		\[
		\dot p_{1}=\alpha\dot q_{2}-\beta u
		=\alpha\dot q_{1}-\alpha\dot u-\beta u
		=\beta u-\alpha\dot u-\beta u
		=-\alpha\dot u .
		\]
		Differentiating $\alpha p_{1}=\beta u$ gives $\alpha\dot p_{1}=\beta\dot u$,
		and combining the two yields $-\alpha^{2}\dot u=\beta\dot u$, which is the
		second relation in \eqref{eq:pgkey}. If $\alpha^{2}+\beta\neq0$ it forces
		$\dot u=0$, and the stated form of the solution follows. If
		$\alpha^{2}+\beta=0$ the relation reads $0=0$ and imposes nothing: given any
		$q_{2}\in C^{1}$, set $u$ so that $\dot q_{1}=p_{1}=\beta u/\alpha$ is
		consistent with $u=q_{1}-q_{2}$, which is a first-order ordinary
		differential equation for $q_{1}$ with a solution for every admissible
		datum; the remaining equation is then satisfied identically.
	\end{proof}
	
	\paragraph{The same transition seen from the bordered kernel.} For
	$\Gamma\neq0$ the bordered matrix is regular, $\ker f^{(1)}=0$, and there is no
	null mode to discuss: the system is second class and the multiplier is
	determined. For $\Gamma=0$, that is $\beta=-\alpha^{2}$, the kernel is
	two-dimensional by \eqref{eq:pgdim}, with one multiplier-free direction and one
	candidate chain link. Solving the first chain equation
	$f^{(0)}u_{0}=\nabla\Om$ with $\beta=-\alpha^{2}$ gives the family
	$u_{0}=(\alpha,\,u_{2},\,\alpha u_{2}-\alpha^{2})$ with $u_{2}$ free, and the
	decisive contraction of Theorem~\ref{thm:chain-decisive} is
	\begin{equation}
		\Phi_{u_{0}}=u_{0}\!\cdot\!\nabla V
		=\alpha\,(u_{2}-\alpha)\,\bigl[\alpha(q_{1}-q_{2})+p_{1}\bigr].
		\label{eq:pgPhi}
	\end{equation}
	Neither factor is cancelled: the first records that the analysis is being
	carried out in the frame \eqref{eq:pgf0}, and the second is what allows the
	choice $u_{2}=\alpha$, for which $\Phi_{u_{0}}\equiv0$ and
	\begin{equation}
		u_{0}=(\alpha,\alpha,0),
		\qquad
		B^{\top}u_{0}=-\beta\alpha+\beta\alpha+0=0,
		\qquad
		(u_{0},-1)\in\ker f^{(1)} .
		\label{eq:pgu0}
	\end{equation}
	By Theorem~\ref{thm:chainkernel} the two-step chain closes and
	\begin{equation}
		\delta q_{1}=\delta q_{2}=\alpha\,\rho(t),
		\qquad
		\delta p_{1}=0,
		\qquad
		\delta\lambda=\dot\rho(t),
		\label{eq:pggauge}
	\end{equation}
	is an exact gauge symmetry. The gauge direction is the simultaneous translation
	of the two coordinates, and it is precisely the freedom that
	Proposition~\ref{prop:pgtransition} exhibits as the arbitrary function
	$q_{2}(t)$.
	
	\begin{remark}[What is gained by integrating]
		\label{rem:pgstrength}
		The locus $\Gamma=0$ could have been reported as an algebraic degeneracy:
		a determinant vanishes, a kernel appears. Proposition~\ref{prop:pgtransition}
		says more, and says it in the language of solutions rather than of matrices.
		Off the locus the equations of motion determine the evolution completely;
		on it the general solution acquires one arbitrary function of time. What the
		critical parameter value controls is the dimension of the functional
		indeterminacy of the dynamics, and that statement is established by
		integrating \eqref{eq:pgeom}, not by inspecting $\det f^{(1)}$ or $\Gamma$.
		The matrix identifies the locus; the integration establishes its meaning.
		In the terminology fixed in Remark~\ref{rem:parameters} below, we call this
		a parameter-controlled transition of dynamical indeterminacy. We do not call
		it a bifurcation: there are no equilibrium branches that collide, split or
		exchange stability here, and the reduced dynamics $\dot u=0$ has identically
		vanishing spectrum.
	\end{remark}
	
	\begin{remark}[Admissible domain]
		\label{rem:pgdomain}
		The locus requires $\beta=-\alpha^{2}\leq0$. If $\beta$ is read as an
		ordinary stiffness, then apart from the trivial point $\alpha=\beta=0$ the
		locus lies outside the domain of a passive elastic coupling, and reaching it
		requires an inverted or actively driven coupling. The statement of
		Proposition~\ref{prop:pgtransition} holds on the whole real parameter plane
		and is in that sense a mathematical transition with a dynamical
		characterization; it is not an assertion about the response of a passive
		spring. We record the three levels separately: the full real parameter plane,
		on which the proposition holds; the subdomain of ordinary positive
		couplings, which the locus meets only trivially; and the extended models
		with negative or active couplings, in which the locus is attained.
	\end{remark}
	
	\begin{remark}[Scope of the example]
		\label{rem:pgscope}
		Nothing here upgrades to a general theorem of the form
		$\Gamma=0\Rightarrow$ one arbitrary function of time. The general
		classification remains the one developed in
		Sections~\ref{sec:kernel} to \ref{sec:weak}: the bordered kernel, the
		first-class count, ideal membership and the closure of the chain. In the
		present system the count \eqref{eq:pgindet} comes out as it does because
		$k=d=1$, because $M$ vanishes for that reason alone, and because the chain
		closes; Section~\ref{sec:family} already shows that the same locus
		$\Gamma=0$ can instead produce inconsistency. The example is a verification
		of the theory on a case where the answer can be obtained independently, not
		an extension of it.
	\end{remark}
	
	\section{Mechanical realizations of the gauge criterion}
	\label{sec:mechanical}
	
	The systems examined so far were chosen for their algebraic transparency. We now
	turn to ordinary finite-dimensional mechanics, where the same chain (null
	direction, potential invariance, primary constraint, gauge generator) appears in
	systems built from masses, springs, rods and pulleys, and where its last link is
	again what carries the informative distinction.
	
	Brown provides representative mechanical realizations of singular Lagrangian
	systems and analyses them within the Dirac--Bergmann framework
	\cite{Brown2023}. Here we revisit those realizations from the first-order
	Faddeev--Jackiw perspective developed above, using them as concrete tests of the
	gauge-chain criterion. The division of labour is worth stating plainly. Brown
	supplies the mechanical constructions, the singular second-order Lagrangians
	they obey, the identification of the singularity, and the constrained
	Hamiltonian analysis with its constraints and gauge transformations. What is
	developed here is the identification of the degenerate kinetic structure as an
	object to be reorganized, the passage to a first-order Faddeev--Jackiw system,
	the analysis of its presymplectic matrix and null modes, the bordering, and the
	chain criterion of Theorem~\ref{thm:chain-decisive}. Where Brown's
	Dirac--Bergmann analysis reaches the same conclusion we record the agreement as
	external validation, not as the ground of the result. A Faddeev--Jackiw
	treatment of several of these systems has also been given by Paulin-Fuentes,
	Arellano and Cabrera \cite{PaulinFuentes2024}, with the emphasis on the
	generalized brackets and the equations of motion rather than on the gauge
	criterion at issue here.
	
	That treatment also supplies a convenient illustration of a point of method
	that will recur below. For the pendulum suspended from two springs, both the
	Dirac--Bergmann analysis of \cite{Brown2023} and the Faddeev--Jackiw analysis
	of \cite{PaulinFuentes2024} produce the constraint in the factorized form
	\begin{equation}
		2k\ell\,\bigl(x\cos\theta+y\sin\theta\bigr)=0,
		\label{eq:pendfactor}
	\end{equation}
	and the analysis proceeds with the reduced condition
	$x\cos\theta+y\sin\theta=0$. On the regular stratum $k\ell\neq0$ the two are
	equivalent and the reduction is exactly the right move. They are not
	equivalent on the whole parameter family: on the stratum $k\ell=0$ the
	left-hand side of \eqref{eq:pendfactor} vanishes identically, so the reduced
	condition asserts a restriction that the original expression does not impose.
	The cancellation is therefore harmless when read as a restriction to the
	regular stratum, and it is not a globally equivalent reformulation unless the
	vanishing stratum has been classified separately. Nothing in this observation
	bears on the correctness of the results obtained on the regular branch; it
	bears on the scope of the statement, and it is the reason for the discipline
	set out in Remark~\ref{rem:parameters}. Section~\ref{sec:square-control} then
	exhibits a system in which the analogous vanishing stratum is not empty of
	content.
	
	Two matrices must be kept apart throughout. The Hessian of the second-order
	Lagrangian with respect to the velocities,
	\begin{equation}
		K_{ij}=\frac{\partial^{2}L_{\mathrm{mech}}}{\partial\dot q^{i}\partial\dot q^{j}},
		\label{eq:Khessian}
	\end{equation}
	is the object whose degeneracy signals the singular character of the mechanical
	system in its Lagrangian formulation, and it is the form in which the
	singularity is naturally detected there. The presymplectic matrix
	$f^{(0)}_{AB}=\pd{A}a_{B}-\pd{B}a_{A}$ of Section~\ref{sec:conventions} is a
	different object, attached to a first-order Lagrangian
	$L=a_{A}(\xi)\dot\xi^{A}-V(\xi)$ that we construct from the mechanical system.
	The Faddeev--Jackiw route does not identify $K$ with $f^{(0)}$. It reorganizes
	the information carried by the degeneracy of $K$ into the presymplectic
	structure of an enlarged first-order system, in which $\ker K$ reappears first
	as a set of momentum constraints and then, after bordering, inside
	$\ker f^{(m)}$. Proposition~\ref{prop:template} makes that passage explicit and
	is used in every example below.
	
	Each subsection therefore follows the same route: singular second-order
	Lagrangian, degenerate $K$, first-order reformulation, presymplectic matrix
	$f^{(0)}$ and its null directions, bordering, consistency and the chain
	criterion, gauge generator. Every derivation is analytic and can be followed and
	checked by hand.
	
	\subsection{A template for degenerate mechanical systems}
	\label{sec:template}
	
	Almost every mechanical example of interest has the same shape, and it is
	economical to treat it once. The passage from the second-order mechanical
	system to the first-order Faddeev--Jackiw system is carried out explicitly
	first, because it is the step at which $K$ and $f^{(0)}$ are most easily
	confused.
	
	\begin{lemma}[From the second-order system to the first-order form]
		\label{lem:legendre}
		Let $L_{\mathrm{mech}}=\tfrac12\dot q^{\top}K\dot q-V(q)$ with
		$q\in\R^{n}$ and $K$ constant, symmetric and positive semidefinite of rank
		$n-c$ with $c\geq1$. Let $w_{1},\dots,w_{c}$ be a basis of $\ker K$ and let
		$K^{+}$ be the Moore--Penrose inverse of $K$, so that $KK^{+}K=K$ and
		$K^{+}K=P$ is the orthogonal projector onto $\im K=(\ker K)^{\perp}$. Then:
		\begin{enumerate}[label=(\roman*),leftmargin=2.4em,itemsep=0.25em]
			\item the Legendre map is $p_{i}=\partial L_{\mathrm{mech}}/\partial\dot
			q^{i}=(K\dot q)_{i}$, its image is $\im K$, and the relations
			\begin{equation}
				\Om_{a}:=w_{a}\!\cdot\!p=0,
				\qquad a=1,\dots,c,
				\label{eq:primaryK}
			\end{equation}
			are $c$ independent primary constraints and exhaust them;
			\item on the primary surface $\dot q=K^{+}p$ modulo $\ker K$, and the
			canonical Hamiltonian is
			\begin{equation}
				H_{C}=p\!\cdot\!\dot q-L_{\mathrm{mech}}
				=\tfrac12\,p^{\top}K^{+}p+V(q),
				\label{eq:HC}
			\end{equation}
			determined off that surface only modulo the ideal generated by the
			$\Om_{a}$;
			\item with $\lambda^{a}$ the multipliers of the primary constraints, the
			total Hamiltonian is $H_{T}=H_{C}+\lambda^{a}\Om_{a}$, and the
			first-order action $\int(p_{i}\dot q^{i}-H_{T})\,\dd t$ has as
			Euler--Lagrange equations the Hamilton equations of $H_{T}$ together with
			\eqref{eq:primaryK}, which are equivalent to the Euler--Lagrange
			equations of $L_{\mathrm{mech}}$;
			\item the ambiguity in (ii) changes $H_{T}$ by
			$\sum_{a}c^{a}(\xi)\,\Om_{a}$ and is therefore of the form covered by
			Proposition~\ref{prop:convention}: it does not affect any of the verdicts
			below.
		\end{enumerate}
	\end{lemma}
	
	\begin{proof}
		(i) $K$ is symmetric, so $\im K=(\ker K)^{\perp}$; hence $p=K\dot q$ holds
		for some $\dot q$ if and only if $w_{a}\!\cdot\!p=0$ for every $a$, and the
		$w_{a}$ being independent the $c$ relations are independent. Since
		$\rank K=n-c$, the fibres of the Legendre map are the cosets of $\ker K$ and
		no further relation among $q$ and $p$ follows.
		
		(ii) Write $\dot q=K^{+}p+\nu$ with $\nu\in\ker K$, which solves $K\dot q=p$
		on the primary surface because $KK^{+}p=Pp=p$ there. Then
		$p\!\cdot\!\dot q=p^{\top}K^{+}p$, since $p\perp\ker K$, and
		\[
		\tfrac12\dot q^{\top}K\dot q
		=\tfrac12\bigl(K^{+}p+\nu\bigr)^{\top}K\bigl(K^{+}p+\nu\bigr)
		=\tfrac12\,p^{\top}K^{+}KK^{+}p
		=\tfrac12\,p^{\top}K^{+}p ,
		\]
		using $K\nu=0$ and $K^{+}KK^{+}=K^{+}$. Subtracting gives \eqref{eq:HC}.
		Off the primary surface any two expressions of $H_{C}$ in the variables
		$(q,p)$ differ by a combination of the $\Om_{a}$, which is the standard
		non-uniqueness of the canonical Hamiltonian for a singular Lagrangian.
		
		(iii) and (iv) are the standard Dirac construction together with
		Proposition~\ref{prop:convention}.
	\end{proof}
	
	\begin{remark}[$K$ and $f^{(0)}$ have different sizes]
		\label{rem:KvsF}
		Lemma~\ref{lem:legendre} makes the distinction of
		Section~\ref{sec:mechanical} concrete. The velocity Hessian $K$ is an
		$n\times n$ matrix on the configuration space of the mechanical system,
		of rank $n-c$. The presymplectic matrix $f^{(0)}$ produced below is a
		$(2n+c)\times(2n+c)$ matrix on the extended space
		$\xi=(q,p,\lambda)$, and its rank is $2n$, independently of $\rank K$. They
		are not two matrix representations of one object, and neither is obtained
		from the other by a change of basis. What passes from one to the other is
		$\ker K$, which enters $f^{(0)}$ not as a kernel but through the primary
		constraints \eqref{eq:primaryK} and the multipliers conjugate to them.
	\end{remark}
	
	\begin{proposition}[Mechanical template]
		\label{prop:template}
		Let
		\begin{equation}
			L_{\mathrm{mech}}=\tfrac12\,\dot q^{\top}K\dot q-V(q),
			\qquad q\in\R^{n},
			\label{eq:Lmech}
		\end{equation}
		be a singular second-order Lagrangian, so that $K$ is its velocity Hessian
		\eqref{eq:Khessian}, taken constant, symmetric and positive semidefinite of rank
		$n-c$, and let $w_{1},\dots,w_{c}$ be a basis of $\ker K$. By
		Lemma~\ref{lem:legendre} the first-order Faddeev--Jackiw form of the system
		is carried by $\xi=(q^{i},p_{i},\lambda^{a})\in\R^{2n+c}$ with
		\begin{equation}
			L=p_{i}\dot q^{i}-V_{\mathrm{ext}},
			\qquad
			V_{\mathrm{ext}}=\tfrac12\,p^{\top}K^{+}p+V(q)+\lambda^{a}\,(w_{a}\!\cdot\!p).
			\label{eq:Lfo}
		\end{equation}
		Then:
		\begin{enumerate}[label=(\alph*),leftmargin=2.2em,itemsep=0.25em]
			\item the Faddeev--Jackiw one-form of \eqref{eq:Lfo} is
			$a=(p_{1},\dots,p_{n},0,\dots,0,0,\dots,0)$ in the ordering
			$(q,p,\lambda)$, and the associated presymplectic matrix is
			\begin{equation}
				f^{(0)}
				=\begin{pmatrix}0&-I_{n}&0\\ I_{n}&0&0\\ 0&0&0_{c}\end{pmatrix}
				\in\R^{(2n+c)\times(2n+c)},
				\qquad
				\rank f^{(0)}=2n,
				\label{eq:templatef0}
			\end{equation}
			so $d=c$ and
			$\ker f^{(0)}=\spn\{\partial_{\lambda^{1}},\dots,\partial_{\lambda^{c}}\}$;
			in particular $f^{(0)}$ is not $K$, and by Remark~\ref{rem:KvsF} it is
			not even a matrix of the same size;
			\item the generated constraints are the momentum constraints
			$\Om_{a}=w_{a}\!\cdot\!p$, and
			\begin{equation}
				\Gamma=0_{c\times c},
				\qquad
				M=0_{c\times c},
				\qquad
				\dim\ker f^{(1)}=2c ;
				\label{eq:templateGM}
			\end{equation}
			in particular every combination $\Om_{w}$ is first class, the count
			\eqref{eq:fccount} being $c-0=c$, and a basis of $\ker f^{(1)}$ is
			\begin{equation}
				\bigl\{(\partial_{\lambda^{a}},0)\bigr\}_{a=1}^{c}
				\;\cup\;
				\bigl\{(w_{b}^{i}\partial_{q^{i}},-e_{b})\bigr\}_{b=1}^{c},
				\label{eq:templatekerbasis}
			\end{equation}
			the first family being the multiplier-free directions of
			Proposition~\ref{prop:mfree} and the second the candidate chain links;
			\item for $u_{1}=\partial_{\lambda^{b}}$ the first chain equation
			$f^{(0)}u_{0}=\nabla\Om_{b}$ has the solution
			$u_{0}=w_{b}^{i}\partial_{q^{i}}$, unique modulo $\ker f^{(0)}$, and its
			contraction is
			\begin{equation}
				\Phi_{u_{0}}=u_{0}\!\cdot\!\nabla V_{\mathrm{ext}}
				=w_{b}\!\cdot\!\nabla V(q);
				\label{eq:templatePhi}
			\end{equation}
			\item the two-step chain closes, and
			\begin{equation}
				\delta q^{i}=\rho(t)\,w_{b}^{i},
				\qquad
				\delta p_{i}=0,
				\qquad
				\delta\lambda^{b}=\dot\rho(t),
				\label{eq:templatechain}
			\end{equation}
			is an exact gauge symmetry, \emph{if and only if}
			$w_{b}\!\cdot\!\nabla V\equiv0$. The associated canonical charge, for the
			Darboux pairs $q^{i}\leftrightarrow p_{i}$, is
			$G=w_{b}\!\cdot\!p=\Om_{b}$;
			\item if $w_{b}\!\cdot\!\nabla V\not\equiv0$ then the decisive class of
			Theorem~\ref{thm:chain-decisive} is nonzero, no two-step chain exists, and
			$w_{b}\!\cdot\!\nabla V=0$ is a new constraint on which the iteration
			continues.
		\end{enumerate}
	\end{proposition}
	
	\begin{proof}
		(a) From $L=p_{i}\dot q^{i}-V_{\mathrm{ext}}$ the coefficient of $\dot\xi^{A}$
		is $a_{q^{i}}=p_{i}$, $a_{p_{i}}=0$, $a_{\lambda^{a}}=0$, whence
		$f^{(0)}_{q^{i}p_{j}}=\pd{q^{i}}a_{p_{j}}-\pd{p_{j}}a_{q^{i}}
		=0-\delta_{i}^{j}=-\delta_{i}^{j}$,
		$f^{(0)}_{p_{j}q^{i}}=\delta_{i}^{j}$, and every entry with a $\lambda$ index
		vanishes because no $a_{A}$ depends on $\lambda$ and $a_{\lambda^{a}}=0$.
		This is \eqref{eq:templatef0}; its rank is $2n$ and its kernel is spanned by
		the $\lambda$ directions.
		
		(b) $\Om_{a}=\partial_{\lambda^{a}}V_{\mathrm{ext}}=w_{a}\!\cdot\!p$, whose
		gradient has no $\lambda$ component, so $\Gamma_{ab}=\partial_{\lambda^{a}}
		(w_{b}\!\cdot\!p)=0$. Here the frame data are constant, so condition
		\eqref{eq:constframe} holds and the identities of Theorem~\ref{thm:dirac} are
		strong. For $M$, by Lemma~\ref{lem:fjdirac} and Theorem~\ref{thm:dirac},
		$M_{ab}=\{w_{a}\!\cdot\!p,\,w_{b}\!\cdot\!p\}_{\FJ}=0$ because momenta commute
		in the canonical block. Explicitly, $B$ has the single nonzero block
		$B_{p_{i}\,b}=w_{b}^{i}$, the invertible block of \eqref{eq:templatef0} is
		$J=\bigl(\begin{smallmatrix}0&-I_{n}\\ I_{n}&0\end{smallmatrix}\bigr)$ with
		$J^{-1}=\bigl(\begin{smallmatrix}0&I_{n}\\ -I_{n}&0\end{smallmatrix}\bigr)$,
		so $B_{u,b}=(0,w_{b})$ and
		$M_{ab}=(0,w_{a})^{\top}J^{-1}(0,w_{b})=(0,w_{a})\!\cdot\!(w_{b},0)=0$.
		Formula \eqref{eq:dim} then gives $\dim\ker f^{(1)}=c+c-0=2c$, and
		\eqref{eq:fccount} gives $c$ first-class combinations. For the basis
		\eqref{eq:templatekerbasis}: the multiplier-free elements are
		Proposition~\ref{prop:mfree} with $\Gamma^{\top}=0$; and for $w=-e_{b}$,
		Theorem~\ref{thm:kernel}(c) gives
		$x=-J^{-1}B_{u}w=J^{-1}(0,w_{b})=(w_{b},0)$, that is $u_{q}=w_{b}$,
		$u_{p}=0$, while $\Gamma^{\top}y=Mw=0$ allows $y=0$.
		
		(c) $\nabla\Om_{b}=(0,w_{b},0)$. Writing $u_{0}=(u_{q},u_{p},u_{\lambda})$, the
		equation $f^{(0)}u_{0}=\nabla\Om_{b}$ reads $-u_{p}=0$ and $u_{q}=w_{b}$;
		the $\lambda$ rows are vacuous, which is the stated indeterminacy modulo
		$\ker f^{(0)}$. Since neither $\tfrac12p^{\top}K^{+}p$ nor
		$\lambda^{a}(w_{a}\!\cdot\!p)$ depends on $q$, the contraction reduces to
		\eqref{eq:templatePhi}.
		
		(d) By Theorem~\ref{thm:chain2} the chain closes exactly when
		$\Phi_{u_{0}}\equiv0$, which by \eqref{eq:templatePhi} is
		$w_{b}\!\cdot\!\nabla V\equiv0$. The transformation \eqref{eq:templatechain} is
		$\delta\xi=\rho\,u_{0}+\dot\rho\,u_{1}$ written out. For the charge,
		$\iota_{u_{0}}\omega_{\mathrm{can}}=w_{b}^{i}\dd p_{i}$, which is closed and
		constant, so Proposition~\ref{prop:nocharge} and the constant-direction formula
		of Section~\ref{sec:charges} give $G=w_{b}^{i}p_{i}$.
		
		(e) The ideal is $\Ical=\gen{w_{1}\!\cdot\!p,\dots,w_{c}\!\cdot\!p}$, generated
		by elements linear and homogeneous in $p$, whereas
		$w_{b}\!\cdot\!\nabla V$ depends on $q$ alone; hence
		$w_{b}\!\cdot\!\nabla V\in\Ical$ forces $w_{b}\!\cdot\!\nabla V=0$. The class of
		Theorem~\ref{thm:chain-decisive} is therefore nonzero precisely when
		$w_{b}\!\cdot\!\nabla V\not\equiv0$, and then case (3) of
		Theorem~\ref{thm:trichotomy} applies.
	\end{proof}
	
	Proposition~\ref{prop:template} isolates the content of the whole theory in a
	single line: for this class of systems the algebraic data $\Gamma$ and $M$
	vanish identically, so the bordered kernel is as large as it can be and every
	generated constraint is first class, and yet gauge freedom holds or fails
	according to a condition, $w\!\cdot\!\nabla V\equiv0$, that neither $\Gamma$ nor
	$M$ sees. It is the chain criterion, and only the chain criterion, that decides.
	
	\begin{remark}[Parameter discipline]
		\label{rem:parameters}
		Every system below carries physical parameters, among them masses,
		stiffnesses, radii and equilibrium lengths. The discipline of carrying such
		factors through the computation rather than clearing them was introduced in
		the companion paper \cite{CLMCP2026}, where it was justified structurally:
		premature cancellation can erase a locus in parameter space on which a rank,
		a kernel or a constraint structure changes, and the erasure is irreversible
		because the locus no longer appears in the surviving expressions. We follow
		that discipline here, and the points of method below are its concrete form.
		They matter because the objects of the theory are rank-dependent.
		
		First, the hypotheses of Proposition~\ref{prop:template} are themselves
		parametric. The statement presupposes $\rank K=n-c$ with a fixed $c$, so
		every example records explicitly the locus in parameter space on which that
		rank is attained, and the analysis is asserted only off the complementary
		locus. Where the excluded locus carries different behaviour, it is described
		rather than discarded. This is hypothesis \textbf{H1} made concrete.
		
		Second, no parametric factor is cancelled from an expression that is being
		set to zero. The general form of the rule is this. Given a relation
		\begin{equation}
			A(\lambda)\,F(z)=0
			\label{eq:stratifyrule}
		\end{equation}
		in which $A$ depends on the parameters $\lambda$ and $F$ on the state
		variables $z$, the parameter space is first divided. On $A(\lambda)\neq0$ the
		relation is equivalent to $F(z)=0$; on $A(\lambda)=0$ it is an identity and
		imposes nothing. Passing from \eqref{eq:stratifyrule} to $F(z)=0$ before that
		division has been made replaces a statement about the whole family by a
		statement about one stratum, and the discarded stratum leaves no trace in the
		surviving expressions. The rule is not that parameters may never be
		cancelled. It is that the cancellation is a restriction to a regular
		stratum, and that the stratum on which it is invalid must be classified
		before it is discarded: stratify first, simplify second. Concretely, if a
		contraction takes the form $w\!\cdot\!\nabla V=\mu\,F(q)$, the branches
		$\mu\neq0$ and $\mu=0$ are treated separately, because they can differ in
		whether the chain closes. Section~\ref{sec:square-control} contains a case
		in which they do.
		
		Third, the rule is not confined to polynomial relations. A parameter can
		carry structural information while sitting in a denominator, and clearing
		denominators before classifying the locus destroys it just as effectively as
		cancelling a common factor. In Section~\ref{sec:paramgauge} the decisive
		combination $\alpha^{2}+\beta$ appears as the denominator of the inverse of
		the bordered matrix, and passing to a common denominator and retaining the
		numerator would have removed exactly the locus at which the gauge structure
		changes. We therefore adopt the following convention throughout:
		parameter-dependent factors are preserved not only in the constraints and in
		the contractions, but also in denominators, determinants, inverses,
		pseudoinverses and reduced operators, until their vanishing or their
		divergence has been classified.
		
		Finally, a word on terminology. The critical loci met below are of several
		kinds, and we name them by what is demonstrated rather than by what they
		resemble. A locus on which a rank changes is a rank-jump locus; one on which
		a frame ceases to be defined is a pole locus; one on which the class of the
		constraints changes is a change of constraint structure; one on which a null
		mode passes from obstructed to gauge is a gauge-character transition; one on
		which an effective stiffness vanishes is a loss of restoring force; and one
		on which the number of arbitrary functions in the general solution changes
		is a transition of dynamical indeterminacy. None of these is called a
		bifurcation. We reserve that word for the collision, creation or exchange of
		stability of equilibrium branches, and no locus examined in this paper meets
		that criterion.
		
		The absence of an equilibrium bifurcation in the examples considered here
		does not diminish the role of parameter preservation. What these examples
		show is that parameter strata may encode other qualitative transitions,
		among them changes in gauge character and changes in the functional
		indeterminacy of the dynamics, and that such transitions are invisible once
		the corresponding factors have been cleared. Equilibrium bifurcations remain
		a further and equally legitimate motivation for preserving the full
		parameter space; none is claimed here, and none is excluded for other
		systems or other parameter families.
	\end{remark}
	
	\subsection{Three pairs of pulleys}
	\label{sec:pulley3}
	
	\begin{figure}[htbp]
		\centering
		\begin{tikzpicture}[line cap=round, line join=round, >=latex, scale=2.25]
			\def\r{0.35}
			
			\coordinate (A) at (-2,-0.12);
			\coordinate (B) at (-2+4*\r,-0.12);
			
			\draw[thick,decorate,decoration={coil,aspect=0.4,amplitude=3mm,segment length=2mm}]
			(A) -- (-2,-0.8);
			\draw[thick,decorate,decoration={coil,aspect=0.4,amplitude=3mm,segment length=4mm}]
			(B) -- (-2+4*\r,-1.3);
			\draw[thick,decorate,decoration={coil,aspect=0.4,amplitude=3mm,segment length=3mm}]
			(-2+8*\r,-0.12) -- (-2+8*\r,-1.05);
			
			\filldraw[fill=black!30] (-3,-0.12) rectangle (2.5,0.12);
			\filldraw[fill=black!30] (-3,-4) rectangle (2.5,-3.76); 
			
			\foreach \x/\y in {
				-2/-1.6,       
				-2+2*\r/-3,    
				-2+4*\r/-2.1,  
				-2+6*\r/-3,    
				-2+8*\r/-1.85, 
				-2+10*\r/-3    
			} {
				\filldraw[fill=black!15] (\x,\y) circle (\r);
				\filldraw[fill=yellow!90] (\x,\y) circle (0.12);
				\filldraw[fill=black] (\x,\y) circle (0.025);
			}
			
			\foreach \x/\y in {0.35/-1.6, 5*\r/-2.1, 9*\r/-1.85} {
				\draw[thick] (-2+\x,\y) -- (-2+\x,-3);
			}
			\foreach \x/\y in {-0.35/-1.6, 3*\r/-2.1, 7*\r/-1.85, 11*\r/-1.85} {
				\draw[thick,->] (-2+\x,\y) -- (-2+\x,\y-0.7);
				\draw[thick]    (-2+\x,\y-0.7) -- (-2+\x,-3);
			}
			
			\foreach \x/\y in {-2/-0.9, -2+4*\r/-1.4, -2+8*\r/-1.15} {
				\draw[very thick, gray!40] (\x,\y) -- (\x,\y-0.7);
			}
			\foreach \x/\y in {-2+2*\r/-3, -2+6*\r/-3, -2+10*\r/-3} {
				\draw[very thick, gray!40] (\x,\y) -- (\x,\y-0.75);
			}
			
			\foreach \x/\y/\l/\color/\size in {
				-2/-0.9/1/red/4.25pt,
				-2+4*\r/-1.4/2/blue/4.25pt,
				-2+8*\r/-1.15/3/green/4.25pt
			} {
				\shade[ball color=\color, draw=black] (\x,\y) circle (\size);
			}
			
			\node at (-1.55,-0.9)        {\normalsize $m$};
			\node at (-2.05+4*\r +0.5,-1.4) {\normalsize $m$};
			\node at (-2.05+8*\r +0.5,-1.15){\normalsize $m$};
			
			\node at (-2+2*\r -0.55,-3.2)  {\normalsize $\alpha_{1}$};
			\node at (-2+6*\r -0.55,-3.2)  {\normalsize $\alpha_{2}$};
			\node at (-2+10*\r -0.55,-3.2) {\normalsize $\alpha_{3}$};
			
			\foreach \x in {2*\r,6*\r,10*\r} {
				\draw[->,thin] (-2+\x+0.25,-2.95) arc(10:230:0.25);
			}
			
		\end{tikzpicture}
		\caption{Brown's system of three pairs of massless pulleys. The lower pulley
			of each pair is fixed and carries the generalized coordinate $\alpha_{i}$;
			the upper pulley of each pair is attached to a mass $m$ and to a spring of
			stiffness $k$, and a single cord threads the three pairs and closes into a
			loop. The gauge direction $w=(1,1,1)^{\top}$ rotates the three fixed pulleys
			by equal amounts.}
		\label{fig:pulley3}
	\end{figure}
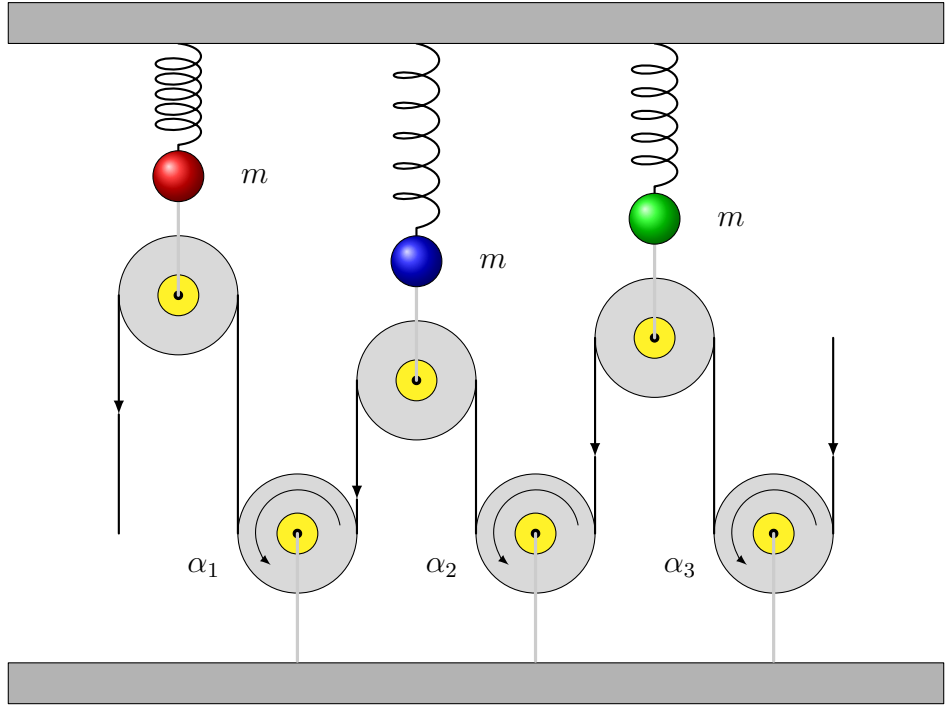
	
	The mechanical system is Brown's loop of cord weaving through three pairs of
	massless pulleys \cite{Brown2023}. In each pair the lower pulley is fixed and
	the upper pulley is attached to a mass and to a spring; the cord runs over and
	under the pulleys and its two ends are identified, so that it forms a single
	continuous loop (Figure~\ref{fig:pulley3}). The generalized coordinates are the
	angles $\alpha_{1},\alpha_{2},\alpha_{3}$ of the three lower, fixed pulleys, all
	of radius $R$, and the three suspended masses are equal. Since the pulleys are
	massless, the parameter $m$ is the suspended mass and not a moment of inertia:
	all the inertia of the system resides in the three masses.
	
	The height of each mass is fixed by a difference of two angles,
	$h_{1}=R(\alpha_{3}-\alpha_{1})/2+c$, $h_{2}=R(\alpha_{1}-\alpha_{2})/2+c$ and
	$h_{3}=R(\alpha_{2}-\alpha_{3})/2+c$ with $c$ constant, because raising one side
	of a pair lowers the other. The kinetic energy
	$\tfrac m2(\dot h_{1}^{2}+\dot h_{2}^{2}+\dot h_{3}^{2})$ and the spring energy
	therefore depend on the angles only through their differences, and the singular
	second-order Lagrangian of the system is
	\begin{equation}
		L=\frac{mR^{2}}{8}\Bigl[(\dot\alpha_{1}-\dot\alpha_{2})^{2}
		+(\dot\alpha_{2}-\dot\alpha_{3})^{2}
		+(\dot\alpha_{3}-\dot\alpha_{1})^{2}\Bigr]
		-\frac{kR^{2}}{8}\Bigl[(\alpha_{1}-\alpha_{2})^{2}
		+(\alpha_{2}-\alpha_{3})^{2}
		+(\alpha_{3}-\alpha_{1})^{2}\Bigr].
		\label{eq:pulley3}
	\end{equation}
	The gravitational term is a constant for this arrangement and has been dropped.
	Equation \eqref{eq:pulley3} is the starting point of the analysis, not yet the
	first-order Lagrangian of the Faddeev--Jackiw formulation; the latter is the
	system \eqref{eq:Lfo} produced by Proposition~\ref{prop:template}.
	
	\paragraph{Kinetic structure.} For $x\in\R^{3}$,
	$\sum_{i<j}(x_{i}-x_{j})^{2}=x^{\top}\Lambda x$ with
	$\Lambda=3I_{3}-\mathbf{1}\mathbf{1}^{\top}$. Hence, in the notation of
	\eqref{eq:Lmech},
	\begin{equation}
		K=\frac{mR^{2}}{4}\,\Lambda,
		\qquad
		V=\frac{kR^{2}}{8}\,\alpha^{\top}\Lambda\,\alpha .
		\label{eq:pulley3K}
	\end{equation}
	Here $K$ is the velocity Hessian \eqref{eq:Khessian} of \eqref{eq:pulley3}; it
	is the degeneracy of this matrix, and not yet that of any presymplectic form,
	that is being computed. Both $K$ and $V$ carry the parameters in the
	factorized form displayed, and they are kept so.
	
	\paragraph{Kernel and parameter locus.} Since
	$\Lambda=3\bigl(I_{3}-\tfrac13\mathbf{1}\mathbf{1}^{\top}\bigr)=3P$ with $P$ the
	orthogonal projector onto $\mathbf{1}^{\perp}$, the eigenvalues of $\Lambda$ are
	$0,3,3$ and $\Lambda\mathbf{1}=0$. Hence, \emph{provided} $mR^{2}\neq0$,
	\begin{equation}
		\rank K=2,
		\qquad
		\ker K=\spn\{w\},\qquad w=(1,1,1)^{\top},
		\qquad c=1 .
		\label{eq:pulley3ker}
	\end{equation}
	The proviso is not decorative. On the locus $mR^{2}=0$ one has $K\equiv0$,
	$\rank K=0$ and $c=3$: the rank drops and the hypotheses of
	Proposition~\ref{prop:template} hold with a different $c$, so that branch is a
	different problem and is set aside rather than absorbed. All that follows is
	asserted for $mR^{2}\neq0$, and $w$ is left unnormalized so that no parameter
	is introduced by the choice of basis.
	
	\paragraph{Pseudoinverse and first-order form.} From $K=\tfrac{3mR^{2}}{4}P$ and
	$P^{2}=P$ one gets, again for $mR^{2}\neq0$,
	\begin{equation}
		K^{+}=\frac{4}{3mR^{2}}\,P
		=\frac{4}{3mR^{2}}\Bigl(I_{3}-\tfrac13\mathbf{1}\mathbf{1}^{\top}\Bigr),
		\qquad
		\tfrac12\,p^{\top}K^{+}p
		=\frac{2}{3mR^{2}}\Bigl(|p|^{2}-\tfrac13(p_{1}+p_{2}+p_{3})^{2}\Bigr).
		\label{eq:pulley3Kplus}
	\end{equation}
	Lemma~\ref{lem:legendre} then gives the extended potential and the first-order
	Lagrangian on $\xi=(\alpha_{1},\alpha_{2},\alpha_{3},p_{1},p_{2},p_{3},\lambda)
	\in\R^{7}$:
	\begin{align}
		V_{\mathrm{ext}}&=\frac{2}{3mR^{2}}
		\Bigl(|p|^{2}-\tfrac13(p_{1}+p_{2}+p_{3})^{2}\Bigr)
		+\frac{kR^{2}}{8}\,\alpha^{\top}\Lambda\,\alpha
		+\lambda\,(p_{1}+p_{2}+p_{3}),
		\label{eq:pulley3Vext}\\[0.2em]
		L&=p_{1}\dot\alpha_{1}+p_{2}\dot\alpha_{2}+p_{3}\dot\alpha_{3}
		-V_{\mathrm{ext}} .
		\label{eq:pulley3LFJ}
	\end{align}
	
	\paragraph{Presymplectic matrix and its kernel.} The one-form of
	\eqref{eq:pulley3LFJ} is $a=(p_{1},p_{2},p_{3},0,0,0,0)$, so by
	\eqref{eq:templatef0}
	\begin{equation}
		f^{(0)}=\begin{pmatrix}0&-I_{3}&0\\ I_{3}&0&0\\ 0&0&0\end{pmatrix}
		\in\R^{7\times7},
		\qquad
		\rank f^{(0)}=6,
		\qquad
		\ker f^{(0)}=\spn\{\partial_{\lambda}\},
		\label{eq:pulley3f0}
	\end{equation}
	and $d=1$. Note that $\rank f^{(0)}=6$ while $\rank K=2$: the two matrices
	measure different things, as Remark~\ref{rem:KvsF} states.
	
	\paragraph{Potential invariance.} Since $\Lambda$ is symmetric and
	$\Lambda w=0$,
	\begin{equation}
		w\!\cdot\!\nabla V=\frac{kR^{2}}{4}\,w^{\top}\Lambda\,\alpha=0
		\qquad\text{identically.}
		\label{eq:pulley3inv}
	\end{equation}
	
	\paragraph{Constraint, bordering, and the reduced data.} Contracting the null
	direction $\partial_{\lambda}$ of \eqref{eq:pulley3f0} with
	$\nabla V_{\mathrm{ext}}$ gives the generated constraint
	\begin{equation}
		\Om=\phi=\partial_{\lambda}V_{\mathrm{ext}}=p_{1}+p_{2}+p_{3}\approx0,
		\label{eq:pulley3phi}
	\end{equation}
	whose gradient in the ordering $(\alpha,p,\lambda)$ is
	$B=(0,0,0,1,1,1,0)^{\top}$. Bordering with this single column produces the
	$8\times8$ matrix $f^{(1)}$, and the reduced data are
	\begin{equation}
		\Gamma=N^{\top}B=\partial_{\lambda}(p_{1}+p_{2}+p_{3})=0,
		\qquad
		M=(0),
		\qquad
		\dim\ker f^{(1)}=1+1-0=2,
		\label{eq:pulley3GM}
	\end{equation}
	with $M=0$ here for two independent reasons: it is $1\times1$ antisymmetric,
	and $\Om$ is linear in the momenta alone. By
	\eqref{eq:templatekerbasis} a basis of $\ker f^{(1)}$ is
	\begin{equation}
		\bigl(\partial_{\lambda},\,0\bigr),
		\qquad
		\bigl(u_{0},\,-1\bigr)
		\quad\text{with}\quad
		u_{0}=\partial_{\alpha_{1}}+\partial_{\alpha_{2}}+\partial_{\alpha_{3}} .
		\label{eq:pulley3ker1}
	\end{equation}
	
	\paragraph{Chain and generator.} The first chain equation is satisfied by
	construction, $f^{(0)}u_{0}=(0,w,0)=\nabla\Om$, and the decisive contraction is
	\eqref{eq:pulley3inv}, which vanishes identically and for every value of
	$k$ and $R$. The chain therefore closes, the gauge transformation is
	$\delta\alpha_{1}=\delta\alpha_{2}=\delta\alpha_{3}=\rho(t)$,
	$\delta p_{i}=0$, $\delta\lambda=\dot\rho(t)$, and the canonical charge for the
	pairs $\alpha_{i}\leftrightarrow p_{i}$ is $G=p_{1}+p_{2}+p_{3}=\phi$.
	
	Geometrically the null direction is the simultaneous rotation of the three fixed
	pulleys by equal amounts. This makes the cord circulate around the loop while
	every mass stays at the same height, so neither the springs, which respond only
	to angle differences, nor the kinetic term, which is built from the same
	differences, can detect it. The system has two physical degrees of freedom.
	
	Brown reaches the same primary constraint $\phi=p_{1}+p_{2}+p_{3}$ and the same
	gauge transformation $\delta\alpha_{i}=\epsilon$ by the Dirac--Bergmann route,
	finding $\phi$ first class because its consistency condition is satisfied
	identically \cite{Brown2023}. In the present formulation that conclusion is
	reached instead through the closure of the two-step chain, and the agreement is
	an external check on the criterion rather than an ingredient of it.
	
	\subsection{Two pulleys and one mass}
	\label{sec:pulley2}
	
	\begin{figure}[htbp]
		\centering
		\begin{tikzpicture}[scale=2.25, line cap=round, line join=round, >=latex]
			
			\def\H{4}           
			\def\Rp{0.6}        
			\def\rs{0.4}        
			\def\massSize{2mm}  
			\def\barHeight{1.0} 
			
			\filldraw[fill=black!30] (-1.5,\H-0.12) rectangle (3,\H+0.12);
			
			\coordinate (P1c) at (0, \H-\barHeight);
			\coordinate (P2c) at (\Rp+\rs, \H-2.5);
			
			\filldraw[fill=black!15] (P1c) circle (\Rp);
			\filldraw[fill=yellow!90] (P1c) circle (0.12);
			\filldraw[fill=black] (P1c) circle (0.025);
			
			\draw[very thick, black!30] (0,\H) -- (P1c);
			
			\filldraw[fill=black!15] (P2c) circle (\rs);
			\filldraw[fill=yellow!90] (P2c) circle (0.12);
			\filldraw[fill=black] (P2c) circle (0.025);
			
			\coordinate (P1right) at ($(P1c) + (\Rp,0)$);
			\coordinate (P1left)  at ($(P1c) + (-\Rp,0)$);
			\coordinate (P2left)  at ($(P2c) + (-\rs,0)$);
			\coordinate (P2right) at ($(P2c) + (\rs,0)$);
			
			\draw[thick] (P1right) -- (P2left);
			\coordinate (masscenter) at ($(P2right)+(0,-1.2)$);
			\coordinate (masstop) at ($(masscenter)+(0,\massSize)$);
			\draw[thick] (P2right) -- (masstop);
			
			\shade[ball color=red, opacity=0.9] (masscenter) circle (\massSize);
			\node[below right=10pt of masscenter] {\large $m$};
			
			\draw[->, thick] 
			($(P1c) + (135:0.6*\Rp)$) 
			arc[start angle=135, end angle=-135, radius=0.6*\Rp];
			\node at ($(P1c) + (1.25*\Rp,0.7*\Rp)$) {\large $b_1$};
			
			\draw[->, thick] 
			($(P2c) + (135:0.6*\rs)$) 
			arc[start angle=135, end angle=-130, radius=0.6*\rs];
			\node at ($(P2c) + (1.5*\rs,0.7*\rs)$) {\large $b_2$};
			
			\draw[dashed, gray] (P1c) -- ++(0,-2.5);
			\draw[dashed, gray] (P2c) -- ++(0,-1.5);
			
		\end{tikzpicture}
		\caption{Two massless, frictionless pulleys with a suspended mass. The axis
			of the upper pulley is fixed and the lower pulley is free to move
			vertically; $b_{1}$ and $b_{2}$ are the orientation angles and the two radii
			need not be equal. The gauge direction $w=(-R_{2},R_{1})^{\top}$ depends on
			the geometry of the apparatus.}
		\label{fig:pulley2}
	\end{figure}
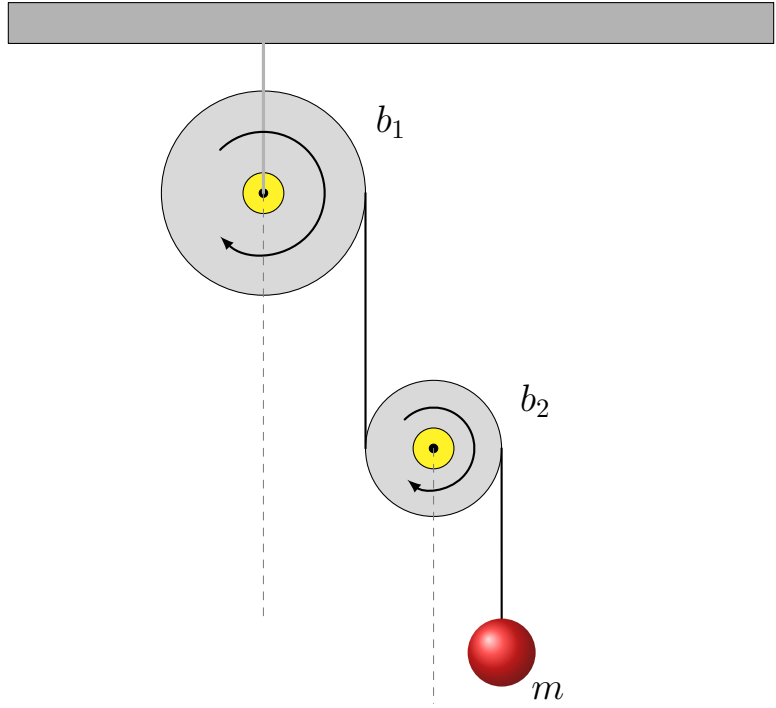
	
	Brown poses this configuration as an exercise: two massless, frictionless
	pulleys, the axis of the upper one fixed and the lower one free to move
	vertically, with a mass $m$ also restricted to move vertically and with radii
	$R_{1}$ and $R_{2}$ that need not be equal \cite{Brown2023}. We take the
	orientation angles $b_{1},b_{2}$ as generalized coordinates
	(Figure~\ref{fig:pulley2}). The model that follows, and in particular the
	reduction to a single physical combination, is ours; Brown states the mechanical
	problem and does not develop it in the text.
	
	The winding of the cord makes the height of the mass depend on the two angles
	only through one combination, which we write as
	\begin{equation}
		h=-(R_{1}b_{1}+R_{2}b_{2})
		\label{eq:pulley2h}
	\end{equation}
	so that the singular second-order Lagrangian is
	\begin{equation}
		L=\frac m2\bigl(R_{1}\dot b_{1}+R_{2}\dot b_{2}\bigr)^{2}-mg\,h,
		\qquad
		K=m\begin{pmatrix}R_{1}^{2}&R_{1}R_{2}\\ R_{1}R_{2}&R_{2}^{2}\end{pmatrix}.
		\label{eq:pulley2}
	\end{equation}
	
	Here $K$ is again the velocity Hessian \eqref{eq:Khessian}, this time a
	$2\times2$ matrix, and it is written as it comes out of the differentiation,
	with the parameters in place.
	
	\paragraph{Kernel and parameter locus.} With $v=(R_{1},R_{2})^{\top}$ one has
	$K=m\,vv^{\top}$, a rank-one matrix whenever $m\neq0$ and $v\neq0$. Assuming
	therefore
	\begin{equation}
		m\neq0,
		\qquad
		|v|^{2}=R_{1}^{2}+R_{2}^{2}\neq0,
		\label{eq:pulley2locus}
	\end{equation}
	we obtain
	\begin{equation}
		\rank K=1,
		\qquad
		\ker K=\spn\{w\},\qquad w=\begin{pmatrix}-R_{2}\\ R_{1}\end{pmatrix},
		\qquad c=1 .
		\label{eq:pulley2ker}
	\end{equation}
	On the excluded locus $m=0$ the kinetic term vanishes and $c=2$; on
	$R_{1}=R_{2}=0$ the coordinates decouple from the mechanics altogether and
	again $c=2$. Both are genuinely different systems and neither is silently
	merged with \eqref{eq:pulley2ker}. Note also that $w$ is nonzero precisely on
	the locus \eqref{eq:pulley2locus}, so no normalization of $w$ is performed and
	no division by $|v|$ is introduced at this stage.
	
	\paragraph{Pseudoinverse and first-order form.} For a rank-one symmetric
	$K=m\,vv^{\top}$ the Moore--Penrose inverse is
	\begin{equation}
		K^{+}=\frac{vv^{\top}}{m\,|v|^{4}},
		\qquad
		\tfrac12\,p^{\top}K^{+}p=\frac{(v\!\cdot\!p)^{2}}{2m\,|v|^{4}},
		\label{eq:pulley2Kplus}
	\end{equation}
	as one checks from $KK^{+}K=K$. Both expressions are defined exactly on the
	locus \eqref{eq:pulley2locus}. Lemma~\ref{lem:legendre} then gives, on
	$\xi=(b_{1},b_{2},p_{b_{1}},p_{b_{2}},\lambda)\in\R^{5}$,
	\begin{align}
		V_{\mathrm{ext}}&=\frac{(R_{1}p_{b_{1}}+R_{2}p_{b_{2}})^{2}}
		{2m\,(R_{1}^{2}+R_{2}^{2})^{2}}
		-mg\,(R_{1}b_{1}+R_{2}b_{2})
		+\lambda\,\bigl(R_{1}p_{b_{2}}-R_{2}p_{b_{1}}\bigr),
		\label{eq:pulley2Vext}\\[0.2em]
		L&=p_{b_{1}}\dot b_{1}+p_{b_{2}}\dot b_{2}-V_{\mathrm{ext}},
		\label{eq:pulley2LFJ}
	\end{align}
	where $V=mg\,h=-mg\,(R_{1}b_{1}+R_{2}b_{2})$ has been substituted from
	\eqref{eq:pulley2h}.
	
	\paragraph{Presymplectic matrix and its kernel.} The one-form is
	$a=(p_{b_{1}},p_{b_{2}},0,0,0)$, so
	\begin{equation}
		f^{(0)}=\begin{pmatrix}0&-I_{2}&0\\ I_{2}&0&0\\ 0&0&0\end{pmatrix}
		\in\R^{5\times5},
		\qquad
		\rank f^{(0)}=4,
		\qquad
		\ker f^{(0)}=\spn\{\partial_{\lambda}\} .
		\label{eq:pulley2f0}
	\end{equation}
	
	\paragraph{Potential invariance.} From \eqref{eq:pulley2h},
	$\nabla h=-(R_{1},R_{2})^{\top}=-v$, so
	\begin{equation}
		w\!\cdot\!\nabla V=mg\,(w\!\cdot\!\nabla h)=-mg\,(w\!\cdot\!v)
		=-mg\,(-R_{2}R_{1}+R_{1}R_{2})=0 .
		\label{eq:pulley2inv}
	\end{equation}
	The vanishing is structural and not parametric: it follows from $w\perp v$,
	which holds identically, and not from any factor being set to zero. In
	particular the prefactor $mg$ is retained and no branch of parameter space is
	created by \eqref{eq:pulley2inv}.
	
	\paragraph{Constraint, bordering, and generator.} Contracting
	$\partial_{\lambda}$ with $\nabla V_{\mathrm{ext}}$ gives the generated
	constraint
	\begin{equation}
		\phi=\Om=w\!\cdot\!p=R_{1}p_{b_{2}}-R_{2}p_{b_{1}}\approx0,
		\label{eq:pulley2phi}
	\end{equation}
	with gradient $B=(0,0,-R_{2},R_{1},0)^{\top}$, whence
	$\Gamma=\partial_{\lambda}\Om=0$ and $M=(0)$, so
	$\dim\ker f^{(1)}=1+1-0=2$ with basis $(\partial_{\lambda},0)$ and
	$(u_{0},-1)$, $u_{0}=-R_{2}\partial_{b_{1}}+R_{1}\partial_{b_{2}}$. The first
	chain equation holds, $f^{(0)}u_{0}=(0,w,0)=\nabla\Om$, and the decisive
	contraction is \eqref{eq:pulley2inv}, which vanishes. The chain closes and
	the gauge transformation is
	\begin{equation}
		\delta b_{1}=-R_{2}\,\rho(t),
		\qquad
		\delta b_{2}=R_{1}\,\rho(t),
		\qquad
		\delta\lambda=\dot\rho(t),
		\label{eq:pulley2gauge}
	\end{equation}
	under which $\delta h=-(R_{1}\delta b_{1}+R_{2}\delta b_{2})=0$, and the charge
	is $G=\phi$.
	
	\paragraph{Interpretation.} Within the present formulation the origin of the
	degeneracy is transparent. Only the combination $R_{1}b_{1}+R_{2}b_{2}$
	determines the physical height, so the orthogonal direction $w$ changes the two
	pulley angles without changing that combination and therefore without any
	inertial response; this is what makes $K$ singular. The Faddeev--Jackiw
	reformulation records $w$ as a null direction of the presymplectic matrix of the
	associated first-order system, and the chain criterion then decides, through
	$w\!\cdot\!\nabla V=0$, whether that null direction actually produces gauge
	freedom. The reading in terms of $\ker K$ and the chain criterion belong to the
	present work; what is taken from \cite{Brown2023} is the mechanical
	configuration.
	
	The example is worth recording because the gauge direction is not a symmetry of
	the coordinate labels but a function of the geometric parameters: rotating the
	two pulleys in the ratio $-R_{2}:R_{1}$ leaves the cord length, and hence the
	physical configuration, unchanged. Unlike the permutation-type direction of
	Section~\ref{sec:pulley3}, it varies continuously with the design of the
	apparatus and would be missed by any analysis that looked only for symmetries of
	the index set. One physical degree of freedom survives, namely the height $h$.
	
	\subsection{Four masses on a square: the gauge case}
	\label{sec:square-gauge}
	
	\begin{figure}[ht]
		\centering
		\begin{tikzpicture}[
			scale=2.5,
			x={(-0.5cm,-0.3cm)},
			y={(1cm,0cm)},
			z={(0cm,1cm)},
			line cap=round,
			line join=round
			]
			
			\coordinate (B1) at (0,0,0);
			\coordinate (B2) at (3,0,0);
			\coordinate (B3) at (3,2,0);
			\coordinate (B4) at (0,2,0);
			
			\coordinate (T1) at (0,0,2);
			\coordinate (T2) at (3,0,2);
			\coordinate (T3) at (3,2,2);
			\coordinate (T4) at (0,2,2);
			
			\fill[gray!50,draw=black] (B1) -- (B2) -- (B3) -- (B4) -- cycle;
			\fill[gray!50,draw=black] (T1) -- (T2) -- (T3) -- (T4) -- cycle;
			
			\coordinate (Bcenter) at (1.5,1,0);
			\coordinate (Tcenter) at (1.5,1,2);
			
			\def\alpha{0.2}
			
			\coordinate (B1a) at ($(B1)! \alpha ! (Bcenter)$);
			\coordinate (B2a) at ($(B2)! \alpha ! (Bcenter)$);
			\coordinate (B3a) at ($(B3)! \alpha ! (Bcenter)$);
			\coordinate (B4a) at ($(B4)! \alpha ! (Bcenter)$);
			
			\coordinate (T1a) at ($(T1)! \alpha ! (Tcenter)$);
			\coordinate (T2a) at ($(T2)! \alpha ! (Tcenter)$);
			\coordinate (T3a) at ($(T3)! \alpha ! (Tcenter)$);
			\coordinate (T4a) at ($(T4)! \alpha ! (Tcenter)$);
			
			\draw[line width=2pt, line cap=round, gray!40,
			preaction={draw, line width=3pt, black}]
			(B1a) -- (T1a);
			\draw[line width=2pt, line cap=round, gray!40,
			preaction={draw, line width=3pt, black}]
			(B2a) -- (T2a);
			\draw[line width=2pt, line cap=round, gray!40,
			preaction={draw, line width=3pt, black}]
			(B3a) -- (T3a);
			\draw[line width=2pt, line cap=round, gray!40,
			preaction={draw, line width=3pt, black}]
			(B4a) -- (T4a);
			
			\pgfmathsetmacro{\fYone}{0.70} 
			\pgfmathsetmacro{\fYtwo}{0.60} 
			\pgfmathsetmacro{\fYthree}{0.65}
			\pgfmathsetmacro{\fYfour}{0.55}
			
			\coordinate (Y1) at ($(T1a)!\fYone!(B1a)$);
			\coordinate (Y2) at ($(T2a)!\fYtwo!(B2a)$);
			\coordinate (Y3) at ($(T3a)!\fYthree!(B3a)$);
			\coordinate (Y4) at ($(T4a)!\fYfour!(B4a)$);
			
			\draw[blue, decorate, decoration={coil, segment length=5pt, amplitude=4pt}] (T1a) -- (Y1);
			\draw[blue, decorate, decoration={coil, segment length=5pt, amplitude=4pt}] (T2a) -- (Y2);
			\draw[blue, decorate, decoration={coil, segment length=5pt, amplitude=4pt}] (T3a) -- (Y3);
			\draw[blue, decorate, decoration={coil, segment length=5pt, amplitude=4pt}] (T4a) -- (Y4);
			
			\fill (Y1) circle (0.8pt)
			node[above=3pt,left] {$y_1$};
			
			\fill (Y2) circle (0.8pt)
			node[left] {$y_2$};
			
			\fill (Y3) circle (0.8pt)
			node at ($(Y3)+(0.58,0.14,0.05)$) {$y_3$};
			
			\fill (Y4) circle (0.8pt)
			node at ($(Y4)+(-0.12,0.12,0)$) {$y_4$};
			
			
			\draw[thick, double=teal!50] (Y1) -- (Y2) -- (Y3) -- (Y4) -- cycle;
			
			\foreach \A/\B in {Y1/Y2, Y2/Y3, Y3/Y4, Y4/Y1} {
				\coordinate (Mid) at ($(\A)!0.5!(\B)$);
				\shade[ball color=red] (Mid) circle (3pt);
			}
			
		\end{tikzpicture}
		\caption{Four masses fixed at the midpoints of four massless, freely
			extensible rods, whose ends are connected at corners that slide on vertical
			posts. The corner heights are $y_{1},\dots,y_{4}$. In the version shown here
			the springs are attached to the masses, and the system carries the gauge
			symmetry generated by $w=(-1,1,-1,1)^{\top}$.}
		\label{fig:brown_4masses}
	\end{figure}
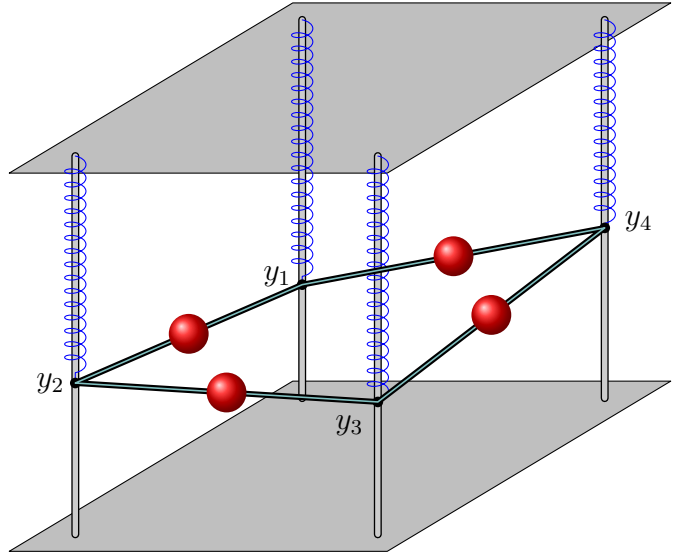
	
	Brown's four-mass system consists of four massless, freely extensible rods, each
	carrying a mass fixed at its midpoint \cite{Brown2023}. The rods are joined into
	a closed square at four corners which slide freely on vertical posts, and
	$y_{1},\dots,y_{4}$ are the heights of those corners
	(Figure~\ref{fig:brown_4masses}). A freely extensible rod is rigid transversally
	but has no fixed length, so the corners move independently and the mass on the
	rod joining corners $i$ and $j$ sits at the height $(y_{i}+y_{j})/2$. Springs of
	stiffness $k$ connect the assembly to the ceiling, and Brown distinguishes two
	placements for them; the placement is exactly what decides whether the system
	has gauge freedom. In the present subsection the springs are attached to the
	masses. The alternative, springs attached to the corners, is the control of
	Section~\ref{sec:square-control}.
	
	With the springs on the masses, the singular second-order Lagrangian is
	$L_{\mathrm{mech}}=T-V$ with
	\begin{align}
		T&=\frac m2\left[\Bigl(\frac{\dot y_{1}+\dot y_{2}}{2}\Bigr)^{2}
		+\Bigl(\frac{\dot y_{2}+\dot y_{3}}{2}\Bigr)^{2}
		+\Bigl(\frac{\dot y_{3}+\dot y_{4}}{2}\Bigr)^{2}
		+\Bigl(\frac{\dot y_{4}+\dot y_{1}}{2}\Bigr)^{2}\right],
		\label{eq:squareT}\\[0.3em]
		V&=mg\,(y_{1}+y_{2}+y_{3}+y_{4})
		+\frac{k}{8}\sum_{(i,j)\in E}\bigl(2a-y_{i}-y_{j}\bigr)^{2},
		\qquad
		E=\{(1,2),(2,3),(3,4),(4,1)\}.
		\label{eq:squareV}
	\end{align}
	
	\paragraph{Kinetic structure and kernel.} Let $S\in\R^{4\times4}$ be the
	edge--vertex matrix with rows $(1,1,0,0)$, $(0,1,1,0)$, $(0,0,1,1)$,
	$(1,0,0,1)$, so that $T=\tfrac m8\,\dot y^{\top}S^{\top}S\,\dot y$ and hence,
	by \eqref{eq:Khessian},
	\begin{equation}
		K=\frac m4\,S^{\top}S=\frac m4\bigl(2I_{4}+A\bigr),
		\qquad
		A=\begin{pmatrix}0&1&0&1\\ 1&0&1&0\\ 0&1&0&1\\ 1&0&1&0\end{pmatrix},
		\label{eq:squareK}
	\end{equation}
	$A$ being the adjacency matrix of the four-cycle: each index lies on exactly
	two edges, giving the diagonal $2$, and two indices share an edge exactly when
	they are cyclically adjacent. The eigenvalues of $A$ are $2,0,0,-2$, so those
	of $2I_{4}+A$ are $4,2,2,0$, with the null eigenvector of $2I_{4}+A$ the
	alternating mode. Equivalently, $y\in\ker S$ if and only if
	$y_{1}+y_{2}=y_{2}+y_{3}=y_{3}+y_{4}=y_{4}+y_{1}=0$, which forces
	$y_{2}=-y_{1}$, $y_{3}=y_{1}$, $y_{4}=-y_{1}$. Hence, \emph{provided}
	$m\neq0$,
	\begin{equation}
		\rank K=3,
		\qquad
		\ker K=\spn\{w\},\qquad w=(-1,1,-1,1)^{\top},
		\qquad c=1 .
		\label{eq:squareker}
	\end{equation}
	The null direction is the alternating, or checkerboard, mode. On the excluded
	locus $m=0$ the kinetic matrix vanishes identically and $c=4$; that branch is
	set aside, as in Remark~\ref{rem:parameters}.
	
	\paragraph{Pseudoinverse and first-order form.} Write
	$\mathbf{1}=(1,1,1,1)^{\top}$ and let
	\[
	P_{1}=\tfrac14\mathbf{1}\mathbf{1}^{\top},
	\qquad
	P_{w}=\tfrac14 ww^{\top},
	\qquad
	P_{2}=I_{4}-P_{1}-P_{w},
	\]
	be the orthogonal projectors onto $\spn\{\mathbf{1}\}$, $\spn\{w\}$ and their
	common orthogonal complement, which is spanned by $(1,0,-1,0)$ and
	$(0,1,0,-1)$. On these three blocks $2I_{4}+A$ acts as $4$, $0$ and $2$
	respectively, so that $K=m\,P_{1}+\tfrac m2\,P_{2}$ and, for $m\neq0$,
	\begin{equation}
		K^{+}=\frac1m\,P_{1}+\frac2m\,P_{2},
		\qquad
		\tfrac12\,p^{\top}K^{+}p
		=\frac{1}{2m}\,p^{\top}P_{1}p+\frac1m\,p^{\top}P_{2}p .
		\label{eq:squareKplus}
	\end{equation}
	Lemma~\ref{lem:legendre} then gives, on
	$\xi=(y_{1},\dots,y_{4},p_{1},\dots,p_{4},\lambda)\in\R^{9}$,
	\begin{equation}
		V_{\mathrm{ext}}=\frac{1}{2m}\,p^{\top}P_{1}p+\frac1m\,p^{\top}P_{2}p
		+V(y)+\lambda\,(w\!\cdot\!p),
		\qquad
		L=p_{i}\dot y^{i}-V_{\mathrm{ext}},
		\label{eq:squareLFJ}
	\end{equation}
	with $V$ as in \eqref{eq:squareV}. The one-form is $a=(p,0,0)$ and therefore
	\begin{equation}
		f^{(0)}=\begin{pmatrix}0&-I_{4}&0\\ I_{4}&0&0\\ 0&0&0\end{pmatrix}
		\in\R^{9\times9},
		\qquad
		\rank f^{(0)}=8,
		\qquad
		\ker f^{(0)}=\spn\{\partial_{\lambda}\} .
		\label{eq:squaref0}
	\end{equation}
	
	\paragraph{Potential invariance.} Under $\delta y=\epsilon\,w$, that is
	\begin{equation}
		\delta y_{1}=-\epsilon,\quad
		\delta y_{2}=\epsilon,\quad
		\delta y_{3}=-\epsilon,\quad
		\delta y_{4}=\epsilon,
		\label{eq:squaredelta}
	\end{equation}
	every edge sum is invariant,
	\begin{equation}
		\delta(y_{1}+y_{2})=\delta(y_{2}+y_{3})=\delta(y_{3}+y_{4})
		=\delta(y_{4}+y_{1})=0,
		\label{eq:squareedges}
	\end{equation}
	and the total height is invariant as well,
	$\delta(y_{1}+y_{2}+y_{3}+y_{4})=0$. Since $V$ in \eqref{eq:squareV} depends on
	$y$ only through the four edge sums, that is, through the heights of the four
	masses, and through the total, we conclude
	\begin{equation}
		w\!\cdot\!\nabla V=0 \qquad\text{identically.}
		\label{eq:squareinv}
	\end{equation}
	
	\paragraph{Constraint, bordering, and generator.} Proposition~\ref{prop:template}
	applies with $c=1$. Contracting $\partial_{\lambda}$ with
	$\nabla V_{\mathrm{ext}}$ gives the generated constraint
	\begin{equation}
		\Om=w\!\cdot\!p=-p_{1}+p_{2}-p_{3}+p_{4}\approx0,
		\label{eq:squarephi}
	\end{equation}
	equivalently $\phi=p_{1}-p_{2}+p_{3}-p_{4}=-\Om$, the two differing only by an
	overall sign and generating the same one-dimensional constraint. Its gradient
	is $B=(0,w,0)^{\top}$, so bordering with one column gives
	$\Gamma=\partial_{\lambda}\Om=0$, $M=(0)$ and
	$\dim\ker f^{(1)}=1+1-0=2$, with basis $(\partial_{\lambda},0)$ and
	$(u_{0},-1)$, $u_{0}=w^{i}\partial_{y^{i}}$. The first chain equation
	$f^{(0)}u_{0}=(0,w,0)=\nabla\Om$ holds, and the decisive contraction is
	\eqref{eq:squareinv}, which vanishes identically for every value of $m$, $g$,
	$k$ and $a$. The chain therefore closes, the gauge transformation is
	$\delta y=\rho(t)\,w$, $\delta p=0$,
	$\delta\lambda=\dot\rho(t)$, and the canonical charge computed with the full set
	of Darboux pairs $y_{i}\leftrightarrow p_{i}$ is
	\begin{equation}
		G=w\!\cdot\!p=-p_{1}+p_{2}-p_{3}+p_{4}=\Om .
		\label{eq:squareG}
	\end{equation}
	
	\begin{remark}[On declared canonical pairs]
		\label{rem:pairs}
		The charge produced by the construction of Section~\ref{sec:charges} is
		$\iota_{u_{0}}\omega_{\mathrm{can}}$ integrated over the pairs the user
		\emph{declares}, and it is therefore sensitive to that declaration. If only the
		three pairs $y_{1}\leftrightarrow p_{1}$, $y_{2}\leftrightarrow p_{2}$,
		$y_{3}\leftrightarrow p_{3}$ are declared, the same formula returns
		\begin{equation}
			G_{\mathrm{red}}=-p_{1}+p_{2}-p_{3},
			\label{eq:squareGred}
		\end{equation}
		which is the restriction of \eqref{eq:squareG} to the declared pairs and differs
		from the constraint \eqref{eq:squarephi} by the term $p_{4}$ attached to the
		undeclared pair. The two coincide only when all pairs are declared. This is not
		a defect of the construction but the content of hypothesis \textbf{H9}: the
		canonical structure is external data, and a partial declaration yields a partial
		charge. Expressions such as \eqref{eq:squareGred} should never be quoted as the
		first-class constraint of the system without the accompanying convention.
	\end{remark}
	
	\paragraph{Interpretation.} The full chain is visible here in one picture. The
	alternating mode is a null direction of the kinetic form because each mass sits
	at the mean of two adjacent corner heights and the mode raises one member of
	every adjacent pair exactly as much as it lowers the other, so that no mass
	moves and there is no inertial response. It is a symmetry of the potential for
	the same reason, the springs being attached to the masses and gravity coupling
	only to the total. It therefore generates the momentum constraint
	\eqref{eq:squarephi}, and the constraint is first class, so the chain closes and
	the mode is an honest gauge direction. The system has three physical degrees of
	freedom.
	
	Brown's Dirac--Bergmann treatment of this system yields the single primary
	constraint $p_{1}-p_{2}+p_{3}-p_{4}$, finds that its consistency condition is
	satisfied identically so that no secondary constraint arises, and concludes that
	the constraint is first class and generates $\delta y=\epsilon\,(1,-1,1,-1)$,
	$\delta p_{i}=0$ \cite{Brown2023}. Our null vector is minus his and the two
	constraints differ by the same overall sign, so the two descriptions agree; the
	gauge-invariant quantities he identifies, the mass positions
	$(y_{i}+y_{j})/2$, are precisely the functions annihilated by $w$. The two
	routes are not the same procedure, and the agreement of their conclusions is
	what makes this system a useful test of the chain criterion.
	
	\subsection{Four masses on a square: the non-gauge control}
	\label{sec:square-control}
	
	\begin{figure}[htbp]
		\centering
		\begin{tikzpicture}[
			scale=2.5,
			x={(-0.5cm,-0.3cm)},
			y={(1cm,0cm)},
			z={(0cm,1cm)},
			line cap=round,
			line join=round
			]
			
			\coordinate (B1) at (0,0,0);
			\coordinate (B2) at (3,0,0);
			\coordinate (B3) at (3,2,0);
			\coordinate (B4) at (0,2,0);
			
			\coordinate (T1) at (0,0,2);
			\coordinate (T2) at (3,0,2);
			\coordinate (T3) at (3,2,2);
			\coordinate (T4) at (0,2,2);
			
			\fill[gray!50,draw=black] (B1) -- (B2) -- (B3) -- (B4) -- cycle;
			\fill[gray!50,draw=black] (T1) -- (T2) -- (T3) -- (T4) -- cycle;
			
			\coordinate (Bcenter) at (1.5,1,0);
			\coordinate (Tcenter) at (1.5,1,2);
			
			\def\alpha{0.2}
			
			\coordinate (B1a) at ($(B1)! \alpha ! (Bcenter)$);
			\coordinate (B2a) at ($(B2)! \alpha ! (Bcenter)$);
			\coordinate (B3a) at ($(B3)! \alpha ! (Bcenter)$);
			\coordinate (B4a) at ($(B4)! \alpha ! (Bcenter)$);
			
			\coordinate (T1a) at ($(T1)! \alpha ! (Tcenter)$);
			\coordinate (T2a) at ($(T2)! \alpha ! (Tcenter)$);
			\coordinate (T3a) at ($(T3)! \alpha ! (Tcenter)$);
			\coordinate (T4a) at ($(T4)! \alpha ! (Tcenter)$);
			
			\draw[line width=2pt, line cap=round, gray!40,
			preaction={draw, line width=3pt, black}]
			(B1a) -- (T1a);
			\draw[line width=2pt, line cap=round, gray!40,
			preaction={draw, line width=3pt, black}]
			(B2a) -- (T2a);
			\draw[line width=2pt, line cap=round, gray!40,
			preaction={draw, line width=3pt, black}]
			(B3a) -- (T3a);
			\draw[line width=2pt, line cap=round, gray!40,
			preaction={draw, line width=3pt, black}]
			(B4a) -- (T4a);
			
			\pgfmathsetmacro{\fYone}{0.70} 
			\pgfmathsetmacro{\fYtwo}{0.60} 
			\pgfmathsetmacro{\fYthree}{0.65}
			\pgfmathsetmacro{\fYfour}{0.55}
			
			\coordinate (Y1) at ($(T1a)!\fYone!(B1a)$);
			\coordinate (Y2) at ($(T2a)!\fYtwo!(B2a)$);
			\coordinate (Y3) at ($(T3a)!\fYthree!(B3a)$);
			\coordinate (Y4) at ($(T4a)!\fYfour!(B4a)$);
			
			\draw[blue, decorate, decoration={coil, segment length=5pt, amplitude=4pt}] (T1a) -- (Y1);
			\draw[blue, decorate, decoration={coil, segment length=5pt, amplitude=4pt}] (T2a) -- (Y2);
			\draw[blue, decorate, decoration={coil, segment length=5pt, amplitude=4pt}] (T3a) -- (Y3);
			\draw[blue, decorate, decoration={coil, segment length=5pt, amplitude=4pt}] (T4a) -- (Y4);
			
			\fill (Y1) circle (0.8pt)
			node[above=3pt,left] {$y_1$};
			
			\fill (Y2) circle (0.8pt)
			node[left] {$y_2$};
			
			\fill (Y3) circle (0.8pt)
			node at ($(Y3)+(0.58,0.14,0.05)$) {$y_3$};
			
			\fill (Y4) circle (0.8pt)
			node at ($(Y4)+(-0.12,0.12,0)$) {$y_4$};
			
			
			\draw[thick, double=teal!50] (Y1) -- (Y2) -- (Y3) -- (Y4) -- cycle;
			
			\foreach \A/\B in {Y1/Y2, Y2/Y3, Y3/Y4, Y4/Y1} {
				\coordinate (Mid) at ($(\A)!0.5!(\B)$);
				\shade[ball color=red] (Mid) circle (3pt);
			}
		\end{tikzpicture}
		\caption{The same mechanical arrangement as in
			Figure~\ref{fig:brown_4masses}, with the springs moved from the masses to
			the corners. The singular kinetic structure is unchanged and the null
			direction is still $w=(-1,1,-1,1)^{\top}$, but the gauge symmetry is broken.}
		\label{fig:square-control}
	\end{figure}
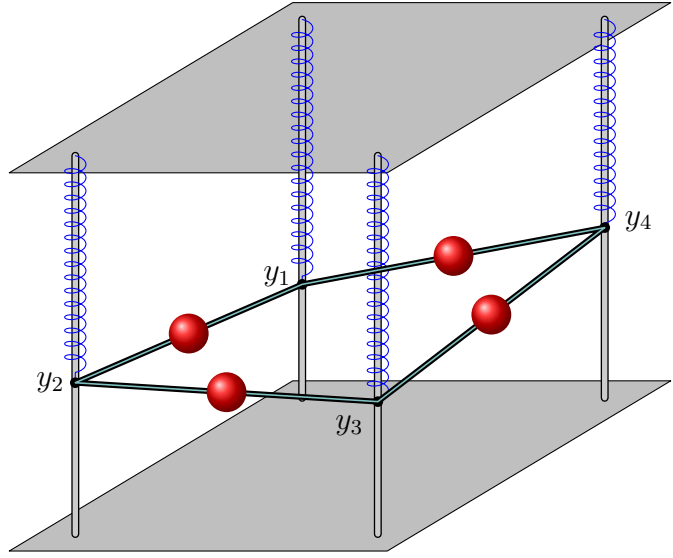
	
	Keep the mechanical arrangement and the kinetic term \eqref{eq:squareT}
	unchanged, and move the springs from the masses to the corners. This is Brown's
	second version of the same system \cite{Brown2023}: each corner is now anchored
	to the ceiling by its own spring, with $a$ the height of the ceiling minus the
	relaxed length (Figure~\ref{fig:square-control}). The potential becomes
	\begin{equation}
		V=mg\,(y_{1}+y_{2}+y_{3}+y_{4})
		+\frac{k}{2}\Bigl[(a-y_{1})^{2}+(a-y_{2})^{2}
		+(a-y_{3})^{2}+(a-y_{4})^{2}\Bigr].
		\label{eq:controlV}
	\end{equation}
	The kinetic structure, and therefore $K$, $\ker K$, $w=(-1,1,-1,1)^{\top}$,
	$\Gamma=0$, $M=0$ and $\dim\ker f^{(1)}=2$, are \emph{identical} to those of
	Section~\ref{sec:square-gauge}. Only the potential has changed. The generated constraint is again
	$\Om=w\!\cdot\!p$, and it is again first class by \eqref{eq:fccount}. Every
	algebraic invariant computed so far coincides with the gauge case.
	
	\paragraph{The chain fails.} Here, however,
	\begin{equation}
		w\!\cdot\!\nabla V
		=\sum_{i}w_{i}\bigl(mg-k(a-y_{i})\bigr)
		=mg\sum_{i}w_{i}+k\sum_{i}w_{i}(y_{i}-a)
		=k\,(y_{2}+y_{4}-y_{1}-y_{3}),
		\label{eq:controlPhi}
	\end{equation}
	using $\sum_{i}w_{i}=0$, which also removes the constant $ka$ and the
	gravitational term. The parametric factor $k$ is retained, and the two branches
	it defines are genuinely different.
	
	\emph{Branch $k\neq0$.} The contraction \eqref{eq:controlPhi} does not vanish
	identically, and by Proposition~\ref{prop:template}(e) the decisive class of
	Theorem~\ref{thm:chain-decisive} is
	\begin{equation}
		\bigl[\Phi_{u_{0}}\bigr]
		=\bigl[k\,(y_{2}+y_{4}-y_{1}-y_{3})\bigr]\neq0
		\quad\text{in }\Rring/\gen{w\!\cdot\!p},
		\label{eq:controlclass}
	\end{equation}
	because the ideal is generated by an element homogeneous of degree one in $p$
	while the representative depends on $y$ alone. No two-step chain closes. What
	happens instead is the third case of Theorem~\ref{thm:trichotomy}: the function
	\eqref{eq:controlPhi} has zeros, so
	\begin{equation}
		k\,\bigl(y_{1}+y_{3}-y_{2}-y_{4}\bigr)=0,
		\qquad\text{that is}\qquad
		y_{1}+y_{3}=y_{2}+y_{4},
		\label{eq:controlnew}
	\end{equation}
	is a genuinely new constraint and the iteration proceeds with an enlarged
	bordering. The alternating mode is not a gauge direction: it is dynamically
	obstructed by the anchoring springs, which see the individual corner heights and
	not merely their pairwise means.
	
	\emph{Branch $k=0$.} Here \eqref{eq:controlPhi} vanishes identically, the
	potential reduces to $V=mg\sum_{i}y_{i}$, which is $w$-invariant because
	$\sum_{i}w_{i}=0$, and Proposition~\ref{prop:template}(d) applies: the chain
	closes and the alternating mode \emph{is} a gauge direction. This is the
	expected mechanical statement, since with no springs nothing couples to the
	individual corner heights. The branch is recorded rather than removed: had the
	factor $k$ been cancelled from \eqref{eq:controlPhi} in order to read off the
	constraint $y_{1}+y_{3}-y_{2}-y_{4}=0$, precisely the stratum of parameter
	space on which the gauge symmetry reappears would have been deleted, in the
	manner warned against in Remark~\ref{rem:parameters}. Everything below refers
	to the branch $k\neq0$, which is the one that contrasts with
	Section~\ref{sec:square-gauge}.
	
	\paragraph{What happens at the critical stiffness.} The stratum $k=0$ deserves
	a description in dynamical terms, both because it is the sharpest instance in
	this paper of a parameter controlling gauge content and because it is easy to
	over-read. Equilibria of $K\ddot y=-\nabla V$ require $\nabla V=0$, since
	$\ddot y=0$; with $\pd{i}V=mg+k(y_{i}-a)$ this gives, for $k\neq0$, the unique
	equilibrium
	\begin{equation}
		y_{i}^{*}=a-\frac{mg}{k},
		\qquad i=1,\dots,4,
		\label{eq:controleq}
	\end{equation}
	which satisfies the constraint \eqref{eq:controlnew}. The Hessian of the
	potential is $H=k\,I_{4}$, so with $K$ as in \eqref{eq:squareK} the generalized
	spectrum is
	\begin{equation}
		\det\bigl(H-\lambda K\bigr)
		=\tfrac{k}{4}\,(k-\lambda m)\,(2k-\lambda m)^{2},
		\qquad
		\lambda\in\Bigl\{\frac{k}{m},\;\frac{2k}{m}\ \text{(twice)}\Bigr\},
		\label{eq:controlspec}
	\end{equation}
	the fourth generalized eigenvalue being infinite because $K$ is singular. Every
	finite eigenvalue is proportional to $k$.
	
	As $k\to0^{\pm}$ the equilibrium \eqref{eq:controleq} escapes to infinity and
	the whole finite spectrum collapses to zero. At $k=0$ exactly there is no
	finite equilibrium at all, since $\nabla V=mg\,\mathbf{1}\neq0$ for $g\neq0$,
	and the system falls freely along the direction of gravity while the
	alternating mode becomes gauge. For $k<0$ the equilibrium \eqref{eq:controleq}
	returns to a finite position and the spectrum is negative, so the equilibrium
	is hyperbolic. The domains must be kept apart: $k>0$ and the limit $k=0$, which
	is the ordinary springless configuration, are physically standard, whereas
	$k<0$ requires negative stiffness and belongs to an extended model.
	
	No equilibrium branch collides with another, none splits, and the eigenvalues
	reach zero only as the equilibrium leaves every bounded region. The correct
	description of $k=0$ is therefore a gauge-character transition accompanied by a
	loss of restoring force and by the loss of the finite equilibrium, and we do
	not call it a bifurcation.
	
	\begin{remark}[Loss of restoring force and gauge transition are independent]
		\label{rem:c1c3}
		The coincidence just described is a property of this family and not a
		general implication, as a second singular system shows. In Brown's compound
		spring \cite{Brown2023}, two springs of stiffnesses $k_{1}$ and $k_{2}$ in
		series carry a mass, the velocity Hessian is
		$K=m\bigl(\begin{smallmatrix}1&1\\1&1\end{smallmatrix}\bigr)$ with
		$\ker K=\spn\{(1,-1)^{\top}\}$, and the contraction along the null direction
		is $w\!\cdot\!\nabla V=k_{1}(x_{1}-\ell_{1})-k_{2}(x_{2}-\ell_{2})$, a
		secondary constraint with no global parametric factor. The effective
		stiffness of the physical coordinate $x_{1}+x_{2}$ is the series combination
		$k_{\mathrm{eff}}=k_{1}k_{2}/(k_{1}+k_{2})$, so on the locus $k_{1}k_{2}=0$,
		which is physically standard and simply means removing one spring, the
		restoring force vanishes, the frequency goes to zero and the equilibrium
		escapes to infinity, exactly as at $k=0$ above. Yet the bracket of the two
		constraints is $k_{1}+k_{2}$, which does not vanish there, so the pair
		remains second class and there is no gauge transition. Loss of restoring
		force and gauge-character transition are therefore logically independent,
		and their coincidence in the four-mass system is a feature of that potential
		rather than a general mechanism.
		
		The complementary locus $k_{1}+k_{2}=0$ of the same system is worth a line
		for contrast. There the bracket vanishes, the two constraints become first
		class, the consistency condition no longer determines the multiplier and
		produces instead a further constraint, and $k_{\mathrm{eff}}$ diverges while
		the equilibrium position stays finite, so that the equilibrium reverses
		stability. This is a change of constraint structure accompanied by a
		stability reversal, and again not a bifurcation: the eigenvalue passes
		through infinity rather than through zero. It also requires
		$k_{2}=-k_{1}$, hence a negative stiffness, and belongs to the extended
		model rather than to a pair of ordinary passive springs.
	\end{remark}
	
	Relation \eqref{eq:controlnew} is, up to an overall factor, the secondary
	constraint that Brown obtains for this system from the consistency condition of
	the Dirac--Bergmann algorithm; he likewise finds the two constraints to be
	second class and the system to have no gauge freedom \cite{Brown2023}. He also
	gives the mechanical reading of the secondary constraint: the corners can be
	displaced in the alternating pattern without moving any mass, so the springs
	snap them into the configuration that minimizes the potential at fixed mass
	positions. In the present formulation the same obstruction appears as the
	non-vanishing of the class \eqref{eq:controlclass}.
	
	\begin{remark}[The two implications]
		\label{rem:twoimplications}
		The two systems of Figures~\ref{fig:brown_4masses} and
		\ref{fig:square-control} share the same mechanical arrangement and the same
		kinetic form, and therefore the same null mode, the same $\Gamma$, the same $M$
		and the same bordered kernel. For $k\neq0$ they differ only in where the
		springs are attached, and they differ completely in their gauge content. The pair therefore establishes, within one
		mechanical family:
		\[
		\text{singular kinetic structure}\;\nRightarrow\;\text{gauge symmetry},
		\]
		\[
		\text{null mode}\;+\;\text{potential invariance}
		\;\Longrightarrow\;\text{gauge generator},
		\]
		the second under the hypotheses of Proposition~\ref{prop:template}.
		This second implication is not a tautology: what
		Proposition~\ref{prop:template}(d) supplies is precisely that, for this class of
		systems, the condition $w\!\cdot\!\nabla V\equiv0$ is not merely necessary but
		also sufficient for the chain to close, so that here, unlike in the general
		situation of Section~\ref{sec:weak}, potential invariance can be used as a
		complete criterion. The reason is structural: $\Gamma$ and $M$ both vanish, so
		the only surviving obstruction is the one carried by the chain. Outside this
		structural class no such shortcut is available: membership of the
		contractions in $\Ical$ is only necessary
		(Corollary~\ref{cor:necsuf}), and the decisive statement remains
		Theorem~\ref{thm:chain-decisive}.
	\end{remark}
	
	\section{Discussion, scope and hypotheses}
	\label{sec:scope}
	
	\subsection{Catalogue of assumptions, scope conditions, and conventions}
	
	\begin{center}
		\renewcommand{\arraystretch}{1.2}
		\begin{tabular}{@{}p{0.05\textwidth}p{0.28\textwidth}p{0.27\textwidth}p{0.28\textwidth}@{}}
			\hline
			& Hypothesis & Used in & Status \\
			\hline
			H1 & locally constant rank of $f$ & Rem.~\ref{rem:darboux}, Thm.~\ref{thm:frobenius} & in theorem statements \\
			H2 & irreducible constraints, $\rank B=k$ & Prop.~\ref{prop:spurious} & in theorem statements \\
			H3 & locally constant kernel frame & \eqref{eq:hessian}, Prop.~\ref{prop:Tterm} & only for the Hessian form and symmetry of $\Gamma$ \\
			H4 & $\Sigma$ regular & Lem.~\ref{lem:pullback} & in theorem statements \\
			H5 & $\Ical$ radical & Thm.~\ref{thm:trichotomy}(2) & scope; conservative if it fails \\
			H6 & real consistency, $\Sigma\neq\emptyset$ & all dynamical readings & scope; see Prop.~\ref{prop:real} \\
			H7 & constant-coefficient combinations & Prop.~\ref{prop:syzygy} & scope; conservative if it fails \\
			H8 & linearity in velocities & \eqref{eq:lag}, \eqref{eq:f} & structural \\
			H9 & declared canonical structure & Sec.~\ref{sec:charges}, Rem.~\ref{rem:pairs} & scope; user input \\
			H10 & exhausted iteration & Thm.~\ref{thm:trichotomy} & scope \\
			H11 & initial data on $\Sigma$ & Rem.~\ref{rem:borderimposes} & scope \\
			H12 & effective ring $\Rring$ & all ideal statements & structural \\
			\hline
		\end{tabular}
	\end{center}
	
	Hypotheses H1, H2, H3, H4 and H8 appear in the statements of the theorems that
	use them. H5, H6, H7, H9, H10, H11 and H12 are conditions of interpretation or
	of effectivity and are collected here rather than inserted artificially into
	algebraic statements. In particular, no algebraic result of
	Sections~\ref{sec:kernel}--\ref{sec:firstclass} requires radicality,
	consistency or exhaustion: those are pure linear algebra over $\R$.
	
	\subsection{Relation to Dirac--Bergmann theory}
	
	Theorem~\ref{thm:dirac} places the bordered formalism inside the Dirac picture
	\cite{Dirac1964,HT1992,Sundermeyer1982,GT1990,RR1997}: the degenerate directions
	of $f^{(0)}$ correspond to primary constraints, the consistency conditions to
	secondary ones, and bordering computes the primary--secondary block $\Gamma$ of
	the reduced constraint matrix $C$ of Theorem~\ref{thm:dirac}. The economy of the Faddeev--Jackiw route is now precisely
	quantifiable: instead of the $(d+k)\times(d+k)$ matrix $C$ one handles the
	$d\times k$ block $\Gamma$, at the cost of losing direct access to the
	secondary--secondary block $M$, which, as Section~\ref{sec:weak} shows, must
	be recovered whenever the first-class count is at issue. The price of the
	reduction is recorded in Proposition~\ref{prop:invariance}: $M$ is defined only
	relative to a choice of second-class splitting, and only its compression
	$\Pi M\Pi$ to $\ker\Gamma$ is intrinsic.
	
	The correspondence should not be overstated. It is a dictionary between two
	constructions, not the assertion that the two algorithms are the same
	procedure. In particular $\Gamma$ is one block of the reduced constraint matrix
	$C$ and not the Dirac constraint matrix, contrary to what was asserted in
	\cite{CLMCP2026}, and the identification of first-class character requires the
	further condition $\Pi M\Pi w=0$ of Theorem~\ref{thm:firstclass}, which no
	computation of $\Gamma$ alone can supply.
	The mechanical systems of Section~\ref{sec:mechanical} illustrate the dictionary
	in both directions: Brown analyses them by the Dirac--Bergmann algorithm
	\cite{Brown2023}, we analyse them by bordering and the chain criterion, and the
	conclusions agree in every case examined. Agreement of conclusions on a family
	of examples is evidence that the dictionary is correctly drawn; it is not a
	proof that the two routes coincide in general, and we do not claim one.
	
	\subsection{Open points}
	
	\begin{enumerate}[label=(O\arabic*),leftmargin=2.6em,itemsep=0.35em]
		\item\label{open:chains} \emph{Chains of length $S\geq2$.} We have the
		general condition \eqref{eq:chaincond} and the complete analysis for $S=1$.
		A general solvability tower, with obstruction classes in $\Rring/\Ical$ at
		each level, is expected but not proved here.
		\item \emph{Functional coefficients.} The determination of $\mathcal{G}$
		rather than $\mathcal{G}_{\R}$ requires syzygy computations modulo $\Ical$;
		the algorithmic cost and the resulting classification are unexplored.
		\item \emph{Degeneracy strata.} We treat $\Delta$ as an object to be excluded.
		A theory of gauge structure \emph{on} the strata, where the kernel jumps,
		remains to be developed.
		\item \emph{Comparison with Gotay--Nester.} The precise relation between the
		halting surface of the bordered iteration and the final constraint manifold
		of \cite{GNH1978} is stated here only as a hypothesis; a theorem giving
		sufficient conditions for the two to coincide would close
		Remark~\ref{rem:gotay}.
		\item \emph{Field theory.} Everything is finite-dimensional. The extension to
		field-theoretic models requires functional-analytic care that we have not
		attempted.
	\end{enumerate}
	
	\section{Conclusions}
	\label{sec:conclusions}
	
	A null mode of the bordered Faddeev--Jackiw matrix is not yet a gauge generator.
	We have shown that the dimension of the residual kernel decomposes canonically
	into two contributions with different meanings, the null directions tangent to
	the constraint surface, which carry no gauge content, and the first-class
	combinations, which are the only candidates; that the reduced pairing $\Gamma$
	is the
	primary--secondary block of the reduced constraint matrix and not the constraint
	bracket matrix, the latter being $M=B^{\top}f^{+}B$, defined relative to a
	choice of second-class splitting and intrinsic only through its compression
	$\Pi M\Pi$; that first-class character requires
	both $\Gamma w=0$ and $\Pi M\Pi w=0$; and that a first-class combination yields
	an actual gauge symmetry precisely when a two-step chain closes, which happens
	if and only if an explicit class in $\Rring/\Ical$ vanishes. A four-variable
	system shows that the standard halting criterion of the iteration can be met by
	a system with no gauge freedom at all, and that the chain criterion separates
	the two cases correctly.
	
	The mechanical realizations make the same point with masses, springs, rods and
	pulleys. For a singular second-order Lagrangian with constant kinetic matrix
	$K$, the first-order reformulation has $\Gamma=0$ and $M=0$, so every generated
	constraint is first class and the bordered kernel is as large as it can be, and
	yet gauge freedom holds or fails according to a condition,
	$w\!\cdot\!\nabla V\equiv0$, that neither $\Gamma$ nor $M$ can see. Two systems
	with the same rods, the same masses and the same singular kinetic structure, and
	differing only in where the springs are attached, land on opposite sides of that
	condition.
	
	A further consequence emerged in the course of that analysis, and it concerns
	the parameters themselves. The companion paper \cite{CLMCP2026} introduced the
	discipline of preserving parameter-dependent factors and justified it
	structurally, as a safeguard: clearing a factor can erase a locus in parameter
	space on which a rank, a kernel or a constraint structure changes, and nothing
	in the surviving expressions records that the locus was ever there. That
	argument establishes why the strata must be kept. It does not, by itself, say
	what they contain. The present paper answers that question in two concrete
	cases, and the answer is that they can contain qualitative dynamical
	information.
	
	In the four-mass system with springs at the corners the contraction is
	$w\!\cdot\!\nabla V=k\,(y_{2}+y_{4}-y_{1}-y_{3})$. Dividing by $k$ would leave
	the constraint $y_{1}+y_{3}=y_{2}+y_{4}$ and remove the stratum $k=0$, which is
	precisely where the alternating null mode passes from dynamically obstructed to
	satisfying the exact gauge condition; it is also where the finite spectrum
	$\{k/m,\,2k/m\}$ collapses and the equilibrium $y_{i}^{*}=a-mg/k$ escapes to
	infinity. The stratum $k=0$ therefore carries a gauge-character transition. In
	the first-order system of Section~\ref{sec:paramgauge}, where
	$\Gamma=\alpha^{2}+\beta$ is a function of the parameters alone, the locus
	$\Gamma=0$ carries a transition of a different kind: the functional
	indeterminacy of the general solution changes from zero to one arbitrary
	function of time, a conclusion obtained from the integrated equations of motion
	rather than from the inspection of a determinant. The two phenomena should not
	be merged. The first concerns the gauge character of a null direction, the
	second the dimension of the indeterminacy of the solutions; what they share is
	that both live on strata that a premature cancellation would have removed.
	Neither statement is conditional or hypothetical: both are demonstrated on the
	systems as given. The factorized constraint \eqref{eq:pendfactor} of
	Section~\ref{sec:mechanical} shows the same structure in a case where the
	vanishing stratum has not been examined, and is what makes the distinction
	between a branchwise and a global reduction worth stating.
	
	These two transitions are not equilibrium bifurcations. Under the criterion
	adopted in Remark~\ref{rem:parameters}, no equilibrium branch collides with
	another or splits, and where a finite spectrum reaches zero it does so only as
	the equilibrium leaves every bounded region. The loci correspond instead to a
	change of gauge character and to a transition of dynamical indeterminacy.
	Equilibrium bifurcations remain a legitimate further motivation for preserving
	the full parameter space; none is claimed here, and the present analysis
	neither establishes nor excludes their occurrence in other mechanical systems
	or other parameter families. What the examples do establish is that critical
	parameter values can control features of a constrained system that no
	equilibrium analysis would detect. The continuity with \cite{CLMCP2026} is
	therefore methodological: the parameter-preservation discipline introduced
	there on structural grounds is what makes these distinct dynamical strata
	visible here. Parameter-dependent simplification is best understood as a
	stratified operation. A factor may be cancelled on a specified regular
	stratum, and doing so is legitimate and often necessary; but the corresponding
	vanishing stratum must be classified before it is discarded, since the
	examples above show that it need not be empty of dynamical content.
	
	Within the regime actually covered here, that is, for two-step candidates with
	constant coefficients and under the hypotheses listed in
	Section~\ref{sec:scope}, the outcome is a decision procedure that costs one
	linear solve and one normal-form reduction per candidate, and that turns a
	structural heuristic into a theorem. It is not a general solution of the
	Faddeev--Jackiw gauge problem, and we do not claim one. What remains open is
	listed in Section~\ref{sec:scope}: chains of length $S\geq2$, functional
	coefficients, the structure on degeneracy strata, the relation between the
	halting surface and the final constraint manifold of Gotay--Nester, and the
	extension to field theory.
	
	\subsection*{Supplementary material}
	
	A self-contained Wolfram Language notebook is provided as ancillary material
	with the arXiv submission. The notebook reconstructs the symbolic calculations
	underlying the mechanical examples of this paper, starting from the mechanical
	Lagrangians and proceeding through the kinetic matrix, its kernel and
	pseudoinverse, the extended potential, the first-order Faddeev--Jackiw
	Lagrangian, the presymplectic matrix, the primary constraints, the bordered
	quantities, the constraint matrix, the bordered kernel, the chain equations,
	and the Gauge criterion and generator. No external package or previous
	software implementation is required.
	
	The notebook follows the construction developed in
	Secs.~\ref{sec:conventions}--\ref{sec:chains} and applies it to the mechanical
	realizations of the gauge criterion in Sec.~\ref{sec:mechanical}. The four
	mechanical systems are treated individually in
	Secs.~\ref{sec:pulley3}--\ref{sec:square-control}. The parameter-stratification
	tests of the worked examples are also reproduced, including the
	parameter-controlled gauge transition of Sec.~\ref{sec:paramgauge}.
	
	Parameter-dependent loci are stratified before parameter-dependent
	simplifications are made. No parameter is cancelled from a relation before
	the corresponding vanishing stratum has been classified. The regular loci of
	the kinetic matrices are determined symbolically before their kernels or
	pseudoinverses are used, and the first-order character of the Lagrangians is
	verified by a total-degree test in the extended velocities.
	
	For each mechanical system, the notebook retains the intermediate symbolic
	objects and verifies the identities used in the manuscript. The resulting
	checks are compared directly with the theoretical structure developed in
	Secs.~\ref{sec:kernel}--\ref{sec:firstclass} and with the chain and generator
	criteria of Sec.~\ref{sec:chains}. The parameter-controlled system of
	Sec.~\ref{sec:paramgauge} is treated separately, including the singular point
	$\alpha=\beta=0$, where the bordering becomes reducible and the Gauge
	mechanism changes from a two-step chain on $\alpha\neq0$ to a strong Gauge
	direction at the singular point. The final verification collects the
	individual symbolic checks into a single audit.
	
	\subsection*{Data availability}
	
	No experimental or observational datasets were generated or analysed during
	this study. The symbolic verification notebook described in the Supplementary
	material is provided with the article as ancillary material. The notebook is
	self-contained, contains the complete Wolfram Language source used for the
	symbolic verification, and can be executed directly in Mathematica.
	
	\section*{Acknowledgements}
	E. Chan-L\'opez acknowledges support from SECIHTI through the
	``Estancias Posdoctorales por M\'exico'' program (CVU 422090).

	\section*{Ethics declarations}
	\subsection*{Conflict of interest}
	The author declares no conflict of interest.
	

\end{document}